\documentclass[11pt]{article}
\usepackage[margin=1in]{geometry}
\usepackage{adjustbox}

\usepackage[
    backend=biber,
    style=numeric-comp,
    sortcites,
    sorting=none,
    giveninits=true,
    natbib,
    hyperref=false,
    maxbibnames=99,
    doi=false,isbn=false,url=false,eprint=false
]{biblatex}

\usepackage[colorinlistoftodos]{todonotes}

\usepackage{xcolor}
\definecolor{darkgreen}{rgb}{0.0, 0.5, 0.0}

\makeatletter
\def\@fnsymbol#1{\ensuremath{\ifcase#1\or *\or 1 \or 2 \or 3 \or 4 \or 5 \else\@ctrerr\fi}}
\makeatother

\title{Testing for Outliers with Conformal p-values}
\author{Stephen Bates\footnote{Authors listed alphabetically.} \thanks{Departments of Statistics and of EECS, UC Berkeley.}, Emmanuel Cand{\`e}s\thanks{Departments of Statistics and of Mathematics, Stanford University.}, Lihua Lei\thanks{Department of Statistics, Stanford University.}, Yaniv Romano\thanks{Departments of Electrical Engineering and of Computer Science, Technion---Israel Institute of Technology.}, Matteo Sesia\thanks{Department of Data Sciences and Operations, University of Southern California.}}
\date{\today}

\usepackage{rotating}
\usepackage{booktabs, multirow, ltablex}
\usepackage{longtable}
\usepackage[font={small}]{caption}

\usepackage{array}
\usepackage{amsmath, amsthm, amssymb}
\usepackage{subfigure}

\newtheorem{thm}{Theorem} 
\newtheorem{prop}{Proposition}
\newtheorem{lemma}{Lemma}
\newtheorem{cor}{Corollary}
\newtheorem{defn}{Definition}
\newtheorem{remark}{Remark}

\usepackage{mathtools}

\newcommand{\defeq}{\vcentcolon=}

\usepackage{algorithm}
\usepackage{algpseudocode}

\usepackage{paralist,enumerate}

\def\I{\mathbb{I}}
\def\betadist{\textsc{Beta}}

\newcommand{\Unif}{\textnormal{Unif}}
\newcommand{\Dcal}{\mathcal{D}^\textnormal{cal}}
\newcommand{\Var}{\textnormal{Var}}

\def\R{\mathbb{R}}
\def\E{\mathbb{E}}

\def\C{\mathcal{C}}
\def\Chat{\hat{\mathcal{C}}}

\def\marg{\textnormal{(marg)}}
\def\cond{\textnormal{(ccv)}}
\def\dkwm{\textnormal{d}}
\def\simes{\textnormal{s}}
\def\asym{\textnormal{a}}

\renewcommand{\P}{\mathbb{P}}
\newcommand{\D}{\mathcal{D}}
\newcommand{\event}{\mathcal{E}}
\newcommand{\lb}{\left(}
\newcommand{\rb}{\right)}
\newcommand{\T}{\mathcal{T}}

\providecommand{\keywords}[1]
{
  \small	
  \textbf{\textit{Keywords---}} #1
}

\begin{document}

\maketitle

\begin{abstract}
This paper studies the construction of p-values for nonparametric outlier detection, taking a multiple-testing perspective. The goal is to test whether new independent samples belong to the same distribution as a reference data set or are outliers. We propose a solution based on conformal inference, a broadly applicable framework which yields p-values that are marginally valid but mutually dependent for different test points. We prove these p-values are positively dependent and enable exact false discovery rate control, although in a relatively weak marginal sense. We then introduce a new method to compute p-values that are both valid conditionally on the training data and independent of each other for different test points; this paves the way to stronger type-I error guarantees. Our results depart from classical conformal inference as we leverage concentration inequalities rather than combinatorial arguments to establish our finite-sample guarantees. Furthermore, our techniques also yield a uniform confidence bound for the false positive rate of any outlier detection algorithm, as a function of the threshold applied to its raw statistics. Finally, the relevance of our results is demonstrated by numerical experiments on real and simulated data.
\end{abstract}

\keywords{Conformal inference, out-of-distribution testing, false discovery rate, positive dependence.}

\section{Introduction}

\subsection{Problem statement and motivation}

We consider an outlier detection problem in which one observes a data set $\mathcal{D}=\{X_i\}_{i=1}^{2n}$ containing $2n$ independent and identically distributed points $X_i \in \mathbb{R}^{d}$ drawn from an unknown distribution $P_{X}$ {(which may be continuous, discrete, or mixed)}.
The goal is to test which among a new set of $n_{\text{test}} \geq 1$ independent observations $\mathcal{D}^{\mathrm{test}} = \{X_{2n+i}\}_{i=1}^{n_{\text{test}}}$ are \emph{outliers}, in the sense that they were not drawn from the same distribution $P_{X}$. By contrast, we refer to points drawn from $P_X$ as \emph{inliers}.
This problem has applications in many domains, including medical diagnostics~\cite{tarassenko1995novelty}, spotting frauds or intrusions~\cite{patcha2007overview}, forensic analysis~\cite{Fortunato2020}, monitoring engineering systems for failures~\cite{Tarassenko2009},
and \emph{out-of-distribution} detection in machine learning~\cite{hendrycks2016baseline, liang2017enhancing, lee2018simple, lee2018training}.
A variety of machine-learning tools have been developed to address this classification task, which is sometimes referred to as \emph{one-class classification}~\cite{moya1993one,pimentel2014review} because the data in $\mathcal{D}$ do not contain any outliers. 
However, such algorithms are often complex and their outputs are not directly covered by any precise statistical guarantees. Fortunately, {conformal} inference~\cite{vovk1999machine, vovk2005algorithmic} allows one to practically convert the output {of any} one-class classifier ({if} it is invariant to the ordering of the training observations) into a provably valid p-value for the null hypothesis $\mathcal{H}_{0,i}: X_{i} \sim P_X$, for any $X_i \in \mathcal{D}^{\mathrm{test}}$. 

In many applications, the number of outlier tests, $n_{\text{test}}$, is large and, therefore, it may be necessary to account for multiple comparisons to avoid making an excessive number of false discoveries.
A meaningful error rate in this setting is the false discovery rate (FDR)~\cite{benjamini1995controlling}: the expected proportion of true inliers among the test points reported as outliers. For example, if a particular financial transaction is labeled by an automated system as likely to be fraudulent (i.e., unusual, or out-of-distribution compared to a data set of normal transactions), someone may then need to review it manually, and possibly contact the involved customer.
Since these follow-up procedures have a cost, controlling the FDR may be a sensible solution to ensure resources are allocated efficiently.
From a statistical perspective, multiple testing in this setting requires some care because classical conformal p-values corresponding to different values of $i>2n$ are independent of each other only conditional on $\mathcal{D}$, although they are valid only marginally over $\mathcal{D}$.
This situation is delicate because FDR control typically requires p-values that either are mutually independent or follow certain patterns of dependence~\cite{benjamini2001control, clarke2009robustness}.
Similarly, global testing (i.e., aggregating evidence from multiple observations to test weaker batch-level hypotheses) may also require independent p-values.
This paper addresses the above issues by carefully studying the theoretical properties of some standard multiple testing procedures applied to conformal p-values, and by developing new methods to compute p-values with stronger validity properties.

The conformal inference methods studied in this paper are statistical wrappers for one-class classifiers. The latter are algorithms trained on data clean of any outliers to compute a score function $\hat{s} : \mathbb{R}^d \to \mathbb{R}$ assigning a scalar value to any future data point, so that smaller (for example) values of $\hat{s}(X)$ provide evidence that $X$ may be an outlier. By design, the classifier attempts to construct scores that separate outliers from inliers effectively, by learning from the data what inliers typically look like, and it may be based on sophisticated black-box models to maximize power. While often effective in practice, these machine-learning algorithms have the drawback of not offering any clear guarantees about the quality of their output. For example, they do not directly provide a null distribution for the classification scores $\hat{s}$ evaluated on true inliers, or any particular threshold to limit the rate of false positives. 
This is where conformal inference comes to help.
After training $\hat{s}$ on a subset of the observations in $\mathcal{D}$, namely those in $\mathcal{D}^{\text{train}} = \{X_1, \ldots, X_n\}$, the scores are evaluated on the remaining $n$ hold-out samples in $\mathcal{D}^{\text{cal}} = \{X_{n+1}, \ldots, X_{2n}\}$.
(Note that $\mathcal{D}^{\text{train}}$ and $\mathcal{D}^{\text{cal}}$ do not need to contain the same number of observations, although the current choice simplifies the notation without loss of generality).
Let us assume, for simplicity, that $\hat{s}(X)$ has a continuous distribution if $X \sim P_X$ is independent of the data used to train $\hat{s}$, although this assumption could be relaxed at the cost of some additional technical details.
Then, define $F$ as the cumulative distribution function (CDF) of $\hat{s}(X)$.
If we knew $F$, we could utilize $F(\hat{s}(X_i))$ as an exact p-value for the null hypothesis $\mathcal{H}_{0,i}: X_i \sim P_{X}$, for any $X_i \in \mathcal{D}^{\mathrm{test}}$, in the sense that $F(\hat{s}(X_i))$ would be uniformly distributed if $\mathcal{H}_{0,i}$ is true.
In practice, however, we do not have direct access to $F$ because $P_X$ is unknown and the machine-learning algorithm upon which $\hat{s}$ depends is assumed to be a black-box.
Instead, we can evaluate the empirical CDF of $\hat{s}(X_i)$ for all $X_i \in \mathcal{D}^{\text{cal}}$, which we denote as $\hat{F}_n$.
In the following, we will discuss how to construct provably valid conformal p-values for a future observation $X_{2n+1}$ by evaluating
\begin{equation} \label{eq:pval_form}
    \hat{u}(X_{2n+1}) = \left(g \circ \hat{F}_n \circ \hat{s} \right)(X_{2n+1}),
\end{equation}
where $g$ is a suitable \emph{adjustment function}, and the symbol $\circ$ denotes a composition; i.e., $(f \circ g)(x) = f(g(x))$.
Note that, hereafter, we will treat the observations in $\mathcal{D}^{\text{train}}$ as fixed and focus on the randomness in the calibration ($\mathcal{D}^{\text{cal}}$) and test ($\mathcal{D}^{\mathrm{test}}$) data, upon which conformal inferences are generally based.

\subsection{Preview of contributions}
\label{subsec:preview}

In Section~\ref{sec:marginal_pvals}, we will focus on the classical conformal inference methods, which produce \emph{marginally super-uniform} (conservative) p-values $\hat{u}^\marg(X_{2n+1})$ satisfying
\begin{equation}  \label{eq:marg_validity_def}
    \mathbb{P}\left[ \hat{u}^\marg(X_{2n+1}) \le t \right] \le t,
\end{equation}
for any $t \in (0,1)$, whenever $X_{2n+1}$ is an inlier.
We say these p-values are marginally valid because they depend on the calibration data in $\mathcal{D}^{\text{cal}}$, and both $\mathcal{D}^{\text{cal}}$ and $X_{2n+1}$ are random in~\eqref{eq:marg_validity_def}.
In particular, the classical $\hat{u}^\marg$ is computed by applying the adjustment function $g^\marg(x) =  (nx + 1) / (n + 1)$ to~\eqref{eq:pval_form}, i.e., 
\begin{equation} \label{eq:marginal-pvals-def}
    \hat{u}^\marg(x) = \frac{1 + |\{i \in \mathcal{D}^{\text{cal}} : \hat{s}(X_i) \le \hat{s}(x)\}|}{n+1}.
\end{equation}
Note that~\eqref{eq:marg_validity_def} {is implied by}~\eqref{eq:marginal-pvals-def} because {when $\hat{s}(X)$ follows a continuous distribution,} $\hat{u}^\marg(X)$ is uniformly distributed on $\{1/(n+1), 2/(n+1), \dots, 1\}$ if $X \sim P_X$ independently of the data in {$\mathcal{D}^{\text{train}}$~\cite{vovk1999machine, vovk2005algorithmic}.} (If $\hat{s}(X)$ is not continuous, one can still {verify} that $\hat{u}^\marg(X)$ is super-uniform in distribution.)
However, this is not necessarily true if one conditions on $\mathcal{D} = \mathcal{D}^{\text{train}} \cup \mathcal{D}^{\text{cal}}$, in which case $\hat{u}^\marg(X)$ may become anti-conservative due to random fluctuations in the distribution of scores within $\mathcal{D}^{\text{cal}}$.
Intuitively, this means the marginal p-values in~\eqref{eq:marginal-pvals-def} are only valid {\it on average} if data in $\mathcal{D}^{\text{cal}}$ are treated as random. Unfortunately, this guarantee may be too weak to be satisfactory for a practitioner who wants to compute p-values for a large number of test points but is constrained to working with a single calibration data set. Indeed, the numerical experiments presented in Section~\ref{sec:exp_simulated} will show that inferences based on marginal conformal p-values may be systematically invalid for a large fraction of practitioners working with ``unlucky'' calibration data sets.

Furthermore, marginal p-values corresponding to different test points,  $\{ \hat{u}^\marg(X)\}_{X \in \mathcal{D}^{\mathrm{test}}}$, are not mutually independent because they are all affected by $\mathcal{D}^{\text{cal}}$; see Figure~\ref{fig:dependence_viz} for a visualization of this dependence.
This should be taken into account when adjusting for multiplicity in outlier detection applications because some common testing procedures are not generally valid for dependent p-values.
For example, we will prove in Section~\ref{sec:marginal_pvals} that the dependence among marginal p-values invalidates Fisher's combination test~\cite{fisher1925statistical} for the global null that there are no outliers in $\mathcal{D}^{\mathrm{test}}$, even if the calibration data in $\mathcal{D}^{\text{cal}}$ are treated as random, although this can be easily fixed by suitably adjusting the critical value. By contrast, we can prove the dependence between conformal p-values does not break the {\it average} FDR control of the Benjamini-Hochberg (BH) procedure~\cite{benjamini1995controlling}, even if the latter is applied with Storey's correction~\cite{storey2004strong}. The behaviours of additional multiple testing procedures, such as the harmonic mean~\cite{wilson2019harmonic}, Simes method~\cite{simes1986improved}, and Stouffer's method~\cite{stouffer1949american}, applied to conformal p-values will be investigated empirically in Section~\ref{sec:exp}.

\begin{figure}[!tb]
    \centering
    \includegraphics[width = .9\textwidth, trim = 0 950 750 0, clip]{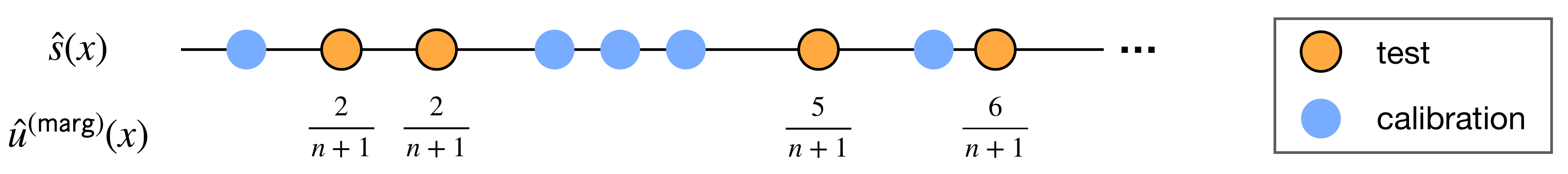}
    \caption{Visualization of the joint distribution of the conformal p-values. The distribution of $\hat{s}(x)$ is the same for calibration and inlier test points. The conformal p-value for each test point is the number of calibration points to its left, divided by the total number of calibration points plus one, as in \eqref{eq:marginal-pvals-def}.}
    \label{fig:dependence_viz}
\end{figure}

In any case, regardless of whether the mutual dependence among marginal p-values theoretically invalidates the average inferences of a particular multiple-testing procedure, one may sometimes be interested in obtaining stronger guarantees conditional on the calibration data.
Consider for instance the following prototypical scenario. 
A researcher, or a company, acquires an expensive data set $\mathcal{D}$ containing clean examples of some variable $X$ of interest, and wishes to leverage that information to construct a system {to} detect outliers in future test points, while avoiding an excess of false positives. Assuming the stakes in this application are sufficiently high, the researcher may need clear statistical guarantees about the output of such procedure (as opposed to blindly trusting a black-box model), and thus decides to employ conformal inference. Unfortunately, the marginal validity property in~\eqref{eq:marg_validity_def} tells us very little about how this outlier detection system may perform in the future for {\em this particular researcher relying on this particular data set $\mathcal{D}$}. Instead, marginal validity suggests the system will work {\em on average} for different researchers starting from different data sets; of course, that may not feel fully satisfactory for any one of them.

Therefore, we will construct in Section~\ref{sec:simult_conformal} conformal p-values satisfying a stronger property, which we call \emph{calibration-conditional validity} (CCV). Formally, the novel p-values $\hat{u}^\cond(x)$ will satisfy
\begin{equation}  \label{eq:cond_validity_def}
    \mathbb{P}\left[ \mathbb{P}\left[ \hat{u}^\cond(X_{2n+1}) \leq t \mid \mathcal{D} \right] \leq t \text{\; for all \,} t \in (0,1) \right] \geq 1-\delta,
\end{equation}
if $X_{2n+1} \sim P_X$, for any value of $\delta \in (0,1)$ pre-specified by the user.
The crucial difference between~\eqref{eq:cond_validity_def} and~\eqref{eq:marg_validity_def} is that the latter intuitively guarantees the p-values are valid for at least a fraction $1 - \delta$ of researchers; this can give a precise measure of confidence {to each one of them.}
Furthermore, calibration-conditional p-values have the advantage of making multiple testing straightforward. In fact, these p-values are still trivially independent of one another conditional on the calibration data, so their high-probability guarantee of validity will immediately extend to the output of any downstream multiple-testing procedure that assumes independence.

While most of this paper focuses on the validity of conformal p-values from a multiple-testing perspective, we will see in Section~\ref{sec:extensions} that our high-probability results can also be utilized to construct a uniform upper confidence bound for the false positive rate of any machine-learning algorithm for outlier detection, as a function of the threshold applied to its raw output scores. This may help practitioners interpret the output of black-box methods directly, without necessarily operating in terms of p-values. (However, as statisticians, we prefer the p-value approach because it is more versatile.)
Furthermore, our results can be easily leveraged to obtain predictive sets with stronger coverage guarantees compared to existing conformal methods.

Finally, in Section~\ref{sec:exp}, we will compare the performance of marginal and calibration-conditional conformal p-values on simulated as well as real data, in combination with different multiple testing procedures. These numerical experiments will provide an empirical confirmation of our theoretical results, and also highlight how stronger guarantees sometimes come at the cost of lower power.

\subsection{Related work}

The outlier detection problem considered in this paper is fully non-parametric, in the sense that we leverage the information contained in an external clean data set, and nothing else, to infer whether a future test point may be an outlier. 
This is in contrast with the more classical problem of multivariate outlier detection within a single data set, leveraging modeling assumptions rather than clean external samples~\cite{wilks1963multivariate,hawkins1980identification,riani2009finding,cerioli2010multivariate}.
A wealth of data mining and machine-learning methods have been developed to address our non-parametric task~\cite{khan2014one,agrawal2015survey,aggarwal2015outlier,sabokrou2018adversarially,chalapathy2019deep}; these do not provide precise finite-sample guarantees on their own, but we can leverage them to compute scoring functions that powerfully separate outliers from inliers.

Our paper is based on conformal inference~\cite{vovk1999machine, vovk2005algorithmic}, which has been applied before in the context of outlier detection~\cite{laxhammar2015inductive,smith2015conformal,ishimtsev2017conformal,guan2019prediction,cai2020real,haroush2021statistical}. 
However, previous works did not study the implications of marginal p-values on the validity of multiple outlier testing procedures, nor did they seek the conditional guarantees obtained here. Another line of work applied conformal inference to test the global null for streaming data \citep{vovk2003, fedorova2012plug, vovk2019testing, vovk2020testing, vovk2021retrain}. However, the guarantee no longer holds in the offline setting or beyond the global null.
The most closely related work is that of~\cite{vovk2012conditional}, which extends conformal inference to provide a form of calibration-conditional coverage. That paper focused explicitly on the prediction setting  rather than on outlier detection, but is also directly relevant in our context, as discussed in Section~\ref{sec:warmup-fpr}.
The main difference is that our novel high-probability bounds in Section~\ref{sec:simult_conformal} hold simultaneously for all possible coverage levels (in the language of~\cite{vovk2012conditional}) not just for a pre-specified one---this feature being necessary to obtain conditionally valid p-values for multiple outlier testing.

Other works on conformal inference focused on different types of conditional coverage.
For example,~\cite{barber2019limits} studied the difficulty of computing valid conformal predictions (in a supervised setting) conditional on the features of a new test point, while we are interested in conditioning on the calibration data (in an outlier detection setting).
Other works have focused on seeking approximate feature-conditional coverage in multi-class classification~\cite{hechtlinger2018cautious, romano2020classification, cauchois2020knowing,angelopoulos2020sets} or in regression problems~\cite{romano2019conformalized, izbicki2019flexible, chernozhukov2019distributional, kivaranovic2019adaptive, gupta2020nested}.
This paper is orthogonal, in the sense that our results could be applied to strengthen their coverage guarantees by conditioning on the calibration data.
It should be noted that, although conformal inference can be based on different data hold-out strategies~\cite{vovk2015cross, barber2019predictive, kim2020predictive}, our paper focuses on {sample} splitting~\cite{papadopoulos2002inductive, lei2013conformal}. 
The latter has the advantage of being the most computationally efficient option, and is necessary for us in theory because our high-probability bounds require the independence of the data points in addition to their exchangeability.

Further, the problem we consider is related to classical two-sample testing~\cite{wilcoxon1992individual}, although we take a different perspective. Two-sample testing {compares} two data sets to determine whether they were sampled from the same distribution, while our goal is to {contrast} many independent test points (or batches thereof) to the same reference set accounting for {multiplicity}. In any case, several recent works have explored the use of machine-learning and data hold-out methods for two-sample testing~{\cite{friedman2004multivariate,lopez2017revisiting,kuchibhotla2020exchangeability,hu2020distributionfree,kim2021classification}}, which reinforces the connection with our work.

Finally, the duality between hypothesis testing and confidence intervals connects our conditionally calibrated p-values to the classical statistical topic of \emph{tolerance regions}, which goes back to Wilks~\cite{wilks1941, wilks1942}, Wald~\cite{wald1943}, and Tukey~\cite{tukey1947}. {See~\cite{krishnamoorthy2009statistical} for a overview of the subject, \cite{vovk2012conditional} for a discussion of their connection with conformal inference, and~\cite{Park2020PAC,bates2021distributionfree} for modern examples using tolerance regions for predictive inference with neural networks.} (Tolerance regions are predictive sets with a high-probability guarantee to contain the desired fraction of the population. For example, one can generate a tolerance region guaranteed to contain at least 80\% of the population with probability 99\%.{)}
{The construction of predictive intervals with (asymptotic) conditional validity in the aforementioned sense was also recently studied in~\cite{zhang2020bootstrap} with bootstrap rather than conformal inference methods.}



\section{Marginal conformal inference for outlier detection}
\label{sec:marginal_pvals}

Before turning to calibration-conditional inferences, we carefully study the marginal validity of multiple tests based on split-conformal outlier detection p-values. The conformal p-values defined in~\eqref{eq:marginal-pvals-def} are marginally valid for the hypothesis that a single test point follows the distribution $P_X$, see \eqref{eq:marg_validity_def}, but they are not independent of each other when considering multiple test points. Consequently, we show they cannot be naively used to test a global null hypothesis that no points in a particular test set are outliers, with Fisher's combination test~\cite{fisher1925statistical} for example.
The failure of Fisher's test is caused by the particular dependence induced by the shared calibration data set, although other procedures turn out to be robust to such dependence. In particular, we then prove conformal p-values are \emph{positive regression dependent on a subset} (PRDS), which combined with the results of~\cite{benjamini2001control}, implies the BH procedure will control the FDR.

\subsection{A negative result: global testing with conformal p-values can fail} \label{sec:fisher}

Fisher's combination test~\citep{fisher1925statistical} is a widely-used method to test the global null, in our case
$$
H_{0}: X_{2n+1}, \ldots, X_{2n + m}\stackrel{\mathrm{i.i.d.}}{\sim} P_{X}.
$$
The idea is to aggregate the evidence from the individual tests, as follows. Given a p-value $p_i$ for each null hypothesis $i$, Fisher's test rejects the global null at level $\alpha$ if
\[-2 \sum_{i=1}^{m} \log p_{i} \ge \chi^2(2m; 1 - \alpha),\]
where $\chi^2(2m; 1 - \alpha)$ is the $(1 - \alpha)$-th quantile of the chi-square distribution with $2m$ degrees of freedom.
This test is valid if the p-values stochastically dominate $\Unif([0, 1])$ and are independent of each other. However, we prove in the following lemma that the standard (marginal) conformal p-values are positively correlated under arbitrary transformations, suggesting an inflation of the variance of the combination statistics.

\begin{lemma}\label{lem:pairwise_correlation}
Assume that $\hat{s}(X)$ is continuously distributed. Then, for any finite-valued function $G: [0, 1]\mapsto \R$, and for any pair of nulls $(i, j)$,
\[\mathrm{Cor}\left[ G(\hat{u}^\marg(X_{2n+i})), G(\hat{u}^\marg(X_{2n+j})) \right] = \frac{1}{n+2}.\]
\end{lemma}

Motivated by Lemma~\ref{lem:pairwise_correlation} (see Appendix~\ref{app:correlation-structure} for a detailed discussion), we obtain the following result which shows Fisher's combination test becomes invalid when applied to marginal conformal p-values. In particular, we characterize its type-I error in the asymptotic regime where $|\mathcal{D}^{\textnormal{test}}|$ is proportional to $|\mathcal{D}^{\textnormal{cal}}|$.
\begin{thm}[Type-I error of Fisher's combination test]\label{thm:fisher}
Assume that $\hat{s}(X)$ is continuously distributed. Then, under the global null, if $m = \lfloor\gamma n\rfloor$ for some $\gamma \in (0, \infty)$, as $n$ tends to infinity, 
\begin{equation*}
\P\left[ -2 \sum_{i=1}^{m} \log \left[ \hat{u}^\marg(X_{2n+i}) \right] \ge \chi^2(2m; 1 - \alpha) \right] \rightarrow  \bar{\Phi}\left(\frac{z_{1 - \alpha}}{\sqrt{1 + \gamma}}\right),
\end{equation*}
where $z_{1 - \alpha}$ and $\bar{\Phi}$ denote the $(1 - \alpha)$-th quantile and survival function of the standard normal distribution, respectively. Furthermore, under the same asymptotic regime, for $W\sim N(0, 1)$,
\begin{equation}\label{eq:conditional_typeI}
\P\left[ -2 \sum_{i=1}^{m} \log \left[\hat{u}^\marg(X_{2n+i})\right] \ge \chi^2(2m; 1 - \alpha)\mid \mathcal{D} \right] \stackrel{d}{\rightarrow}\bar{\Phi}(z_{1 - \alpha} + \sqrt{\gamma} W).
\end{equation}
\end{thm}
Note that the above asymptotic limits are independent of the distribution of $\hat{s}(X)$. In Appendix~\ref{app:proofs}, we prove that Theorem~\ref{thm:fisher} holds for a broad class of combination tests based on $\sum_{i=1}^{n}G(\hat{u}^\marg(X_{2n+i}))$, provided that $G(U)$ has finite moments for $U\sim \Unif([0, 1])$; Fisher's combination test is a special case with $G(u) = -2\log u$ and $G(U)\sim \chi^2(2)$. 

 Since $\gamma > 0$, the marginal type-I error is always larger than $\alpha$ whenever $\alpha < 0.5$. For illustration, consider $\alpha = 5\%$. When $\gamma = 3$, the marginal type-I error is as large as $20.5\%$; when $\gamma \rightarrow\infty$, the marginal type-I error is approaching $50\%$. Similarly, by \eqref{eq:conditional_typeI}, the $90$-th percentile of the conditional type-I error converges to the $90$-th percentile of $\bar{\Phi}(z_{1 - q} + \sqrt{\gamma} W)$, which is $\bar{\Phi}(z_{0.95} + \sqrt{\gamma}z_{0.1})$. When $\gamma = 3$, the limit is $71.7\%$; when $\gamma \rightarrow \infty$, the limit is approaching $100\%$. This demonstrates the substantial adverse effect of dependence among marginal conformal p-values for Fisher's combination test.

Corrections of Fisher's combination test are possible for some dependence structures. 
By Lemma~\ref{lem:pairwise_correlation}, the variance of the combination statistic is inflated by a factor $(1 + \gamma)$ compared to that of the $\chi^2(2m; 1 - \alpha)$ distribution (see Appendix~\ref{app:correlation-structure} for details). This yields an intuitive correction which divides the combination statistic by $\sqrt{1 + \gamma}$. Surprisingly, this correction is asymptotically too conservative for marginal conformal p-values. 
We prove in Appendix~\ref{app:global_testing} (Theorem~\ref{thm:combination_test}) that a valid correction rejects the global null if 
\begin{equation}\label{eq:fisher_corrected}
\frac{-2 \sum_{i=1}^{m} \log \left[\hat{u}^\marg(X_{2n+i})\right] + 2(\sqrt{1 + \gamma} - 1)m}{\sqrt{1 + \gamma}}\ge \chi^2(2m; 1 - \alpha).
\end{equation}
In Appendix~\ref{app:global_testing}, we also confirm the validity of \eqref{eq:fisher_corrected} via Monte-Carlo simulations and show this is asymptotically equivalent to the correction proposed by~\cite{brown1975400, kost2002combining} to address p-value dependence in more general contexts.




\subsection{A positive result: conformal p-values are positively dependent} \label{sec:positive-fdr}

Certain multiple testing methods, such as the BH procedure, are known to be robust to a particular type of mutual p-value dependence called \emph{positive regression dependent on a subset} (PRDS)~\cite{benjamini2001control}. 
\begin{defn} \label{def:PRDS}
A random vector $X = (X_1,\dots,X_m)$ is PRDS on a set $I_0\subset \{1, \ldots, m\}$ if for any $i \in I_0$ and any increasing set $A$, the probability $\P[X \in A \mid X_i = x]$ is increasing in $x$. 
\end{defn}
In the multiple testing literature, $X$ is often said to be PRDS if it is PRDS on the set of nulls. Above, for vectors $a$ and $b$ of equal dimension, we say $a \succeq b$ if every coordinate of $a$ is no smaller than the corresponding coordinate of $b$, and a set $A \subset \mathbb{R}^m$ is \emph{increasing} if $a \in A$ and $b \succeq a$ implies $b \in A$.
The PRDS property is a demanding form of positive dependence which can be interpreted, loosely speaking, as saying all pairwise correlations are positive. 
In view of the definition of marginal p-values in~\eqref{eq:marginal-pvals-def} and the result in Lemma~\ref{lem:pairwise_correlation}, it should be intuitive that larger scores in the calibration set make the p-values for all test points simultaneously smaller, and vice-versa. This idea is formalized by the following result proving marginal conformal p-values are PRDS.

\begin{thm}[Conformal p-values are PRDS] \label{thm:prds}
{Assume that $\hat{s}(X)$ is continuously distributed.} Consider $m$ test points $X_{2n+1},\dots,X_{2n+m}$ such that the inliers are jointly independent of each other and of the data in $\mathcal{D}$. Then, the marginal conformal p-values $(\hat{u}^\marg(X_{2n+1}),\dots,\hat{u}^\marg(X_{2n+m}))$ are PRDS on the set of inliers.
\end{thm}
{When $\hat{s}(X)$ is not continuous, we can also prove the PRDS property by modifying the definition \eqref{eq:marginal-pvals-def} of marginal conformal p-values; see Appendix \ref{app:PRDS} for details. }It follows from Theorem~\ref{thm:prds} that marginal conformal p-values can be used to control the FDR with the BH procedure for the null hypotheses
\begin{equation*}
    H_{0,i} : X_i \sim P_X, \qquad i \in \{2n+1,\dots,2n+m\}.
\end{equation*}
\begin{cor}[Benjamini and Yekutieli~\cite{benjamini2001control}]\label{cor:BH_FDR}
In the setting of Theorem~\ref{thm:prds}, the BH procedure applied at level $\alpha \in (0,1)$ to $(\hat{u}^\marg(X_{2n+1}),\dots,\hat{u}^\marg(X_{2n+m}))$ controls the FDR at level $\pi_0\alpha$, where $\pi_0$ is the proportion of true nulls. That is,
\begin{equation}
    \E\left[\frac{|\mathcal{R} \cap \mathcal{H}_0|}{\max\{1, |\mathcal{R}|\} } \right] \le \pi_0\alpha\le \alpha,
\end{equation}
where $\mathcal{H}_0 = \{i : H_{0,i} \text{ holds}\} \subseteq \{2n+1,\dots,2n+m\}$ is the subset of true inliers in the test set, and $\mathcal{R} \subseteq \{2n+1,\dots,2n+m\}$ is the subset of test points reported as likely outliers. 
\end{cor}

\begin{remark}
It turns out that the BH procedure applied to the marginal conformal p-values is equivalent to the semi-supervised BH procedure proposed by \cite{mary2021semi} (posted on arXiv two months after our paper), which was first studied by \cite{weinstein2017power} and later generalized by \cite{yang2021bonus} and \cite{rava2021burden}. These works employ a martingale-based technique to prove the FDR control without relying on the PRDS property. Theorem 3.1 in \cite{mary2021semi} also proves a lower bound showing that the FDR is almost exactly $\pi_0 \alpha$. 
\end{remark}

This proves the FDR can be controlled, although only on average over the calibration data because the above expectation is taken over both $\mathcal{D}$ and the future test points.
While such marginal guarantee may be satisfactory for someone carrying out several independent applications, individual practitioners committed to a single calibration data set may prefer stronger results. 

\subsection{A positive result: Storey's correction does not break FDR control}
When the proportion of nulls is much smaller than $1$, as it may be the case in many out-of-distribution detection problems, the BH procedure is conservative, as shown in Corollary~\ref{cor:BH_FDR}. If $\pi_0$ is known, a simple remedy is to replace the target FDR level with $\alpha / \pi_0$. However, $\pi_0$ is rarely known in practice and hence it needs to be estimated. Given p-values $p_i$ for all null hypotheses, it was proposed by Storey et al.~in~\cite{storey2002direct,storey2004strong} to estimate $\pi_0$ as
\[\hat{\pi}_0 = \frac{1 + \sum_{i=1}^{m}I(p_i > \lambda)}{m(1 - \lambda)},\]
and then to apply the BH procedure at level $\alpha / \hat{\pi}_0$; see Appendix~\ref{app:storey_BH} for details.
If the null p-values are super-uniform {in the sense of~\eqref{eq:marg_validity_def}}, mutually independent, and independent of the non-null
p-values, this provably controls the FDR in finite samples~\cite{storey2004strong}.
However, unlike in its standard version, the BH procedure with Storey's correction may fail to control the FDR if the p-values are PRDS; see Section 6.3 of~\cite{benjamini2006adaptive}.

Surprisingly, we show below that the positive correlation (Lemma~\ref{lem:pairwise_correlation}) among the marginal conformal p-values does not break the FDR control at all. The proof of Theorem~\ref{thm:storey_BH} rests on a novel FDR bound for the BH procedure with Storey's correction applied to any type of super-uniform p-values that are PRDS and almost-surely bounded from below by a constant; see Theorem~\ref{thm:storey_BH_generic} in Appendix~\ref{app:storey_BH}. Note that this result is not limited to conformal p-values and may also be useful for other multiple testing problems, such as those involving permutation p-values.

\begin{thm}[Storey's BH with conformal p-values controls the FDR]\label{thm:storey_BH}
Set $\lambda = K / (n+1)$ for any integer $K$. Assume $\hat{s}(X)$ is continuously distributed. In the setting of Corollary~\ref{cor:BH_FDR}, the BH procedure with Storey's correction applied at level $\alpha \in (0,1)$ to the marginal p-values $(\hat{u}^\marg(X_{2n+1}),\dots,\hat{u}^\marg(X_{2n+m}))$ controls the FDR at level $\alpha$. That is,
\begin{equation}
    \E\left[\frac{|\mathcal{R} \cap \mathcal{H}_0|}{\max\{1, |\mathcal{R}|\} } \right] \le \alpha.
\end{equation}
\end{thm}

\section{Calibration-conditional conformal p-values} \label{sec:simult_conformal}

\subsection{Warm up: analyzing the false positive rate} \label{sec:warmup-fpr}

Having noted that conformal inferences hold in theory only marginally over the calibration data, the first question one may ask is: how bad can these inferences be conditional on a particular calibration set?
We will address this question by developing high-probability bounds for the conditional deviation from uniformity of marginal p-values, starting here from the simplest case of pointwise bounds. 
The purpose of a pointwise bound is to control the probability that a null p-value (corresponding to a true inlier) is smaller than $\alpha$, conditional on $\mathcal{D}$, for some {\em fixed} threshold $\alpha \in (0,1)$.
In other words, we wish to understand the conditional false positives rate (FPR) corresponding to the threshold $\alpha$,
\begin{equation} \label{eq-fpr}
    \mathrm{FPR}(\alpha; \mathcal{D}) \defeq \P\left[\hat{u}^\marg(X_{2n+1}) \leq \alpha \mid \mathcal{D}\right],
\end{equation}
beyond what we know from the marginal guarantee in~\eqref{eq:marg_validity_def}, which is $\E\left[\mathrm{FPR}(\alpha; \mathcal{D})\right] \le \alpha$.
The quantity in~\eqref{eq-fpr} can be studied precisely with existing results due to~\cite{vovk2012conditional}. We revisit this topic here because it serves as an intuitive introduction to the more involved high-probability bounds that we will propose later. 

{Looking at the definition of $\hat{u}^\marg(X)$ in~\eqref{eq:marginal-pvals-def}}, we see that, if $\hat{s}(X)$ has a continuous distribution,
\begin{equation*}
    \mathrm{FPR}(\alpha; \mathcal{D}) = F\left(\hat{F}_n^{-1}\left(\frac{(n + 1) \alpha}{n}\right)\right),
\end{equation*}
{where $F$ and $\hat{F}_n$ are, respectively, the true and empirical (evaluated on the calibration data) CDF of $\hat{s}(X)$.}
Therefore, the deviation of $\mathrm{FPR}(\alpha; \mathcal{D})$ {(a random variable depending on $\mathcal{D}$)} from $\alpha$ depends on the quality of $\hat{F}_n^{-1}((n+1)\alpha/n)$ as an approximation of $F^{-1}(\alpha)$, which can be understood through classical results for the order statistics of uniform variables.

\begin{prop}[Pointwise FPR of marginal conformal p-values, from~\cite{vovk2012conditional}] \label{prop:dist_coverage}
Let $\ell = \lfloor (n+1)\alpha \rfloor$. If $\hat{s}(X)$ is continuously distributed, $\mathrm{FPR}(\alpha; \mathcal{D})$ follows a $\betadist(\ell, n + 1 - \ell)$ distribution.
\end{prop}

Figure~\ref{fig:beta_curves} visualizes the FPR distribution from Proposition~\ref{prop:dist_coverage}, due to~\cite{vovk2012conditional}, for different values of the calibration set size. This shows precisely how a smaller $\Dcal$ makes marginal p-values more conservative on average, but also more likely to be overly liberal on occasion.
For example, we can see there is a non-negligible probability that $\mathrm{FPR}(0.1; \mathcal{D}) > 0.15$ with 100 calibration points, whereas it seems very unlikely that $\mathrm{FPR}(0.1; \mathcal{D}) > 0.12$ with 1600 calibration points. However, it is still quite possible that $\mathrm{FPR}(0.01; \mathcal{D}) > 0.015$ even with 1600 calibration points.
In general, Proposition~\ref{prop:dist_coverage} implies the coefficient of variation (relative spread) of the FPR is approximately proportional to $(|\mathcal{D}^{\mathrm{cal}}| \alpha)^{-1/2}$.
While this result is informative and it is broadly relevant to the issue of how to best choose the number of calibration data points for split-conformal inference~\cite{sesia2020comparison}, it is limited for our purposes. In fact, it provides only a pointwise bound---it takes $\alpha$ as fixed---whereas uniform bounds are needed to construct conditionally valid p-values that can be safely used with any multiple-testing procedure, as discussed in the next section.

\begin{figure}[!htb]
    \centering
    \includegraphics[width = 0.85\textwidth]{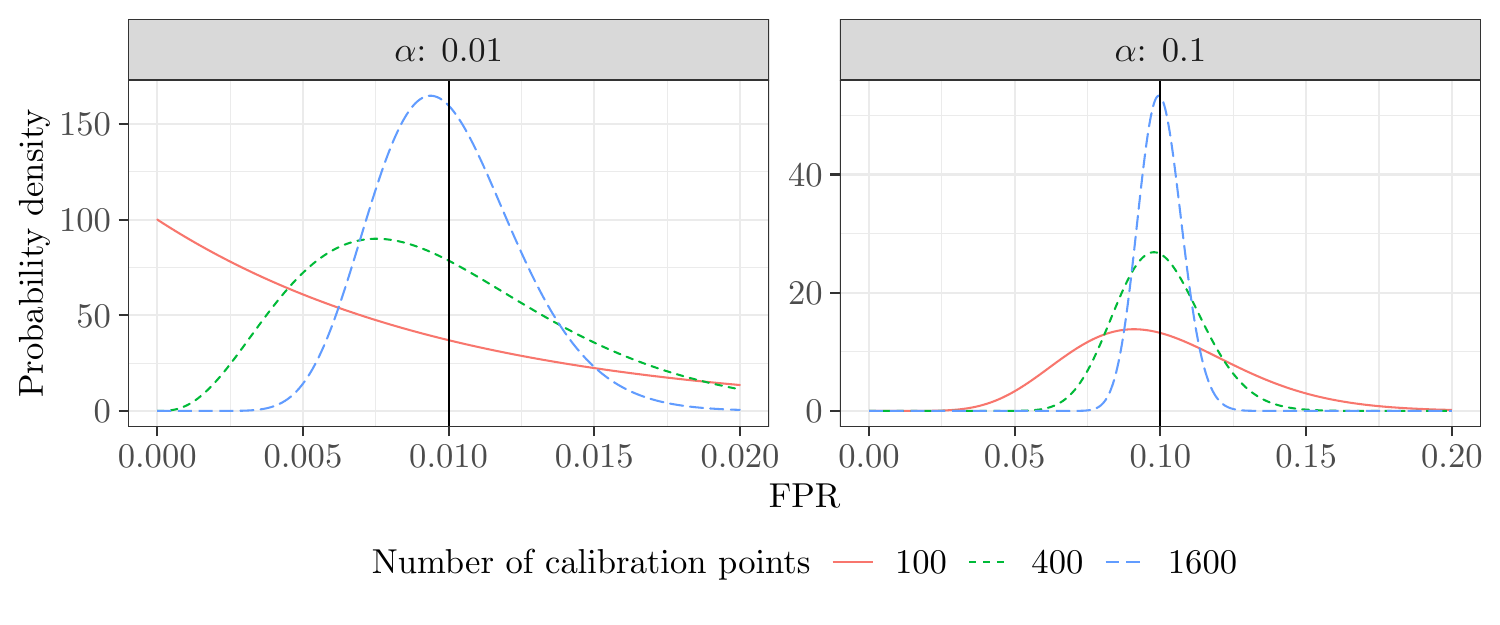}
    \caption{Distribution of the false positive rate obtained by thresholding marginal conformal p-values at levels $\alpha=0.01$ and $\alpha=0.1$, as a function of the number of calibration points.}
    \label{fig:beta_curves}
\end{figure}

\subsection{A generic strategy to adjust marginal conformal p-values}

Proposition~\ref{prop:dist_coverage} implies marginal conformal p-values may be anti-conservative conditional on $\mathcal{D}$. Therefore, in the language of~\eqref{eq:pval_form}, our goal is to find an adjustment function leading to conditionally valid p-values, i.e.,  satisfying \eqref{eq:cond_validity_def}. The following theorem suggests a generic strategy through a simultaneous upper confidence bound for order statistics.

\begin{thm}[Conditional p-value adjustment]\label{thm:generic}
  Let $U_1, \ldots, U_n\stackrel{\mathrm{\mathrm{i.i.d.}}}{\sim}\Unif([0, 1])$, with order statistics $U_{(1)}\le U_{(2)}\le \ldots \le U_{(n)}$, and fix any $\delta \in (0,1)$. Suppose $0\le b_1\le b_2\le \ldots \le b_{n}\le 1$ {are} $n$ reals such that
  \begin{equation}
    \label{eq:confidence_band}
    \P\left[ U_{(1)}\le b_{1}, \ldots, U_{(n)}\le b_{n} \right] \ge 1 - \delta.
  \end{equation}
    Let also {$b_{0} = 0, b_{n+1} = 1$, and} $h: [0, 1]\mapsto [0, 1]$ be a piece-wise constant function such that
    \begin{align} \label{eq:piecewise-f}
      h(t) = b_{\lceil (n + 1)t\rceil}, \,\, t\in [0, 1].
    \end{align}
    Then, $\hat{u}^\cond = h\circ \hat{u}^\marg$ satisfies \eqref{eq:cond_validity_def}, i.e., $\hat{u}^\cond(X_{2n+1})$ is a calibration-conditional valid p-value. 
\end{thm}

Figure~\ref{fig:bseq_illustration} illustrates the idea of Theorem~\ref{thm:generic}.
Here, we set $n=1000$ and  generate $100$ independent realizations of the order statistics $(U_{(1)}, \ldots, U_{(n)})$.
Each of the 100 blue curves corresponds {to} a sample path, plotted against the normalized index $i / n$. The black curve tracks the {theoretical} mean of $(U_{(1)}, \ldots, U_{(n)})$, while the orange and yellow curves correspond to two particular sequences of $b_i$ values derived from the generalized Simes inequality for $\delta=0.1$ and the DKWM~\citep{dvoretzky1956asymptotic, massart1990tight} inequality, detailed in the next subsection.
We observe relatively few paths cross the orange curve, and all crossings occur {at small indices.} This suggests the upper confidence bounds provided by Theorem~\ref{thm:generic} can be especially tight for lower indices of the order statistics, which is essential to obtain reasonably powerful CCV p-values for outlier detection. Of course, calibration-conditional validity still necessarily comes at some power cost.
For example, a marginal p-value of $\hat{u}^\marg(X) = 25 / (n + 1) \approx 0.025$ results in a CCV p-value of $h(25 / (n + 1)) = b_{25} \approx 0.0377$ in this case.

\begin{figure}
    \centering
    \includegraphics[width = 0.75\textwidth]{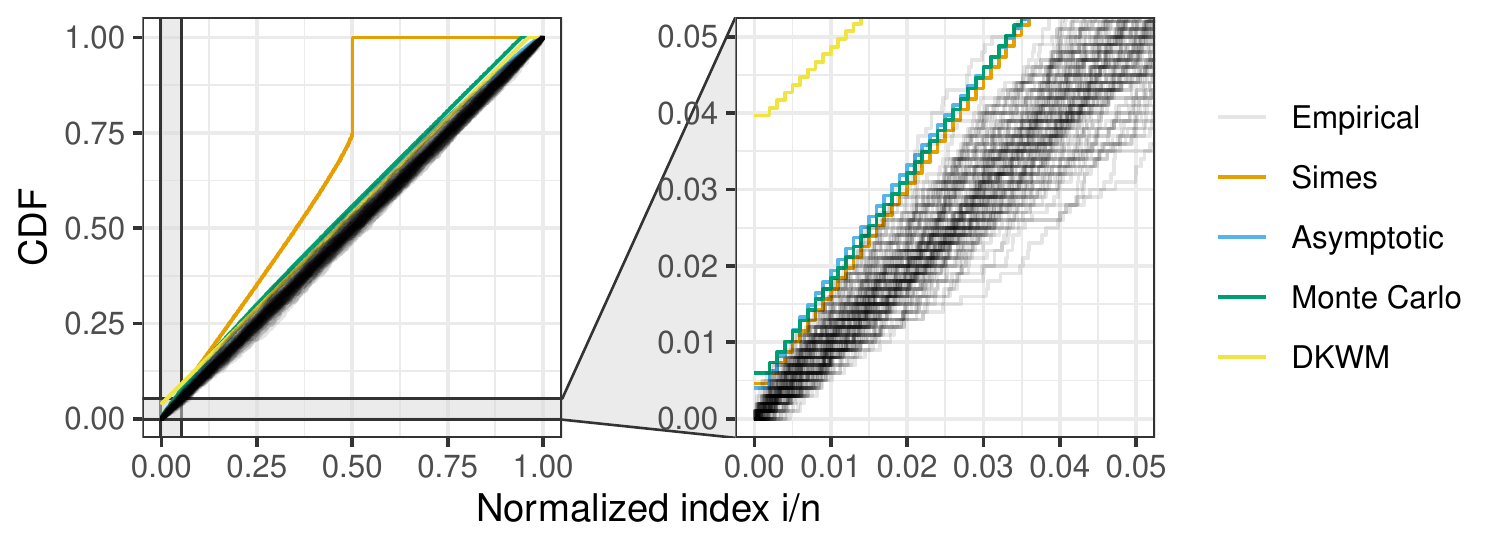}
    \caption{Illustration of Theorem~\ref{thm:generic} with $n=1000$ and $\delta=0.1$. The orange and yellow curves give the sequences derived by the generalized Simes inequality with $k = 500$ and the DKWM inequality, respectively. The blue and green curves (very close to each other) give the corresponding sequences obtained with the asymptotic and Monte Carlo adjustments described below. The right panel zooms in on small indices. }
    \label{fig:bseq_illustration}
\end{figure}

\subsection{Simes adjustment of marginal conformal p-values} \label{sec:gen-simes}

The larger p-values typically do not matter in multiple testing problems, as it is the small ones that determine which hypotheses are rejected.
Therefore, to maximize power, we would like the $b_i$ values in Theorem~\ref{thm:generic} to be as small as possible for low indices $i$, while we may be satisfied with letting $b_{i} = 1$ for large $i$.
The generalized Simes inequality yields a desirable class of $(b_1, \ldots, b_n)$ sequences with this property. 

\begin{prop}[Generalized Simes Inequality, from Equation (3.5) in~\cite{sarkar2008generalizing}]\label{prop:simes}
For any positive integer $k\le n$, the uniform bound~\eqref{eq:confidence_band} in Theorem~\ref{thm:generic} holds with
\begin{align} \label{eq:simes-sequence}
b^{\mathrm{s}}_{n + 1 - i} = 1 - \delta^{1/k}\left(\frac{i\cdots (i - k + 1)}{n\cdots (n - k + 1)}\right)^{1/k}, \qquad i = 1, \ldots, n.
\end{align}
\end{prop}
The original motivation of~\cite{sarkar2008generalizing} was to compute thresholds for step-up procedure to achieve $k$-FWER control; there, the parameter $k$ was set to be a small integer.
Here, we exploit Proposition~\ref{prop:simes} differently, choosing {$k = n / 2$} so that the $b^{\mathrm{s}}_{i}$ values {with} lower indices $i$ are as small as possible while those {with larger indices} $i$ may be uninformative (note that $b^{\mathrm{s}}_{n - k + 2} = \ldots = b^{\mathrm{s}}_{n} = 1$).
In particular, our choice corresponds to
\begin{equation*}
{b^{\mathrm{s}}_{1} = 1 - \delta^{2/n} = 1 - \exp\left\{-\frac{2\log(1 / \delta)}{n}\right\} \approx \frac{2\log(1 / \delta)}{n}.}
\end{equation*}
Therefore, the smallest possible marginal p-value equal to $1/(n+1)$ would be mapped to $h(1 / (n + 1)) \approx 2\log(10) / n = 4.61 / n$, if $\delta = 0.1$, for example, since $\hat{u}^\cond(X) = h(\hat{u}^\marg(X))$.
If $n = 1000$, then $h(1 / (n + 1)) \approx 0.0046$, which is larger than the marginal p-value, but much smaller than what one would obtain from other standard uniform bounds. For example, the DKWM inequality~\citep{dvoretzky1956asymptotic, massart1990tight} would imply a result similar to that of Proposition~\ref{prop:simes} but with 
\begin{equation}\label{eq:bi_dkwm}
b_i^\dkwm = \min\{(i / n) + \sqrt{\log(2 / \delta) / 2n}, 1\};
\end{equation}
this would map the smallest possible marginal p-value to $1 / (n + 1) + \sqrt{\log(2 / \delta) / 2n} > 0.1$, in the above example.
The comparison between the generalized Simes inequality and the DKWM inequality is expanded in Appendix~\ref{app:comparison}, where we also consider an additional uniform bound based on the linear-boundary crossing probability for the empirical CDF~\citep{dempster1959generalized}.
This comparison confirms the generalized Simes inequality yields the most powerful adjustment for our multiple testing purposes.
In practice, we {find that $k = n/2$ works well}, as motivated empirically in Appendix~\ref{app:sim}.
(Note that larger values of $k$ would lower further the smallest possible adjusted p-value, but at the cost of raising other small p-values).

\subsection{Asymptotic adjustment of marginal conformal p-values} \label{sec:gen-asymptotic}

The Simes adjustment with $k=n/2$ leads to p-values satisfying~\eqref{eq:cond_validity_def} exactly; however, this causes the smallest possible marginal conformal p-values to be inflated by a factor of order $1/n$, and larger ones may be inflated even more.
A natural question at this point is whether this approach is statistically efficient or whether more powerful alternatives may be available to achieve~\eqref{eq:cond_validity_def}.
We begin to address this matter by comparing the Simes adjustment to an alternative {\it asymptotic} approach that provides a natural benchmark; this solution will be valid in the limit of large $n$ but does not guarantee~\eqref{eq:cond_validity_def} exactly in finite samples.
Recall Donsker's theorem, the classical result from empirical process theory stating that, in the large-$n$ limit, the rescaled difference between the true and the empirical CDFs of the calibration scores, respectively $F$ and $\hat{F}_n$, converges in distribution to a standard Brownian Bridge. Precisely,
$\sqrt{n} (\hat{F}_n - F) \overset{d}{\rightarrow} \mathbb{G}$,
where $\mathbb{G}$ is the Gaussian process on $[0,1]$ with
mean zero and covariance $\mathbb{E}[\mathbb{G}(t_1)\mathbb{G}(t_2)] = t_1 \wedge t_2 - t_1 t_2$, for all $t_1, t_2 \in [0,1]$. 
This result suggests the following {\em asymptotic adjustment} of marginal conformal p-values. 

As a starting point, note that $\sup_{t \in [0,1]}|\mathbb{G}(t)|$ follows the Kolmogorov distribution \cite{kolmogorov1933sulla}, whose $1-\delta$ quantile, namely $q^{\text{K}}_{\delta}$, can be computed. Therefore, a simple way of constructing approximately valid conditional conformal p-values would be to add $q^{\text{K}}_{\delta} / \sqrt{n}$ to the marginal p-values. Unfortunately, this naive solution would suffer from the same limitation of the DKWM approach mentioned in the previous section: it is a correction of constant size which is not very attractive for multiple testing because it is extremely conservative for small p-values of order $1/n$.
Instead, a more useful solution is suggested by the adaptive bound of \cite{eicker1979asymptotic}, which proved that the empirical process $\hat{V}_n(t)$ defined as
\begin{align*}
\hat{V}_n(t) = \sqrt{n} \frac{F(t) - \hat{F}_n(t)}{ \sqrt{\hat{F}_n(t)[1 - \hat{F}_n(t)]}}, \qquad t \in [0,1],
\end{align*}
satisfies $\lim_{n \to \infty} \P[ \sup_{t \in [0, 1]}\hat{V}_n(t) \leq c_n(\delta) ] \geq 1-\delta$,
where $c_n(\delta)$ is defined as
\begin{align*}
c_n(\delta) \defeq \frac{-\log[-\log(1-\delta)]+2\log\log n + (1/2) \log \log \log n - (1/2) \log\pi }{\sqrt{2 \log \log n}}.
\end{align*}
This yields a straightforward asymptotic simultaneous upper confidence bound for $F(t)$ and, in light of Theorem~\ref{thm:generic}, it suggests the following approximately valid adjustment of marginal conformal p-values: 
\begin{align} \label{eq:a-ccv}
\hat{u}^{\textnormal{(a-ccv)}} = h^{\mathrm{a}}\circ \hat{u}^\marg,
\end{align}
where $h^{\mathrm{a}}$ is the piece-wise constant function on $[0,1]$ defined such that $h^{\mathrm{a}}(t) = b^{\mathrm{a}}_{\lceil (n + 1)t\rceil}$, for $t\in [0, 1]$, with $b^{\mathrm{a}}_0=0$, $b^{\mathrm{a}}_{n+1}=1$, and 
\begin{align} \label{eq:a-sequence}
b^{\mathrm{a}}_i = \min\left\{\frac{i}{n} + c_n(\delta) \frac{\sqrt{i(n-i)}}{n \sqrt{n}}, 1\right\},
\quad i = 1,\ldots,n.
\end{align}
In Appendix \ref{subsubsec:fisher_asym_monotonicity}, we will show that $b_1^{\mathrm{a}} \le b_2^{\mathrm{a}}\le \ldots \le b_n^{\mathrm{a}}$, as required by Theorem \ref{thm:generic}. See Figure~\ref{fig:bseq_illustration} for a visualization of the simultaneous CDF bound corresponding to this adjustment function.
The smallest possible marginal p-value is mapped by this function to $h^{\mathrm{a}}(1 / (n + 1)) \approx (1+c_n(\delta))/n$. For example, if $\delta = 0.1$ and $n=1000$, this is approximately $4.09/n \approx 0.0041$, which is very similar to the corresponding constant $0.0046$ obtained with the Simes adjustment.
However, $\hat{u}^{\textnormal{(a-ccv)}}$ has the advantage of being reasonably tight for all p-values, not just the smallest ones, and thus it will generally allow for higher power compared to the Simes adjustment when $n$ is large.

\subsection{Monte Carlo adjustment of marginal conformal p-values} \label{sec:gen-mcmc}

Although the Simes adjustment is more conservative than the asymptotic one in the limit of large $n$, it has two distinct advantages in finite samples. First, it leads to p-values satisfying~\eqref{eq:cond_validity_def} exactly, with no asymptotic approximations. Second, the peculiar shape of its uniform empirical CDF envelope allows it to apply smaller corrections to relatively low p-values, possibly yielding higher power in multiple-testing applications; see Figure~\ref{fig:compare_adjustments} for an illustration.
These observations motivate the development of the following new type of adjustment function, which is based on {\em Monte Carlo} rather than analytical calculations and is designed to combine the strengths of the two aforementioned approaches.
In particular, the Monte Carlo solution proposed here is based on a uniform empirical CDF bound that is (a) theoretically valid in finite samples and (b) whose shape mimics that of the Simes approach for very small p-values while tracking the asymptotic envelope relatively closely for larger ones; see Figure~\ref{fig:compare_adjustments} for a preview.

\begin{figure}[!htb]
    \centering
    \includegraphics[width=0.75\textwidth]{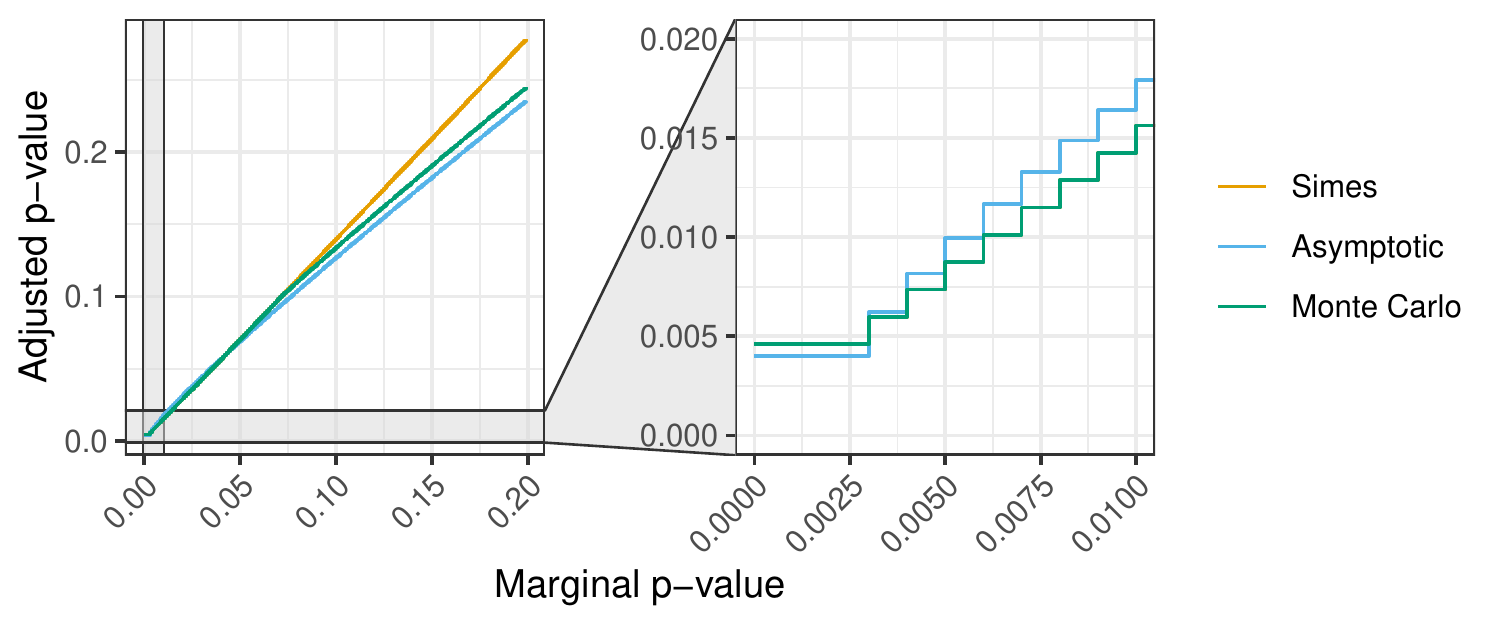}
    \caption{Comparison of different adjustment functions, with $n=1000$ and $\delta=0.1$. In the zoomed-in panel on the right-hand-side, the Simes (orange) and Monte Carlo (green) curves cannot be distinguished.}
    \label{fig:compare_adjustments}
\end{figure}

Having fixed any $n$ and $\delta$, denote by $h^{\mathrm{s}}: [0,1] \to [0,1]$ the Simes piece-wise constant function obtained by combining~\eqref{eq:piecewise-f} with~\eqref{eq:simes-sequence}, using $k=n/2$. Recall that this satisfies~\eqref{eq:confidence_band} exactly.
Let also $h^{\mathrm{a},\hat{\delta}}: [0,1] \to [0,1]$ denote the asymptotic piece-wise constant function obtained by combining~\eqref{eq:piecewise-f} with~\eqref{eq:a-sequence}, after replacing the pre-specific parameter $\delta$ with a variable $\hat{\delta}$, which can take any values in $(0,1)$. Note that it will be useful to keep the dependence of this function on $\hat{\delta}$ explicit. 
Recall that $h^{\mathrm{a},\hat{\delta}}$ satisfies~\eqref{eq:confidence_band} approximately if $n$ is large and $\hat{\delta} = \delta$.
Next, define a new piece-wise constant function $h^{\mathrm{m},\hat{\delta}}: [0,1] \to [0,1]$ as:
\begin{align} \label{mcmc-sequence}
  h^{\mathrm{m},\hat{\delta}}(t) = \min \left\{ h^{\mathrm{s}}(t), h^{\mathrm{a},\hat{\delta}}(t)\right\},  \qquad t\in [0, 1].
\end{align}
Note that this function can be conveniently written in the form of~\eqref{eq:piecewise-f} with a suitable choice of $b_1,\ldots,b_n$.
Now, the goal is to find the smallest possible $\hat{\delta}$, as a function of $n$ and $\delta$, such that the $b_1,\ldots,b_n$ sequence corresponding to the function $h^{\mathrm{m},\hat{\delta}}$ defined in~\eqref{mcmc-sequence} satisfies~\eqref{eq:confidence_band}. The problem can be solved with a bisection search for $\hat{\delta}$ on $(0,1)$, approximating the probability in~\eqref{eq:confidence_band} through a simple Monte Carlo simulation---it suffices to generate a sufficiently large number of independent random samples of size $n$ from a uniform distribution. A feasible solution always exists because $h^{\mathrm{m},\hat{\delta}}$ reduces to $h^{\mathrm{s}}$ as $\hat{\delta} \to 1$, and $h^{\mathrm{s}}$ satisfies~\eqref{eq:confidence_band}. This Monte Carlo simulation is not computationally expensive for reasonable values of $n$, as long as $\delta$ is not too small; for example, it takes a few seconds on a personal computer to obtain a very accurate estimate of $\hat{\delta}$ with $\delta=0.1$ and $n$ as large as $10,000$. Of course, if $n$ is extremely large, the Monte Carlo simulation is not even needed, as in that case one could just rely directly on the asymptotic adjustment.
See Figure~\ref{fig:bseq_illustration} for a visualization of the simultaneous CDF bound corresponding to this adjustment function.

While the Monte Carlo adjustment approaches the asymptotic one in the limit of large $n$, it may lead to more powerful p-values for multiple testing if $n$ is small. In fact, the Simes function $h^{\mathrm{s}}(t)$ can be lower than the asymptotic $h^{\mathrm{a},\delta}(t)$ for values of $t$ very close to 0, and $h^{\mathrm{m},\hat{\delta}}(t)$ inherits this ability of preserving very small p-values relatively intact, as shown in the right-hand-side panel of Figure~\ref{fig:compare_adjustments}.
At the same time, as it will be demonstrated shortly, the Monte Carlo adjustment tends to be more powerful than the Simes adjustment when testing a single hypothesis, or when dealing with many non-null hypotheses, because $h^{\mathrm{a},\delta}(t)$ is lower than $h^{\mathrm{s}}(t)$ for moderately small values of $t$; see the left-hand-side panel of Figure~\ref{fig:compare_adjustments}. Additional figures in Appendix~\ref{app:comparison} show that this relative advantage grows even larger as $n$ increases.

The Monte Carlo adjustment applied in this paper and implemented in the accompanying software package involves an additional modification to the expression in~\eqref{mcmc-sequence}, whose discussion has been postponed until now to simplify the explanation. In practice, $h^{\mathrm{m},\hat{\delta}}(t)$ is defined as in~\eqref{mcmc-sequence} only for $t \leq 1/2$; then, for $t>1/2$, the function is extended it as a tangent straight line because there would be little point in tightening the CDF envelope above $1/2$, as that region involves p-values unlikely to be rejected anyway. The advantage of this approach is that it decreases the boundary crossing probability of the empirical CDF for all $t>1/2$ compared to the asymptotic solution, allowing a slightly more liberal adjustment for the more interesting p-values below 1/2; see Figure~\ref{fig:compare_mcmc_opposite} in Appendix~\ref{app:comparison}.

\subsection{Power analyses of conformal p-value adjustments} \label{sec:gen-power}

As marginal p-values are smaller than calibration-conditional p-values, the latter tend to involve some loss of power, while the former are not always valid, depending on the multiple testing procedure utilized.
In this section, we would like to study the power gap between the marginal and calibration-conditional approaches within settings in which both types of conformal p-values lead to valid tests.
However, traditional power analyses require stronger modeling assumptions (i.e., the distributions of inliers and outliers) and the specification of additional algorithmic details (i.e., the form of the conformity score functions) compared to the framework followed in this paper; in fact, conformal p-values are extremely flexible and can be applied in fully non-parametric settings with any conformity score function.
We overcome this hurdle by analyzing the {\em effective level} of a test applied to calibration-conditional p-values as a proxy for a power analysis. 
More precisely, a test at level $\alpha$ applied to calibration-conditional p-values is generally equivalent to an analogous test at level $\alpha'$ applied to marginal p-values, for some $\alpha' < \alpha$.
Comparing $\alpha$ to $\alpha'$ gives a measure of the loss in power incurred by calibration-conditional p-values that is specific to a particular testing procedure, but requires no assumptions about either the machine learning model utilized to compute conformity scores or the inlier and outlier distributions. 
Thus, $\alpha'$ is studied below for different testing procedures.

\subsubsection{Testing a single hypothesis} \label{sec:power-single}

Consider the problem in which a marginal conformal p-value $\hat{u}^\marg(X_{2n+1})$ for a single test point $X_{2n+1}$ is available, and we wish to test whether $X_{2n+1}$ is an outlier. The level-$\alpha$ test based on the marginal p-value rejects when $\hat{u}^\marg(X_{2n+1}) \le \alpha$. We will compare this to a test based on a calibration-conditional p-value. That is, we take the marginal p-value and adjust it with a generic piece-wise constant function $h : [0,1] \to [0,1]$ in the form of~\eqref{eq:piecewise-f}. Then, we reject the null if $h \circ \hat{u}^\marg \leq \alpha$, or, equivalently, if
\begin{equation*}
  \hat{u}^\marg \leq {i^*(\alpha; h)}/{(n+1)},
\end{equation*}
where $i^*(\alpha; h) = \max \left\{ i \in \{1,\ldots,n\} : b_i \leq \alpha \right\}$ and $b_1, \ldots, b_n$ indicate the step positions defining $h$ in~\eqref{eq:piecewise-f}.
Since $\hat{u}^\marg$ is uniformly distributed, $i^{*}(\alpha; h) / (n + 1)$ is the effective level of the analogous marginal test. 

With the asymptotic adjustment $h^{\mathrm{a}}$, the threshold for the calibration-conditional test can be calculated explicitly by solving a quadratic equation, and the solution in the large-$n$ limit take the following form:
\begin{align*}
  \frac{i^*(\alpha; h^\asym)}{n+1}
  & = O\left(\frac{\alpha}{1+ c^2_n(\delta)/n}\right) = \alpha \left[ 1 - O \left(\frac{\log \log n}{n} \right) \right],
\end{align*}
because
\begin{align*}
  i^*(\alpha; h^{\mathrm{a}})
    = 
    \left\lfloor\frac{c^2_n(\delta) n + 2 n^2 \alpha - c_n(\delta) n \sqrt{c^2_n(\delta) + 4n \alpha - 4n \alpha^2}}{2 [c^2_n(\delta) + n]}\right\rfloor.
\end{align*}
In words, the cost in power of the asymptotic p-value adjustment from Section~\ref{sec:gen-asymptotic} can be understood by noting that the significance threshold $\alpha$ is effectively decreased by a factor of order $(\log \log n)/n$. Similarly, the effective $\alpha$-level with the DKWM adjustment $h^\dkwm$, given by \eqref{eq:bi_dkwm}, is $\alpha - O(1/\sqrt{n})$. By contrast, for the Simes adjustment, we can show the effective $\alpha$-level is strictly below $\alpha$ when $k = \lceil\zeta n\rceil$ for some $\zeta > 0$. 
In fact, using the concavity of the mapping $a(x) = \log(1 - 1/x)$, Jensen's inequality implies
\begin{equation}\label{eq:bi_simes}
b_{i}^{\simes} = 1 - \delta^{1/k}\exp\left\{\frac{1}{k}\sum_{\ell = n-k+1}^{n}a\left(\frac{\ell}{i - 1}\right)\right\}\ge 1 - \exp\left\{a\left(\frac{n - k / 2 + 1/2}{i - 1}\right)\right\} = \frac{i - 1}{n - k /2 + 1/2}.
\end{equation}
As a result,
\begin{equation}\label{eq:simes_alpha}
\frac{i^{*}(a; h^{\simes})}{n+1} \le  \alpha(1 - \zeta / 2) + o(1).
\end{equation}
In this sense, the asymptotic and DKWM adjustment are nearly as efficient as the marginal test for a single hypothesis, though the former is more powerful, while the Simes adjustment is asymptotically inefficient.

Analogous threshold calculations for the Monte Carlo adjustments in the same setting cannot be performed analytically because $i^*(\alpha; h^{\mathrm{m},\hat{\delta}})$ no longer has a simple expression for the  sequences $b$ corresponding to those functions $h$. However, these analyses are easy to carry out numerically. 
Figure~\ref{fig:power_single} (a) summarizes the results of these power analyses by comparing the effective significance levels obtained with these three alternative adjustment functions, as a function of $n$. 
The results show the Monte Carlo adjustment behaves very similarly to the efficient asymptotic solution in the limit of large $n$, but it can be even more powerful when the sample size is small thanks to the shape of its CDF envelope, which reduces the inflation of smaller p-values. The Simes adjustment behaves similarly to the Monte Carlo one when the sample size is small, but it is not efficient in the large-$n$ limit. In that case, the effective significance level for testing a single hypothesis does not converge at all to the nominal level $\alpha$ in the large-$n$ limit.

\begin{figure}[!htb]
    \centering
    \includegraphics[width=\textwidth]{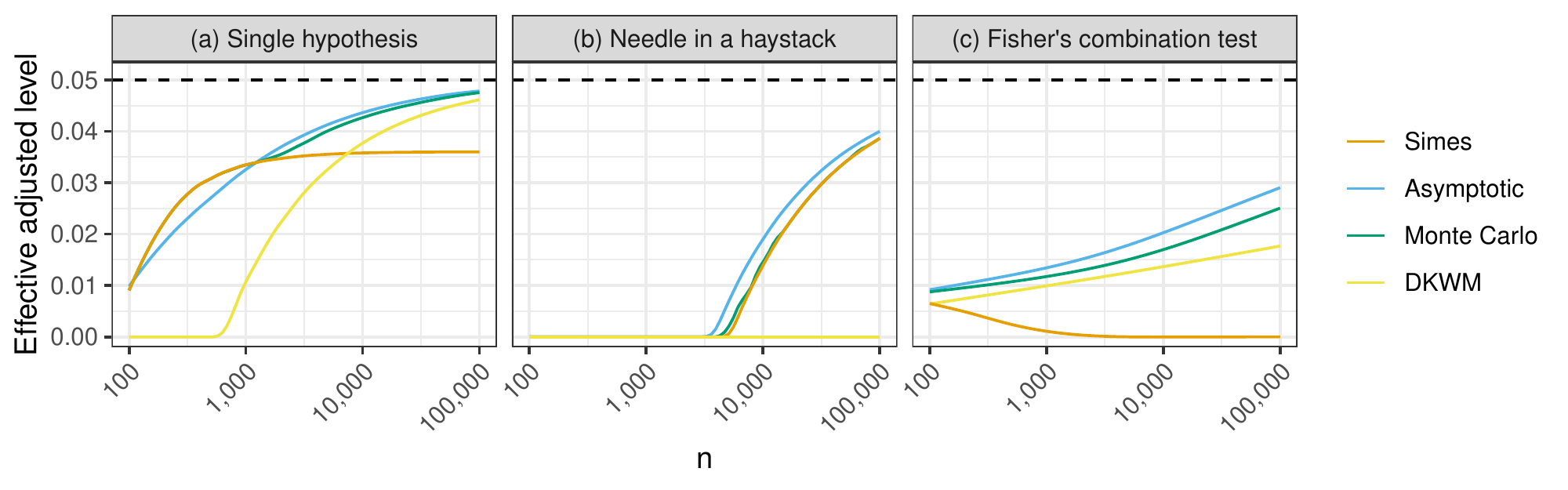}
    \caption{Power analysis of different adjustments for marginal conformal p-values under 3 alternative settings. The effective level resulting from the p-value adjustment for a test at nominal level $\alpha=0.05$ (dashed horizontal line) is plotted as a function of the number of calibration samples, assuming the number of test points $m$ grows as $\sqrt{n}$. (a) Testing a single hypothesis. (b) FWER control with a single strong signal (here the values for DKWM are all equal to 0). (c) Testing a global null with Fisher's combination test.
}
    \label{fig:power_single}
\end{figure}

\subsubsection{Needle in a haystack}

Consider a multiple testing problem in which there are $m$ possible outliers to be tested: the first one of these data points, $X_{2n+1}$, is an outlier (a false null hypothesis), while the remaining $m-1$, $X_{2n+2},\ldots,X_{2n+m}$, are inliers (true nulls). The goal is to identify the outlier, controlling the family-wise error rate below $\alpha$. To further simplify the problem, imagine the signal strength for the true outlier is so high that the marginal conformal p-value for this point takes its minimal value with probability one:
\begin{align*}
\hat{u}^\marg_{2n+1} = \frac{1}{n+1}.
\end{align*}
Then, we reject the null if the adjusted p-value for the outlier is below the Bonferroni level:
\begin{align} \label{eq:power-needle-generic}
  h \circ \hat{u}^\marg_{2n+1} \leq \frac{\alpha}{m}.
\end{align}
In the case of the asymptotic adjustment function, the rejection event can be written as:
\begin{align*}
  \left\{ h^{\mathrm{a}} \circ \hat{u}^\marg_{2n+1} \leq \frac{\alpha}{m} \right\}
   \iff 
   \left\{\hat{u}^\marg_{2n+1} + \frac{1}{n(n+1)} + c_n(\delta) \frac{\sqrt{n-1}}{n \sqrt{n}}  \leq \frac{\alpha}{m}\right\}.
\end{align*}
Thus, the calibration-conditional test at level $\alpha$ is equivalent to the marginal test at level $(\alpha + \Delta \alpha)/m$, where
\begin{align*}
  \Delta \alpha 
  & = - \frac{m}{n} \left( \frac{1}{n+1} + c_n(\delta) \sqrt{\frac{n-1}{n}} \right)
    = - \frac{m}{n} \sqrt{2 \log \log n} (1 + o(1)).
\end{align*}
In this regime, the calibration-conditional and marginal tests only differ by a $\sqrt{\log\log n}$ factor. 

In the case of the Simes adjustment with $k=n/2$, the rejection event is
\begin{align*}
   \left\{h^{\mathrm{s}} \circ \hat{u}^\marg_{2n+1} \leq \frac{\alpha}{m}\right\}
  \iff \left\{\hat{u}^\marg_{2n+1} + \frac{2 \log(1/\delta)}{n}(1 + o(1)) - \frac{1}{n+1} \leq \frac{\alpha}{m}\right\},
\end{align*}
which implies the equivalent level for the test is $(\alpha + \Delta \alpha)/m$, with
\begin{align*}
  \Delta \alpha 
  & = - \frac{m}{n} \left(2\log(1/\delta)-1 + o(1)\right).
\end{align*}
Similarly, for the DKWM adjustment, it is easy to see that
\[\Delta \alpha = -\frac{m}{\sqrt{n}}\left(\sqrt{\frac{\log(2/\delta)}{2}} + o(1)\right).\]

Therefore, in the large-$n$ limit, the Simes adjustment is even more powerful than the asymptotic correction for this problem because it does not involve the slightly sub-optimal $\sqrt{\log \log n}$ factor. Unsurprisingly, the large additive inflation by the DKWM adjustment results in a large power loss.
Although the Monte Carlo method is not as amenable to analytical calculations, it is easy to verify numerically that its power is almost the same as that of the asymptotic correction in this setting; see Figure~\ref{fig:power_single} (b).
Interestingly, the numerical power analysis in Figure~\ref{fig:power_single} (b) also highlights that the asymptotic adjustment, although slightly less powerful in the large-$n$ limit, tends to be more powerful than the Simes adjustment for this problem. In fact, $\sqrt{\log \log n} < \left(2\log(1/\delta)-1\right)$ unless $n$ is extremely large or $\delta$ is extremely small.


\subsubsection{Fisher's combination test of the global null}

Consider a multiple testing problem in which there are $m$ test data points $X_{2n+1}, \ldots, X_{2n+m}$ and none of them are outliers. The goal is to test the global null by applying Fisher's combination test to conformal p-values modified by an adjustment function $h$, for different choices of the latter. Intuitively, the effective $\alpha$-level of this test will depend on the expected value of Fisher's combination statistic under the null---a smaller $\E_{H_0}[-\log (h\circ \hat{u}^\marg)]$ yields a more conservative test. Therefore, we begin by deriving this quantity analytically for the asymptotic, DKWM, and Simes adjustments; see Appendix \ref{sec:power_fisher} for further details. 
\begin{thm}[Expected value of Fisher's combination statistic with conformal p-values] \label{theorem:fisher-mean}
Fixing $\delta > 0$ and letting $n\rightarrow \infty$, 
\begin{enumerate}[(a)]
\item $\displaystyle\E_{H_0}[-\log (h^\asym\circ \hat{u}^\marg)] = 1 - \frac{\pi}{2}\frac{c_n(\delta)}{\sqrt{n}} + O\lb \frac{(\log n)(\log \log n)}{n}\rb.$
\item $\displaystyle\E_{H_0}[-\log (h^\dkwm\circ \hat{u}^\marg)] = 1 - b_n(\delta)\log\lb\frac{e}{b_n(\delta)}\rb + O\lb\frac{\log n}{n}\rb, \,\, \text{where} \,\,b_{n}(\delta) = \sqrt{\frac{\log(2 / \delta)}{2n}}.$
\item   Assume that $k = \lceil\zeta n\rceil$ for some $\zeta > 0$. Then
  \[\E_{H_0}[-\log (h^\simes \circ \hat{u}^\marg)] \le 1 - \zeta - (1 - \zeta)\log (1 - \zeta) + O\lb\frac{\log n}{n}\rb.\]
\end{enumerate}
\end{thm}
All three adjustments yield conservative tests because $\E_{H_0}[\sum_{i=1}^{m}-2\log(h\circ \hat{u}^\marg(X_{2n+i}))] < 2m$ asymptotically. The gap is $O(m\sqrt{\log \log n} / \sqrt{n})$ for the asymptotic adjustment (the most efficient one in this case), $O(m\log n / \sqrt{n})$ for the DKWM adjustment, and $O(m)$ for the Simes adjustment (the least efficient one in this case). In Appendix \ref{sec:power_fisher}, we compute the effective $\alpha$-level for each adjustment in different regimes. As those derivations are lengthy, we summarize the results below. 
\begin{itemize}
\item For the asymptotic adjustment, the effective $\alpha$-level is $\alpha(1 + o(1))$ if $m = o(n / \log \log n)$, and $O(1 / \log^c n)$ for some constant $c$ when $m = \gamma n$ for some $\gamma \in (0, 1)$.
\item For the DKWM adjustment, the effective $\alpha$-level is $\alpha(1 + o(1))$ if $m = o(n / \log^2 n)$, and $\exp\{-O(\log^2 n)\}$ when $m = \gamma n$ for some $\gamma \in (0, 1)$.
\item For the Simes adjustment, the effective $\alpha$-level is $\exp\{-O(\min\{m, n\} / \log n)\}$ if $m / \log n\rightarrow \infty$.
\end{itemize}

In Figure~\ref{fig:power_single} (c), we compare the effective $\alpha$-levels computed numerically with $m = \sqrt{n}$, including also the theoretically intractable Monte Carlo adjustment. These results confirm the Simes method becomes extremely conservative for large $n$, as its effective level tends to $0$ instead of $\alpha$. By contrast, the Monte Carlo adjustment yields approximately the same effective significance threshold as the asymptotic method.

Finally, it is interesting to compare these power analyses for calibration-conditional p-values with the exact adjustment of Fisher's combination test from Theorem~\ref{thm:fisher}. 
Under a regime in which $m = \gamma n$ for some $\gamma \in (0, 1)$, it follows from~\eqref{eq:conditional_typeI} that Fisher's combination test applied to marginal conformal p-values is valid at level $\alpha$, conditional on the calibration data, if its nominal significance level is lowered by a factor that depends on $\delta$---the proportion of calibration data sets for which the test is allowed to be invalid---but remains constant with respect to $n$. By contrast, applying Fisher's combination test to calibration-conditional p-values results in an effective level $\alpha$ that at best {\em decreases} as $1/\text{polylog}(n)$, for the asymptotic adjustment.
Therefore, calibration-conditional p-values are not always optimal with Fisher's combination test, at least not compared to the ad-hoc correction of the latter presented in Theorem~\ref{thm:fisher} when $m = \gamma n$, but they have the advantage of flexibility. In fact, calibration-conditional p-values can be utilized by any multiple testing algorithm, including for example the BH procedure, whose power analysis is discussed next.

\subsubsection{Testing multiple hypotheses by the BH procedure} \label{sec:power-BH}

Consider a multiple testing problem in which there are $m$ test data points $X_{2n+1}, \ldots, X_{2n+m}$ and the goal is to detect outliers with FDR control. If the BH procedure is applied to the adjusted p-values, all hypotheses with $h\circ \hat{u}^\marg(X_{2n+i})\le \alpha R(\alpha; h) / m$ are rejected, where
\[R(\alpha; h) = \max\left\{r \in \{0,1,\ldots,m\}: \#\left\{i: h\circ \hat{u}^\marg(X_{2n+i})\le \frac{r\alpha}{m}\right\}\ge r\right\}.\]
As a benchmark, we consider the number of rejections obtained with the marginal p-values:
\[R_{\mathrm{marg}}(\alpha) = \max\left\{r \in \{0,1,\ldots,m\}: \#\left\{i: \hat{u}^\marg(X_{2n+i})\le \frac{r\alpha}{m}\right\}\ge r\right\}.\]

In the case of the asymptotic adjustment, 
\begin{equation}\label{eq:asym_ratio}
\frac{h^\asym (i/(n+1))}{i/(n+1)} \le \frac{n+1}{n}\left\{1 + c_{n}(\delta)\sqrt{\frac{n - i}{ni}}\right\}.
\end{equation}
This quantity is decreasing in $i$, implying that
 \[\max_{i}\frac{h^\asym (i/(n+1))}{i/(n+1)}\le \frac{n+1}{n}\left\{1 + c_{n}(\delta)\sqrt{\frac{n - 1}{n}}\right\}= \sqrt{2\log \log n} + o(1).\]
Therefore, all hypotheses rejected by the BH procedure applied to marginal p-values at a lower level $\alpha / (\sqrt{2\log \log n} + o(1))$ would be guaranteed to be rejected by the BH procedure applied to adjusted p-values, implying the effective FDR level for $h^\asym$ is at least $\alpha / (\sqrt{2\log\log n} + o(1))$. If $\sqrt{2\log \log n} <\!\!< \log m$, this is more powerful than the Benjamini-Yekutieli procedure \citep{benjamini2001control}, whose effective FDR level is $\alpha / (\log m + O(1))$. Further, the ratio given by \eqref{eq:asym_ratio} is $1 + o(1)$ if $i / \log \log n \rightarrow \infty$, implying that, in the limit of $R_{\mathrm{marg}}(\alpha) / \log \log n\rightarrow \infty$, all marginal rejections are also rejected by the BH procedure applied to adjusted p-values with the target FDR level $\alpha(1 + o(1))$. In summary, the cost of the asymptotic adjustment never exceeds $\sqrt{2\log \log n} + o(1)$, and it is negligible if the number of rejections made by the marginal BH procedure grows faster than $\log \log n$.

In the case of the DKWM adjustment, the maximal ratio between the adjusted and marginal p-values is as large as $O(\sqrt{n})$, though the ratio becomes $1 + o(1)$ when $i / \sqrt{n}\rightarrow 0$. Thus, unless the marginal BH procedure can reject many more than $\sqrt{n}$ hypotheses, the power cost of the DKWM adjustment will be much higher than that of the asymptotic adjustment.

In the case of the Simes adjustment, we can show that, if $k = \lceil \zeta n\rceil$ for some $\zeta \in (0, 1)$, the ratio between the adjusted and marginal p-values is bounded by a constant that depends on $\delta$ and $\zeta$. Analogous to \eqref{eq:bi_simes}, the concavity of $a(x)$ implies
\[b_{i}^{\simes} \le 1 - \delta^{1/k}\exp\left\{\frac{1}{2}\left( a\left(\frac{n}{i - 1}\right) +a\left(\frac{n-k+1}{i - 1}\right) \right)\right\} = 1 - \delta^{1/k}\sqrt{\left(1 - \frac{i-1}{n}\right)\left(1 - \frac{i-1}{n-k+1}\right)}.\]
Since $k = \lceil\zeta n\rceil$, 
\begin{align*}
\sqrt{\left(1 - \frac{i-1}{n}\right)\left(1 - \frac{i-1}{n-k+1}\right)} &= \sqrt{\left(1 - \frac{i}{n}\right)\left(1 - \frac{i}{(1-\zeta)n}\right)} + o\left(\frac{1}{n}\right)\\
& = 1 - \frac{2-\zeta}{2(1 - \zeta)}\frac{i}{n} + o\left(\frac{1}{n}\right),
\end{align*}
and
\[\delta^{1/k} = \exp\left\{-\frac{\log(1/\delta)}{k}\right\} = 1 - \frac{\log(1/\delta)}{\zeta n} + o\left(\frac{1}{n}\right);\]
above, all $o(1/n)$ terms are uniform over $i$. 
Then, 
\[b_{i}^\simes \le \frac{\log(1/\delta)}{\zeta n}
 + \frac{2-\zeta}{2(1 - \zeta)}\frac{i}{n}
 + o\left(\frac{1}{n}\right),\]
 and for any $i$,
 \[\frac{h^\simes(i / (n + 1))}{i / (n + 1)} \le \frac{\log(1/\delta)}{\zeta}
 + \frac{2-\zeta}{2(1 - \zeta)} + o\left(\frac{1}{n}\right).\]
Thus, the power cost of the Simes adjustment does not grow with $n$, which is more appealing compared to the asymptotic adjustment in the worst case. However, \eqref{eq:simes_alpha} indicates the cost is never negligible even if $R_{\mathrm{marg}}(\alpha)$ is large, consistent with the behaviour of the Simes adjustment observed in Section~\ref{sec:power-single} for the case of a single hypothesis tested without multiplicity corrections. Thus, the asymptotic adjustment (and the substantially similar Monte Carlo approach) can be expected to be more powerful in practical applications involving FDR control, as long as a reasonably large number of discoveries is expected.

\section{Extensions beyond conformal p-values} \label{sec:extensions}

\subsection{Simultaneous confidence bounds for the false positive rate}

Some practitioners may be accustomed to thinking about outlier detection in terms of FPR---the probability of incorrectly reporting as outlier any true inlier---rather than p-values.
In particular, they may wonder what the FPR can be if they report $X_{2n+1}$ as likely to be an outlier whenever the classification score $\hat{s}(X_{2n+1})$ (computed by some black-box outlier detection algorithm) is below a threshold $t$, as a function of $t$, so that they may choose a posteriori which value of $t$ to adopt.
This question is closely related to the problem of constructing CCV p-values, so our method provides an answer.
In fact, the next result shows Theorem~\ref{thm:generic} also yields a simultaneous upper confidence bound for the CDF. 
\begin{prop}[Simultaneous confidence bounds for the FPR] \label{prop:ucb-fpr}
Let $F$ denote the true CDF of some distribution from which $n$ i.i.d.~samples, $Z_1, \ldots, Z_n$, are drawn, and denote by $\hat{F}_n$ the corresponding empirical CDF.
With the same notation as in Theorem~\ref{thm:generic},
\begin{equation}\label{eq:uniform_CDF}
\P\left[F(z)\le h(\hat{F}_{n}(z)), \,\,\forall z \in \R\right]\ge 1 - \delta.
\end{equation}
\end{prop}

Applying Proposition~\ref{prop:ucb-fpr} to the CDF of the scores $\hat{s}$ computed by any one-class classification algorithm provides a uniform upper confidence bound for its FPR, namely $\mathrm{FPR}(t) \defeq \P\left[\hat{s}(X_{2n+1}) \leq t \right]$, as a function of the detection threshold $t$. In other words, this guarantees that reporting as outliers an observation with black-box score equal to $z$ is likely (with probability at least $1-\delta$) to result in a FPR no greater than $h(\hat{F}_n(z))$, where $\hat{F}_n(z)$ is the empirical CDF of the analogous scores computed on a calibration data set of size $n$.
Figure~\ref{fig:bands} shows a practical example of this upper bound based on the empirical distribution of scores evaluated on 1000 calibration points, with $\delta=0.1$ and $k=n/2$ (the exact details of this example are the same as those of the numerical experiments presented later in Section~\ref{sec:exp_simulated}). For instance, this plot informs us that reporting as outliers future samples with scores below -0.5 is likely to result in an FPR below 0.025.

\begin{figure}[!htb]
  \centering
  { \includegraphics[width=0.8\textwidth]{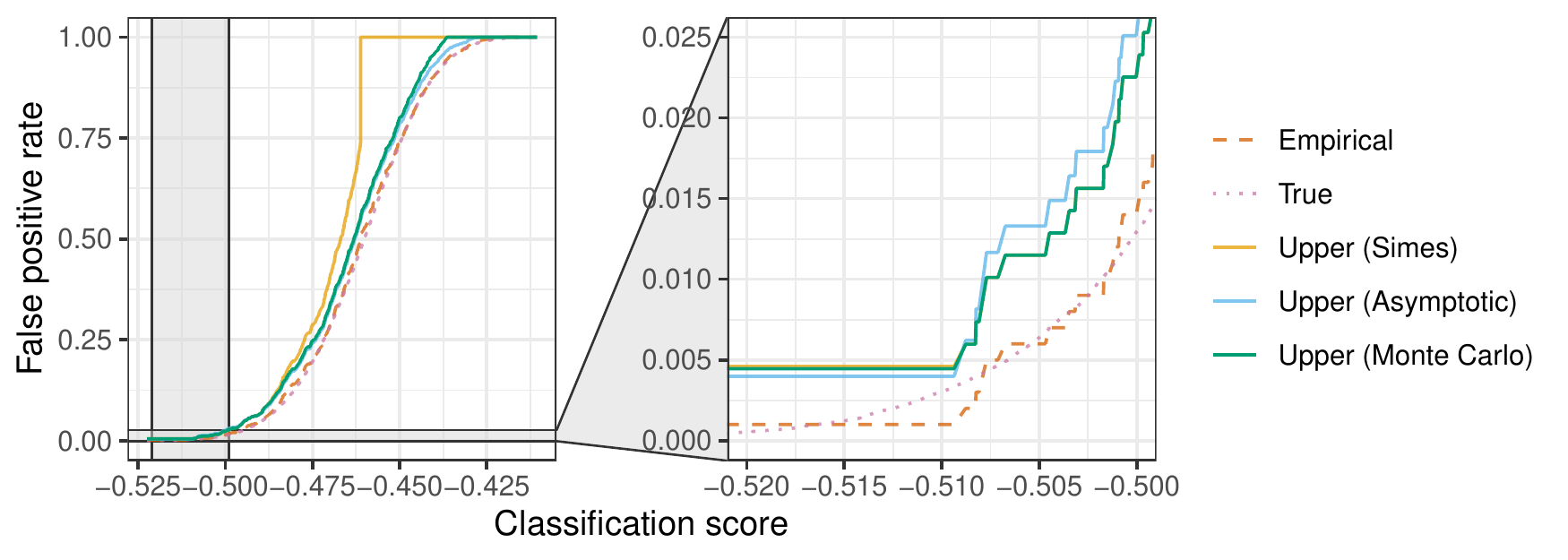}}
\caption{FPR calibration curves obtained with different adjustment methods for an isolation forest one-class classifier on simulated data, as a function of the reporting threshold for the classification scores. Each upper bound (solid) is guaranteed to lie above the true FPR curve (dotted) with probability 90\%. The dashed curve corresponds to the empirical FPR. The panel on {the right} zooms in on small values (likely outliers).
}
  \label{fig:bands}
\end{figure}

Note that the construction of a uniform confidence band for an unknown CDF is a widely studied problem.
For example, the DKWM inequality~\citep{dvoretzky1956asymptotic, massart1990tight} implies the bound in~\eqref{eq:uniform_CDF} with $h(z) = \min\{z + \sqrt{\log(2 / \delta) / 2n}, 1\}$. However, the DKWM bound is tightest at $z=1/2$ and loose near 0, which would limit the power to detect outliers. Therefore, it is preferable for our purposes to have a function $h(z)$ {that is} as close as possible to the identity for small values of $z$, as discussed earlier in Section~\ref{sec:gen-simes}.

\subsection{Simultaneously-valid prediction sets} \label{sec:predictive-sets}
Lastly, CCV p-values can be easily re-purposed to strengthen the marginal guarantees generally obtainable for conformal predictions.
In particular, for each $\alpha \in (0,1)$, one can define a predictive set
\begin{equation} \label{eq:simult_pred_sets_def}
    \Chat^\alpha \defeq \{x : \hat{u}^\cond(x) > \alpha\}.
\end{equation}
These sets are simultaneously valid for all $\alpha$, conditional on the calibration data. That is, they satisfy
\begin{equation}
\label{eq:simult_sets_def}
    \P\bigg[ \P\big[X_{2n+1} \in \hat{\C}^\alpha \mid \mathcal{D} \big] \ge 1 - \alpha \text{ for all } \alpha \in (0,1) \bigg] \ge 1 - \delta.
\end{equation}
In other words, if we use CCV p-values to construct prediction sets, the probability that a new observation falls within $\Chat^\alpha$ is at least $1-\alpha$, simultaneously for all $\alpha \in (0,1)$ with high probability. This is stronger than the usual conformal guarantee, as the latter holds marginally over $\mathcal{D}$ and only for a single pre-specified~$\alpha$.

\section{Numerical experiments} \label{sec:exp}

\subsection{Setup} \label{sec:exp-setup}

The following experiments are designed to simulate a world in which our methods are {independently applied by $J$ practitioners}. Each practitioner $j \in [J]$ has an independent data set $\mathcal{D}_j$ (to train and calibrate the method), and $L$ test sets $\mathcal{D}^{\mathrm{test}}_{j,l}$ (to compute p-values and evaluate performance), each corresponding to different possible future scenarios $l \in [L]$.
The data sets contain $2n$ observations each ($|\mathcal{D}_j|=2n$), and the test sets contain $n_{\text{test}}$ observations each ($|\mathcal{D}^{\mathrm{test}}_{j,l}|=n_{\text{test}}$).
Imagine that, from the practitioner's present point of view, the data set $\mathcal{D}_j$ is fixed but the test set is random, so that $\mathcal{D}^{\mathrm{test}}_{j,l}$ represents the test set for practitioner $j$ under future scenario $l$. 
Then, as discussed in Section~\ref{subsec:preview}, practitioner $j$ {is most interested in the FDR} (or other measures of type-I errors, alternatively) conditional on $\mathcal{D}_j${, i.e., in the random variable}
\begin{align*}
\text{cFDR}(\mathcal{D}_j) \defeq \E\left[ \text{FDP}(\mathcal{D}^{\mathrm{test}}; \mathcal{D}_j) \mid \mathcal{D}_j \right],
\end{align*}
where $\text{FDP}(\mathcal{D}^{\mathrm{test}}; \mathcal{D}_j)$ is the proportion of inliers among the test points reported as outliers, based on the procedure calibrated on $\mathcal{D}_j$.
This motivates the definition of the following performance measures. For any $j \in [J]$, we compute
\begin{align} \label{eq:fdr-pow-sim}
    & \widehat{\text{cFDR}}(\mathcal{D}_j) \defeq \frac{1}{L} \sum_{l=1}^{L} \text{FDP}(\mathcal{D}^{\mathrm{test}}_{j,l}; \mathcal{D}_j),
    & \widehat{\text{cPower}}(\mathcal{D}_j)
    \defeq \frac{1}{L} \sum_{l=1}^{L} \text{Power}(\mathcal{D}^{\mathrm{test}}_{j,l}; \mathcal{D}_j),
\end{align}
where $\text{Power}(\mathcal{D}^{\mathrm{test}}_{j,l}; \mathcal{D}_j)$ is the proportion of outliers in $\mathcal{D}^{\mathrm{test}}_{j,l}$ correctly identified as such by practitioner $j$.

Our experiments will demonstrate that the proposed simultaneous calibration method leads to sufficiently small $\widehat{\text{cFDR}}(\mathcal{D}_j)$ for the desired fraction of practitioners, while the traditional point-wise calibration generally only leads to small values of the marginal FDR, namely $\widehat{\text{mFDR}} \defeq \frac{1}{J} \sum_{j=1}^{J} \widehat{\text{cFDR}}(\mathcal{D}_j)$.

\subsection{Outlier detection on simulated data} \label{sec:exp_simulated}

\subsubsection{Data description}

We begin to investigate the empirical performance of different methods for calibrating conformal p-values on synthetic data. 
The data are generated by sampling each data point $X_i \in \mathbb{R}^{50}$ from a multivariate Gaussian mixture model $P_X^{a}$, such that $X_i = \sqrt{a} \, V_i + W_i$, for some constant $a\geq 1$ and appropriate random vectors $V_i,W_i \in \mathbb{R}^{50}$.
Here, $V_i$ has independent standard Gaussian components, and each coordinate of $W_i$ is independent and uniformly distributed on a discrete set $\mathcal{W} \subseteq \mathbb{R}^{50}$ with cardinality $|\mathcal{W}|=50$.
The vectors in $\mathcal{W}$ are sampled independently from the uniform distribution on $[-3,3]^{50}$, before the beginning of our experiments, and then held constant thereafter. (Therefore, each coordinate of $W_i$ is uniformly distributed on $[-3,3]$, but it is not the case that the different $W_i$'s are independent and identically distributed on $[-3,3]^{50}$; instead, the fixed set $\mathcal{W}$ makes this a mixture model.)

The data sets $\mathcal{D}_j$ are sampled from $P_X^{a}$ with $a=1$ and $n=1000$.
The total $2n$ observations in each $\mathcal{D}_j$ are further divided into $n_{\text{train}}=1000$ observations used to fit a one-class SVM classifier scoring function $\hat{s}$ (implemented in the Python package \texttt{scikit-learn}~\cite{scikit-learn}), and $n_{\text{cal}}=1000$ observations used to calibrate the conformal p-values, as in~\eqref{eq:pval_form}, leading to a valid p-value $\hat{u}(X_{n+1}) \in [0,1]$ for any new data point $X_{n+1}$.
The total number of data sets is $J=100$, each of which is associated with $L=100$ test sets.
A random subset of the observations in each test set $\mathcal{D}^{\mathrm{test}}_{j,l}$ is sampled from $P_X^{a}$ with $a=1$, while the others are outliers, in the sense that they are sampled from $P_X^{a}$ with $a>1$, as specified below.

\subsubsection{Individual outlier detection}
\label{sec:simulated_individual_outlier_detection}

First, we focus on a data generating model under which $90\%$ of the $n_{\text{test}}=1000$ observations in each $\mathcal{D}^{\mathrm{test}}_{j,l}$ are sampled from $P_X^{a}$ with $a=1$, and we seek to identify the remaining 10\% of outliers. 
For this purpose, we calibrate a conformal p-value for all  observations in  $\mathcal{D}^{\mathrm{test}}_{j,l}$, and then we apply the BH procedure at some nominal FDR level $\alpha$ to account for the multiple comparisons,  with and without Storey's correction based on the estimated null proportion. In the following, we apply our conditional calibration method with the parameters $\delta=0.1$ and  $k=n_{\text{cal}}/2$ (see below for comments about the choice of $k$).

Figure~\ref{fig:sim-bh} shows the distribution of $\widehat{\text{cFDR}}(\mathcal{D}_s)$ and $\widehat{\text{cPower}}(\mathcal{D}_s)$, corresponding to $\alpha=0.1$, for different values of the signal strength $a$ (recall that here $a=1$ corresponds to no signal), when the BH procedure is utilized to account for the multiple comparisons.
The results confirm the calibration-conditional p-values control the conditional FDR for at least 90\% of practitioners, while the marginal p-values do not. In fact, marginal p-values only control the conditional FDR if the number of samples in the calibration data set is very large; see Figure~\ref{fig:sim-bh-n}, Appendix~\ref{app:sim}.
Among the three conditional calibration alternatives, the Monte Carlo and Simes methods yield slightly higher power than the asymptotic approximation in this setting.
Note that all methods control the marginal FDR, as also predicted by our theoretical results.
Figure~\ref{fig:sim-bh-storey} presents the results obtained by applying Storeys' correction to the BH procedure, while Figure~\ref{fig:sim-bh-delta0.25} summarizes additional experiments in which the conditional calibration is applied with $\delta=0.25$. Finally, Figure~\ref{fig:sim-simes-tune} visualizes the effect of different values of the $k$ on the conditional p-values calibrated with the Simes method, showing that $k=n_{\text{cal}}/2$ works relatively well, although the performance does not appear to be extremely sensitive to this choice.

\begin{figure}[!htb]
  \centering
  { \includegraphics[width=\textwidth]{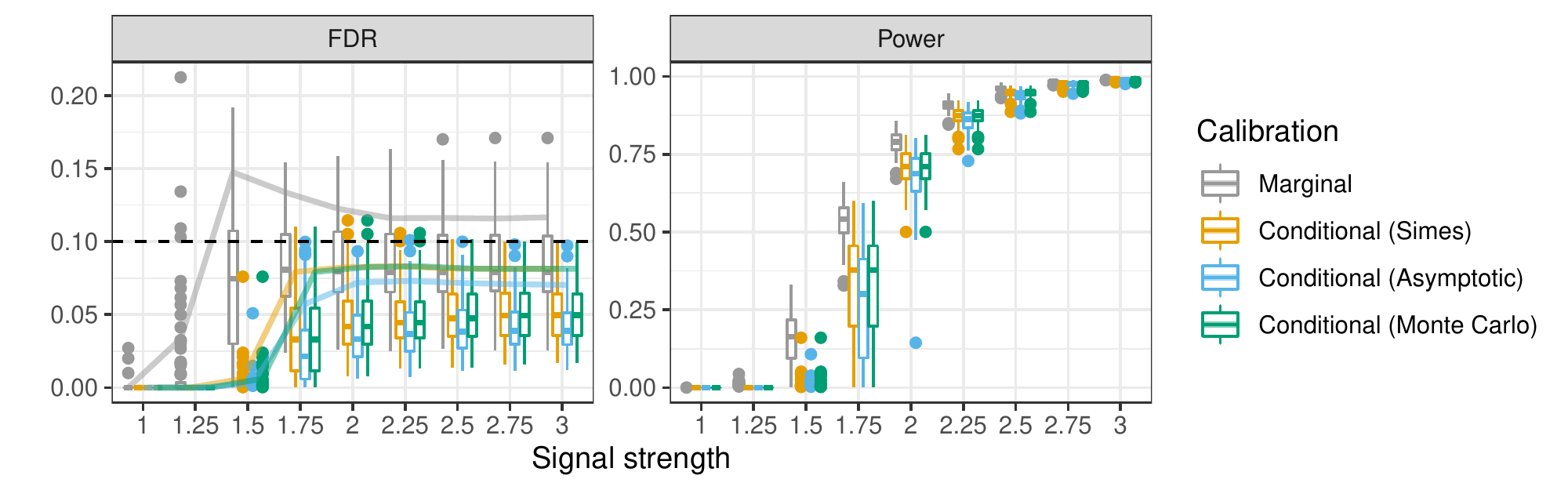}}
\caption{Performance of different methods for calibrating conformal p-values in a simulated outlier detection problem, as a function of the signal strength.
The box plots visualize the distribution of FDR and power, as defined in~\eqref{eq:fdr-pow-sim}, conditional on 100 independent data sets. The solid curves indicate the 90-th quantile of the conditional FDR distribution. The nominal FDR $0.1$, and the conditional method is applied with $\delta=0.1$. 
}
  \label{fig:sim-bh}
\end{figure}

\subsubsection{Batch outlier detection} \label{sec:simulated_batch_outlier_detection}

We now consider the global testing problem of detecting whether a batch of new observations contains any outliers.
For this purpose, we follow the same approach as before, with the only difference that the $n_{\text{test}} = 1000$ observations in each test set are {now} sub-divided into 100 batches of size 10. 
The 10 calibrated p-values in each batch are combined with Fisher's method to test the batch-specific global null. Then, the BH procedure with Storey's correction is applied to control the FDR over {all} batches. This simulation is designed such that 90\% of the batches contain no outliers (i.e., all samples are drawn from $P_X^{a}$ with $a=1$), while 50\% of the samples in the remaining batches are outliers (i.e., they are drawn from $P_X^{a}$ with $a=1.75$).
Of course, batched testing is less informative than the precise identification of outliers discussed in the previous section, but the advantage now is that we {can achieve} higher power.
Figure~\ref{fig:global-storey} shows that, even though this problem is relatively easy (the power is close to 1), the use of marginal p-values may still lead to a conditional FDR that is noticeably higher than expected for many researchers.
By contrast, simultaneous calibration appears to be conservative for all of them, without much power loss. 
Among the three conditional calibration alternatives, the Monte Carlo method and the asymptotic approximation yield higher power in this setting.

\begin{figure}[!htb]
  \centering
  { \includegraphics[width=0.8\textwidth]{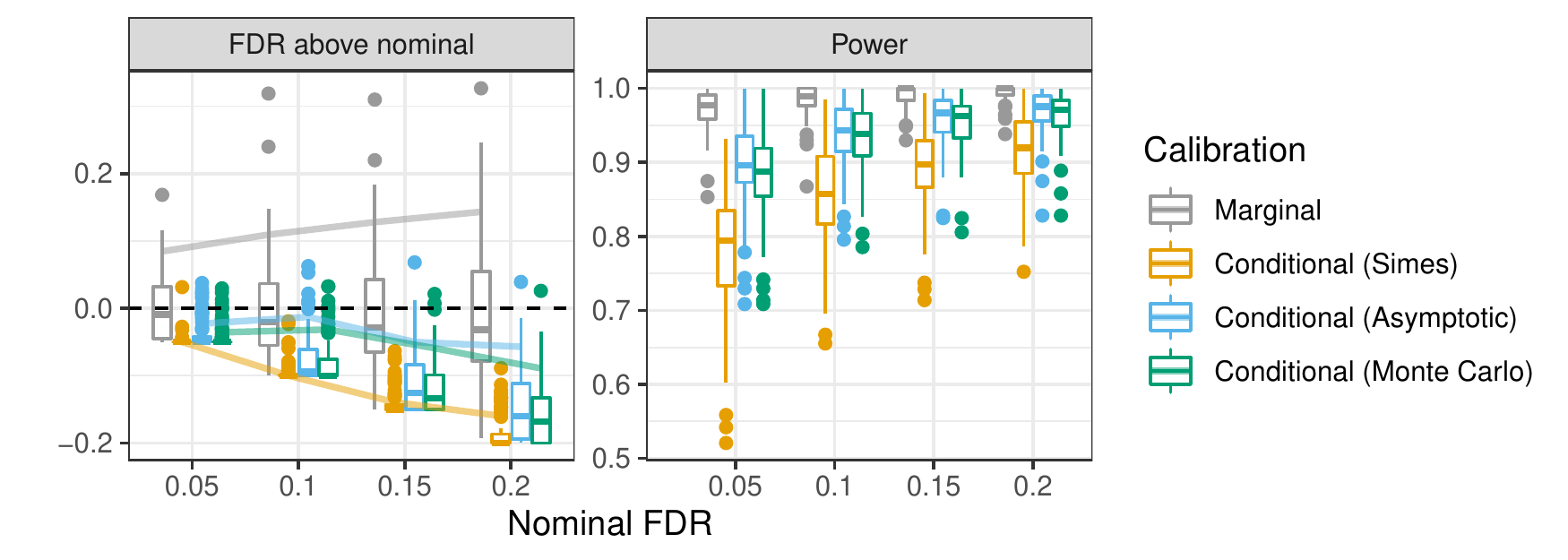}}
\caption{Performance of different methods for calibrating conformal p-values in a simulated outlier batch detection problem, as a function of the nominal FDR level. 
The excess FDR is defined as the difference between the empirical FDR and the nominal FDR. Other details are as in Figure~\ref{fig:sim-bh}.
}
  \label{fig:global-storey}
\end{figure}

Next, we study the effect of the batch size on the performance of different calibration methods under the global null hypothesis (i.e., when there are no outliers in the test set). As before, the p-values in each batch are combined with Fisher's method and the global null is rejected if the resulting p-value is smaller than 0.1. As before, the experiment is repeated for 100 independent data sets and 1000 test sets. 
Figure~\ref{fig:global-storey-fwer} shows that marginal p-values do not lead to valid inferences, especially if the batch size is large. By contrast, the tests based on calibration-conditional p-values always remain valid.

\begin{figure}[!htb]
  \centering
  { \includegraphics[width=0.6\textwidth]{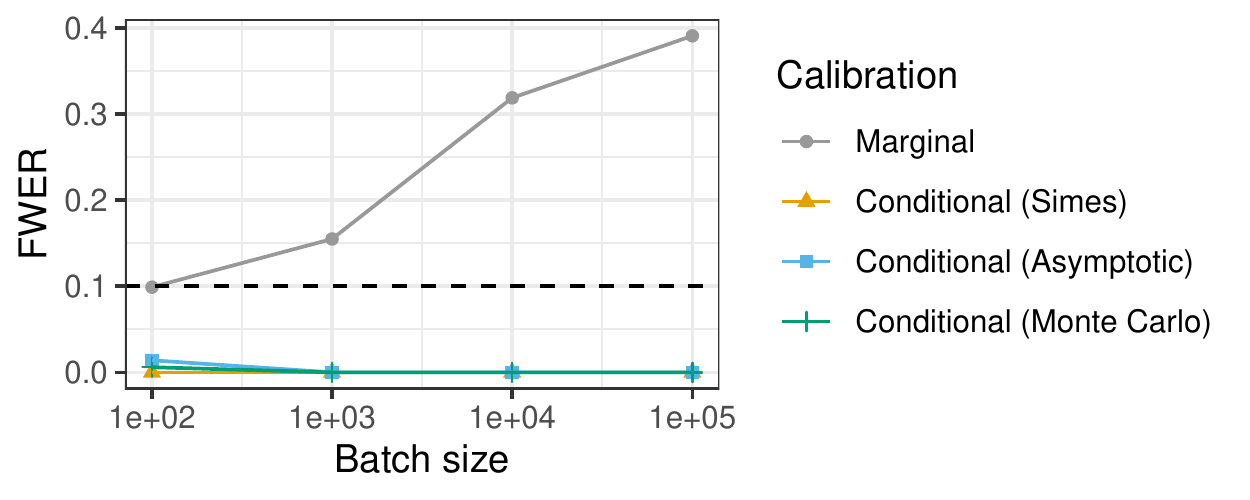}}
\caption{Family-wise error rate (FWER) in a simulated outlier batch detection problem under the global null hypothesis, using different calibration methods for the conformal p-values. The results are shown as a function of the batch size. The global null is rejected if the Fisher's combined p-value is below 0.1, which means the nominal FWER is 10\% (horizontal dashed line).
}
  \label{fig:global-storey-fwer}
\end{figure}

Finally, Figure~\ref{fig:global-combinations} compares the performances of different global testing methods for combining the p-values in each batch (in addition to Fisher's combination test), in the same experiments as in Figure~\ref{fig:global-storey}. The alternative combinations we consider are the harmonic mean with equal weights~\cite{wilson2019harmonic}, Simes'~\cite{simes1986improved}, and Stouffer's~\cite{stouffer1949american} p-values. The results show that the harmonic mean and Simes' p-values yield no discoveries. This should be unsurprising because those methods are designed to have power against an alternative in which the signals are few and strong (e.g., a single outlier in each non-null batch), which is not the case here because each non-null batch contains several outliers and marginal conformal p-values can never be smaller than $1/n$. Fisher's marginal conformal p-values appear to be more powerful than Stouffer's in these experiments, even if the former are simultaneously adjusted with our Monte Carlo method and the latter are not.
It is worth emphasizing that, unlike Fisher's combination test, not all global testing methods may become invalid on average when applied to positively dependent p-values. For example, the harmonic mean~\cite{wilson2019harmonic} and Simes's p-values are known to be robust to positive dependencies \cite{sarkar1997simes}, and Stouffer's combination p-value can also be modified to account for known dependencies~\cite{strube1985combining}. 
Yet, our simultaneous adjustment for conformal p-values remains useful even with combination tests that are robust to positive dependency because this adjustment happens to be necessary to guarantee valid inferences conditional on the calibration data; see Figure~\ref{fig:global-combinations}.

\subsection{Outlier detection on real data} \label{sec:exp_real}

\subsubsection{Data description}

\begin{table}[!htb]
\centering
\small
\caption{\label{tab:datasets}Summary of the benchmark data sets for outlier detection utilized in our applications.}
\begin{tabular}{@{}lccccccc@{}}
\toprule
                  & ALOI  & Cover  & Credit card & KDDCup99 & Mammography & Digits & Shuttle \\
                         & \cite{campos2016evaluation,aloi}   & \cite{cover}    & \cite{creditcard}         & \cite{campos2016evaluation,KDDCup99}      & \cite{mammography}         & \cite{pendigits}    & \cite{shuttle}     \\ \midrule
Features $d$      & 27    & 10     & 30          & 40       & 6           & 16     & 9       \\
Inliers ${n_\text{inliers}}$   & 283301 & 286048 & 284315      & 47913    & 10923       & 6714   & 45586   \\
Outliers ${n_\text{outliers}}$ & 1508  & 2747   & 492         & 200      & 260         & 156    & 3511   \\ \bottomrule
\end{tabular}
\end{table}

We turn to study the performance of the calibration schemes from Section~\ref{sec:exp_simulated} on several benchmark data sets for outlier detection, summarized in Table~\ref{tab:datasets}. The conditional p-values are calibrated with $\delta = 0.1$ using the Monte Carlo method, which is valid in finite samples and has demonstrated in the previous sections to be more powerful than other two alternatives.
We utilize an isolation forest~\cite{liu2008isolation} machine-learning algorithms $\hat{s}$ as the base method for detecting anomalies, available in the Python \texttt{sklearn} package. We rely on the default hyper-parameters, except for the `contamination' parameter which we set equal to $0.1$. Additional experiments based on one-class SVM and Local Outlier Factor (LOF) algorithms are presented in Appendix~\ref{app:sim} (Tables~\ref{tab:data-bh-long}--\ref{tab:data-global-long-storey}).

\subsubsection{Individual outlier detection}\label{sec:real_individual_outlier_detection}

Here, we follow the experimental setup of Section~\ref{sec:simulated_individual_outlier_detection}.
The difference is that we need to construct multiple training, calibration, and test sets by randomly splitting the $n_{\text{inlier}}$ inlier examples into three disjoint subsets of size $n_{\text{train}}$, $n_{\text{cal}}$ and $n_{\text{test}}$, respectively. 
A total of $n_{\text{inlier}}/2 $ data points is used for training and calibration, i.e., $n_{\text{train}} + n_{\text{cal}} = n_{\text{inlier}}/2 $ with $n_{
\text{cal}} = \min\{2000, n_{\text{train}}/2\}$, while outlier examples are only included in the test sets. 
For each training/calibration data subset, we sample 100 test sets of size $n_{\text{test}} = \min\{2000, n_{\text{train}}/3\}$. Each test set contains 90\% of randomly chosen inliers, and 10\% of outliers. 
It should be noted that, in contrast to the simulated experiments of Section~\ref{sec:simulated_individual_outlier_detection} in which the data were effectively infinitely abundant, {there is} some overlap between the samples in different test sets.


Figure~\ref{fig:creditcard-bh-storey} compares the performance of marginal and simultaneously calibrated p-values on the credit card data set~\cite{creditcard}, as a function of the nominal FDR level. 
Here, the BH procedure is applied with Storey's correction.
Note that the proposed Monte Carlo simultaneous calibration leads to FDR control for at least 90\% of simulated practitioners, as expected. This stands in contrast with the marginal calibration approach, which controls the FDR only marginally.

\begin{figure}[!htb]
  \centering
  { \includegraphics[width=0.9\textwidth]{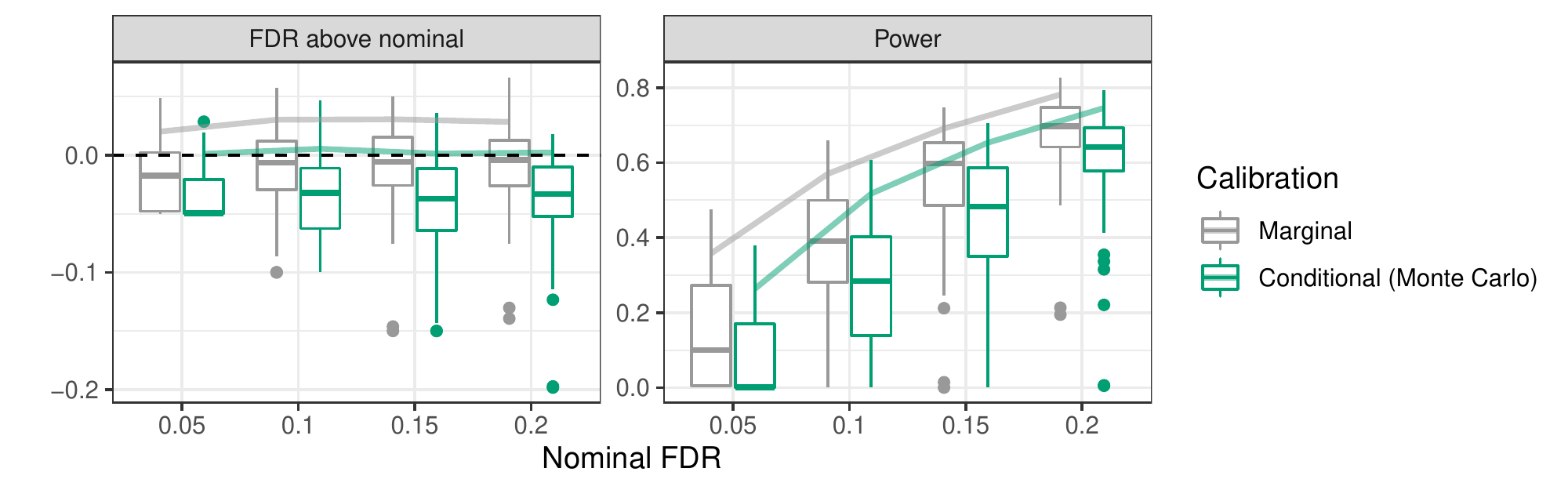}}
\caption{Outlier detection performance on credit card fraud data. Conformal p-values based on an isolation forest model are calibrated using different methods. The Benjamini-Hochberg procedure with Storey's correction is then applied to control the FDR over the set of test points. The results are shown as a function of the nominal FDR level. Other details are as in Figure~\ref{fig:sim-bh}.
}
  \label{fig:creditcard-bh-storey}
\end{figure}

Consistent conclusion can be drawn from Table~\ref{tab:data-bh-storey}, which compares the two calibration procedures on all benchmark data sets at the nominal FDR level of 0.2.
Additional results corresponding to different outlier detection algorithms (one-class SVM and LOF) can be found in Table~\ref{tab:data-bh-long-storey}, Appendix~\ref{app:real_data_exp}. In all cases, we adopt the \texttt{sklearn} default parameters. Finally, Table~\ref{tab:data-bh-long} summarizes the performance of different calibration and detection methods across all data sets when the BH procedure is applied without Storey's correction.

{\small
\begin{table}[!htb]

\caption{Outlier detection performance on different data sets, using alternative methods for calibrating conformal p-values. 
The FDR and power diagnostics are defined conditional on the training and calibration data, as explained in Section~\ref{sec:exp-setup}. The nominal marginal FDR level is 0.2. Empirical FDR values larger than the nominal level are colored in orange; values at least one standard deviation above it are colored in red. \label{tab:data-bh-storey}}
\centering
\fontsize{10}{12}\selectfont
\begin{tabular}[t]{lllllllll}
\toprule
\multicolumn{1}{c}{ } & \multicolumn{4}{c}{FDR} & \multicolumn{4}{c}{Power} \\
\cmidrule(l{3pt}r{3pt}){2-5} \cmidrule(l{3pt}r{3pt}){6-9}
\multicolumn{1}{c}{ } & \multicolumn{2}{c}{Mean} & \multicolumn{2}{c}{90th percentile} & \multicolumn{2}{c}{Mean} & \multicolumn{2}{c}{90-th quantile} \\
\cmidrule(l{3pt}r{3pt}){2-3} \cmidrule(l{3pt}r{3pt}){4-5} \cmidrule(l{3pt}r{3pt}){6-7} \cmidrule(l{3pt}r{3pt}){8-9}
Dataset & Marg. & Cond. & Marg. & Cond. & Marg. & Cond. & Marg. & Cond.\\
\midrule
ALOI & 0.025 & 0.001 & 0.048 & 0 & 0 & 0 & 0 & 0\\
Cover & 0.099 & 0.044 & \textcolor{red}{0.297} & 0.148 & 0.012 & 0.006 & 0.038 & 0.02\\
Credit card & 0.191 & 0.162 & \textcolor{orange}{0.228} & \textcolor{orange}{0.202} & 0.679 & 0.611 & 0.782 & 0.746\\
KDDCup99 & 0.194 & 0.131 & \textcolor{orange}{0.23} & 0.168 & 0.754 & 0.684 & 0.825 & 0.753\\
Mammography & 0.187 & 0.056 & \textcolor{red}{0.286} & 0.17 & 0.176 & 0.059 & 0.337 & 0.22\\
\addlinespace
Digits & \textcolor{orange}{0.202} & 0.052 & \textcolor{red}{0.266} & 0.173 & 0.417 & 0.096 & 0.629 & 0.355\\
Shuttle & 0.196 & 0.163 & \textcolor{orange}{0.228} & 0.198 & 0.981 & 0.98 & 0.984 & 0.983\\
\bottomrule
\end{tabular}
\end{table}
}

\subsubsection{Batch outlier detection} \label{sec:real_batch_outlier_detection}

We now focus on global testing for outlier batch detection, similarly to Section~\ref{sec:simulated_batch_outlier_detection}.
The available data are divided into training, calibration, and test sets according to the same scheme as in Section~\ref{sec:real_individual_outlier_detection}; the only difference is that the size of the test sets is now equal to 1000, so as to follow as closely as possible the same experimental protocol as in Section~\ref{sec:simulated_batch_outlier_detection}.

Figure~\ref{fig:creditcard-global} compares the performance of the different calibration methods as a function of the nominal FDR level. 
The p-values in each batch are combined with Fisher's method, and then the  BH procedure is applied with Storey's correction.
Again, we observe that simultaneous calibration is required to ensure the conditional FDR is controlled in at least 90\% of the applications, although it involves some power loss. Both calibration methods control the marginal FDR.

\begin{figure}[!htb]
  \centering
  { \includegraphics[width=0.9\textwidth]{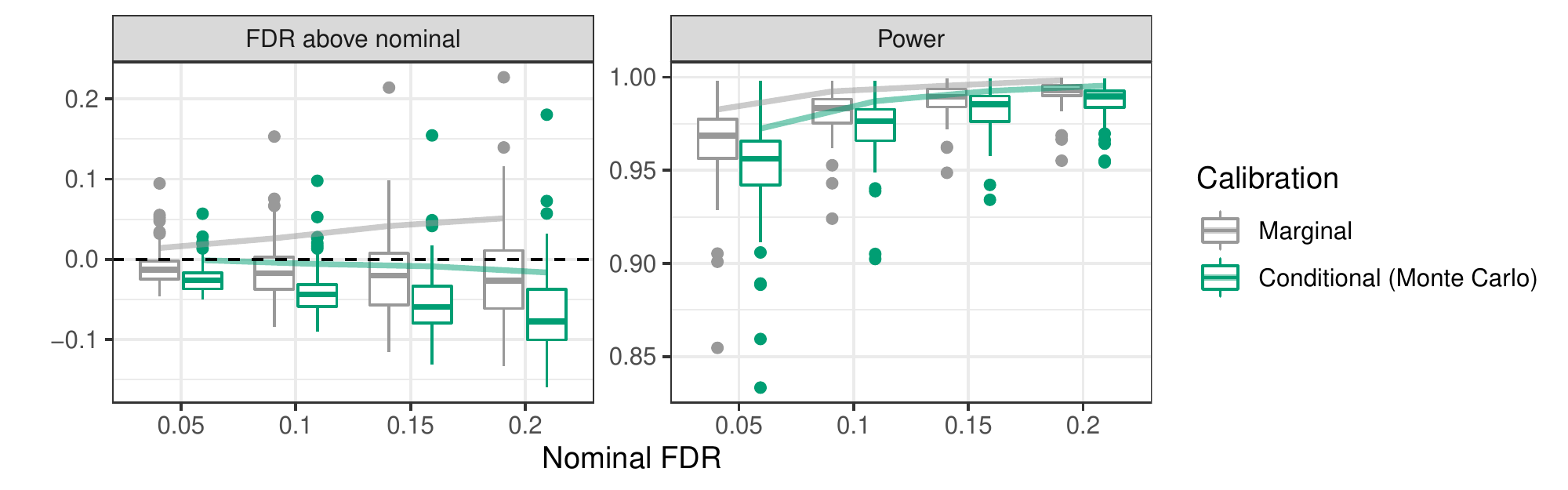}}
\caption{Outlier batch detection performance on credit card fraud data. Conformal p-values are computed based on an isolation forest model and calibrated using different methods. Other details are as in Figure~\ref{fig:global-storey}.
}
  \label{fig:creditcard-global}
\end{figure}

 Table~\ref{tab:data-global-storey} summarizes the performance of the two alternative calibration methods on all data sets.
 Here, the nominal FDR level is 0.1 and the BH procedure is applied with the Storey correction. Again, the results show that the Monte Carlo method controls the conditional FDR 90\% of the time, although at some cost in power, while the marginal calibration method does not.
  See Table~\ref{tab:data-global-long-storey}, Appendix~\ref{app:real_data_exp} for additional results that, in addition to the isolation forest, include also the one-class SVM and LOF algorithms for outlier detection. Finally, Table~\ref{tab:data-global-long} summarizes performance of the different methods on all data sets when the BH procedure is applied without the Storey correction.

{\small
\begin{table}

\caption{Outlier batch detection performance on different data sets, using alternative methods for calibrating conformal p-values. The nominal FDR level is 0.1. Other details are as in Table~\ref{tab:data-bh-storey}. \label{tab:data-global-storey}}
\centering
\fontsize{10}{12}\selectfont
\begin{tabular}[t]{lllllllll}
\toprule
\multicolumn{1}{c}{ } & \multicolumn{4}{c}{FDR} & \multicolumn{4}{c}{Power} \\
\cmidrule(l{3pt}r{3pt}){2-5} \cmidrule(l{3pt}r{3pt}){6-9}
\multicolumn{1}{c}{ } & \multicolumn{2}{c}{Mean} & \multicolumn{2}{c}{90-th quantile} & \multicolumn{2}{c}{Mean} & \multicolumn{2}{c}{90-th quantile} \\
\cmidrule(l{3pt}r{3pt}){2-3} \cmidrule(l{3pt}r{3pt}){4-5} \cmidrule(l{3pt}r{3pt}){6-7} \cmidrule(l{3pt}r{3pt}){8-9}
Data set & Marg. & Cond. & Marg. & Cond. & Marg. & Cond. & Marg. & Cond.\\
\midrule
ALOI & 0.07 & 0.016 & \textcolor{red}{0.2} & 0.081 & 0.001 & 0 & 0.004 & 0.002\\
Cover & 0.08 & 0.05 & \textcolor{red}{0.158} & \textcolor{orange}{0.12} & 0.184 & 0.132 & 0.333 & 0.243\\
Credit card & 0.086 & 0.059 & \textcolor{orange}{0.126} & 0.094 & 0.981 & 0.973 & 0.992 & 0.987\\
KDDCup99 & 0.091 & 0.044 & \textcolor{red}{0.145} & 0.08 & 1 & 0.998 & 1 & 1\\
Mammography & 0.069 & 0.014 & \textcolor{orange}{0.116} & 0.03 & 0.599 & 0.334 & 0.742 & 0.521\\
Digits & 0.084 & 0.016 & \textcolor{red}{0.141} & 0.033 & 0.968 & 0.814 & 0.995 & 0.926\\
Shuttle & 0.094 & 0.047 & \textcolor{red}{0.142} & 0.087 & 1 & 1 & 1 & 1\\
\bottomrule
\end{tabular}
\end{table}
}

\section{Discussion} \label{sec:discussion}

This paper has studied the multiple testing problem for outlier detection using conformal p-values.
Conformal p-values provide a natural approach to outlier detection (when clean training data are available) with the advantage of being able to leverage any black-box machine-learning tool, producing fully non-parametric inferences that are provably valid in finite samples and require no modeling beyond the i.i.d.~assumption. Of course, a possible limitation (or perhaps strength, depending on the viewpoint) of conformal inference is that its agnosticism prevents very confident statements, as conformal p-values can never be smaller than $1/(n+1)$, where $n$ is the number of clean data points available for calibration. Therefore, this solution may not be as powerful as likelihood-based approaches, especially if the signals are strong but sparse.
However, it does seem preferable if clean data are available but accurate models are not.

Whenever the conformal framework is appropriate for a particular outlier detection application, the problem of multiple testing considered in this paper is likely to be relevant, as possible outliers are often to be detected among many possible inlier test points, and reporting an excess of false discoveries would be undesirable.
Our work brings attention to the delicacy of such task, showing that the mutual dependence of conformal p-values breaks certain methods (e.g., Fisher's combination test) and makes the validity of others (e.g., the BH procedure) not obvious.
In particular, we find our PRDS result interesting because this property is well-known as a theoretical assumption for FDR control, but it is typically difficult to verify in practical applications~\cite{benjamini2001control, clarke2009robustness}.

Our methodological contribution is a technique based on high-probability bounds to compute calibration-conditional conformal p-values that are mutually independent and can thus be directly trusted in any multiple testing procedure. Our bounds are stronger than those in the previous conformal inference literature because they are simultaneous in nature and, consequently, they can also be useful for practitioners to tune a posteriori the significance threshold for machine-learning statistics above which to report their discoveries.
Unsurprisingly, our simulations demonstrate that calibration-conditional inferences are less powerful on average than marginal conformal inferences; therefore, the additional comfort of their stronger guarantees should be weighted against the potential loss of some interesting findings. Nonetheless, we prefer to leave such considerations to practitioners on a case-by-case basis, as our objective here was simply to explain the theoretical properties and general relative advantages of different statistical methods. 

Finally, this work opens new directions for future research. For example, focusing on split-conformal p-values, we did not study other hold-out approaches, such as the jackknife+~\cite{barber2019predictive} or bootstrap sampling~\cite{kim2020predictive}, that may practically yield higher power, although they are also more {computationally expensive}. A separate line of research may focus on relaxing the i.i.d.~assumption to improve power in a multiple testing setting with structured outliers~\cite{li2019multiple}. In fact, our theory only requires the calibration and test inliers to be exchangeable and mutually independent, while the outliers in the test data may have dependencies with one another.
Furthermore, we mentioned but did not explore the possible connection between our multiple outlier testing problem (especially regarding our results on Fisher's combination method) and classical two-sample testing. Finally, the high-probability bounds developed here may prove useful for purposes other than the calibration of conformal p-values; for instance, we already discussed a straightforward extension to obtain simultaneously valid prediction sets, but other possible applications may involve predictive distributions~\cite{vovk2018cross} and functionals thereof~\cite{wisniewski2020application}, or the comparison of different machine-learning algorithms in terms of estimated generalization error~\cite{holland2020making,Bayle2020}, for example.

\section*{Software availability}

A software implementation of the methods described in this paper is available online, in the form of a Python package, at
\url{https://github.com/msesia/conditional-conformal-pvalues.git}, along with usage examples and notebooks to reproduce our numerical experiments.

\section*{Acknowledgements}

S.B.~gratefully acknowledges the support of the Ric Weiland fellowship. E.C.~was supported by Office of Naval Research grant N00014-20-12157, by the National Science Foundation grants OAC 1934578 and DMS 2032014, by the Army Research Office (ARO) under grant W911NF-17-1-0304, and by the Simons Foundation under award 814641.
 L.L.~gratefully acknowledges the support of the National Science Foundation grants OAC 1934578, the Discovery Innovation Fund for Biomedical Data Sciences, and the NIH grant R01MH113078.  Y.R. was supported by the ISRAEL SCIENCE FOUNDATION (grant No. 729/21) and by the Career Advancement Fellowship of the Technion.
We are grateful to the anonymous referees and associate editor for their helpful comments and suggestions.

\printbibliography

\appendix

\renewcommand{\thefigure}{A\arabic{figure}}
\setcounter{figure}{0}
\renewcommand{\thetable}{A\arabic{table}}
\setcounter{table}{0}

\section{Technical proofs} \label{app:proofs}

\subsection{Correlation structure of null marginal conformal p-values} \label{app:correlation-structure}

For notational convenience, we write $p_i$ instead of $\hat{u}^\marg(X_{2n+i})$. When $X_{2n+1}, \ldots, X_{2n+m}$ are all inliers which are drawn from $P_{X}$, the conformal p-values $p_1, \ldots, p_m$ are exchangeable. 
Lemma~\ref{lem:pairwise_correlation} suggests that the variance of the combination statistic with any transformation $G(\cdot)$ is $(1 + \gamma)$ times as large as that when the p-values are i.i.d.. In fact, under the global null,
\begin{align*}
\mathrm{Var}\left[\sum_{i=1}^{m}G(p_i)\right]& = m\mathrm{Var}\left[ G(p_1)\right] + m(m-1)\mathrm{Cov}\left[ G(p_1), G(p_2) \right]\\
& = \left(m + \frac{m(m-1)}{n+2}\right)\mathrm{Var}\left[G(p_1)\right]\\
& \approx (1 + \gamma)m\mathrm{Var}\left[G(p_1)\right].
\end{align*}

\begin{proof}[Proof of Lemma~\ref{lem:pairwise_correlation}]
Without loss of generality, assume $i = 1$ and $j = 2$. Let $(R_1, \ldots, R_n, R_{n+1}, R_{n+2})$ be the rank of $(S_{1}, \ldots, S_{n+2})\stackrel{d}{=}(\hat{s}(X_{n+1}), \ldots, \hat{s}(X_{2n}), \hat{s}(X_{2n+1}), \hat{s}(X_{2n+2}))$ in the ascending order. Then $S_1, \ldots, S_{n+2}$ are i.i.d.~draws from a non-atomic distribution, $(R_1, \ldots, R_{n+2})$ are mutually distinct almost surely and for any permutation $\pi: \{1, \ldots, n+2\}\mapsto \{1, \ldots, n+2\}$, 
\[(S_{\pi(1)}, \ldots, S_{\pi(n+2)})\stackrel{d}{=}(S_1, \ldots, S_{n+2}).\]
Therefore, $(R_1, \ldots, R_{n+2})$ is uniformly distributed over all permutations of $\{1, \ldots, n + 2\}$.
By definition, 
\[p_1 = \frac{1}{n+1}\sum_{i=1}^{n+1}I(S_i \le S_{n+1}), \quad R_{n+1} = \sum_{i=1}^{n+2}I(S_i \le S_{n+1}).\]
Thus, 
\[p_1 = \frac{R_{n+1} - I(S_{n+2}\le S_{n+1})}{n+1}.\]
Similarly,
\[p_2 = \frac{R_{n+2} - I(S_{n+1}\le S_{n+2})}{n+1}.\]
For any $j\in \{1, \ldots, n+1\}$, 
\begin{align*}
&\P\left[p_1 = p_2 = \frac{j}{n+1}\right] = \P\left[ R_{n+1} = j+1, R_{n+2} = j\right]+ \P\left[ R_{n+1} = j, R_{n+2} = j+1\right]  \\
& = 2\P\left[ R_{n+1} = j+1, R_{n+2} = j \right] = \frac{2}{(n+2)(n+1)}.
\end{align*}
For any $1 \le j < k\le n+1$, 
\begin{align*}
&\P\left[ p_1 = \frac{j}{n+1}, p_2 = \frac{k}{n+1}\right] = \P\left[ R_{n+1} = j, R_{n+2} = k + 1 \right] = \frac{1}{(n+2)(n+1)}.
\end{align*}
By symmetry,
\begin{align*}
&\P\left[ p_1 = \frac{k}{n+1}, p_2 = \frac{j}{n+1} \right] = \frac{1}{(n+2)(n+1)}.
\end{align*}
As a result,
\begin{align*}
\E[G(p_1)G(p_2)] &= \frac{2}{(n+2)(n+1)}\sum_{j=1}^{n+1}G^2\left(\frac{j}{n+1}\right) + \frac{1}{(n+2)(n+1)}\sum_{j\not=k}G\left(\frac{j}{n+1}\right)G\left(\frac{k}{n+1}\right)\\
& = \frac{1}{(n+2)(n+1)}\sum_{j=1}^{n+1}G^2\left(\frac{j}{n+1}\right) + \frac{1}{(n+2)(n+1)}\left\{\sum_{j=1}^{n+1}G\left(\frac{j}{n+1}\right)\right\}^2.
\end{align*}
On the other hand, since $p_1$ is uniformly distributed on $\{1 / (n+1), 2 / (n + 1), \ldots, 1\}$, 
\[\E[G(p_1)] = \frac{1}{n+1}\sum_{j=1}^{n+1}G\left(\frac{j}{n+1}\right), \quad \E[G^2(p_1)] = \frac{1}{n+1}\sum_{j=1}^{n+1}G^2\left(\frac{j}{n+1}\right).\]
Note that $\E[G^2(p_1)] < \infty$ since $G(i / (n + 1))\in \mathbb{R}$. As a result,
\begin{align*}
  \mathrm{Cov}\left[ G(p_1), G(p_2) \right] 
  &= \frac{1}{(n+2)(n+1)}\sum_{j=1}^{n+1}G^2\left(\frac{j}{n+1}\right) - \frac{1}{(n+2)(n+1)^2}\left\{\sum_{j=1}^{n+1}G\left(\frac{j}{n+1}\right)\right\}^2\\
  & = \frac{1}{n+2}\left\{\E[G^2(p_1)] - (\E[G(p_1)])^2\right\} = \frac{1}{n+2}\mathrm{Var}[G(p_1)].
\end{align*}
Therefore, 
\[\mathrm{Cor}\left[ G(p_1), G(p_2) \right] = \frac{\mathrm{Cov}\left[ G(p_1), G(p_2) \right]}{\sqrt{\mathrm{Var}[G(p_1)]\mathrm{Var}[G(p_2)]}} = \frac{1}{n + 2}.\]
\end{proof}

\subsection{Failure of type-I error control with combination tests} \label{app:global_testing}
We state a theorem for general (adjusted) combination tests which reject the global null if $$
\sum_{i=1}^{m}G(\hat{u}^\marg(Z_{2n+i}))\ge \xi c_{1-\alpha}(G) - m(\xi - 1)\int_0^1 G(u)du
$$
where $\xi > 0$ is a pre-specified constant and 
$$c_{1-\alpha}(G)\triangleq \textnormal{Quantile}_{1 - \alpha}\left(\sum_{i=1}^{m}G(U_i)\right), \quad U_i\stackrel{\mathrm{i.i.d.}}{\sim}\Unif([0, 1]).
$$

\begin{thm}\label{thm:combination_test}
Assume $\hat{s}(X)$ is continuously distributed and $G(\cdot): [0, 1]\mapsto \R$ is a non-constant function satisfying
\begin{enumerate}[(i)]
    \item $\int_{0}^{1}G^{2 + \eta}(u)du < \infty$;
    \item $\big|\frac{1}{n+1}\sum_{j=1}^{n+1}G^k\left(j / (n + 1)\right) - \int_{0}^{1}G^{k}(u)du\big| = o(1 / \sqrt{n})$, for $k \in \{1, 2\}$;
    \item $\max_{j\in \{1, \ldots, n + 1\}}G(j / (n + 1)) = o(\sqrt{n})$.
\end{enumerate}
Then, under the global null, if $m = \lfloor\gamma n\rfloor$ for some $\gamma \in (0, \infty)$, as $n \to \infty$, 
\begin{equation}\label{eq:unconditional_typeI}
\P\left[\sum_{i=1}^{m} G(\hat{u}^\marg(X_{2n+i}))\ge \xi c_{1 - \alpha}(G) - m(\xi - 1)\int_0^1 G(u)du\right]\rightarrow \bar{\Phi}\left(\frac{\xi z_{1 - \alpha}}{\sqrt{1 + \gamma}}\right),
\end{equation}
where $z_{1 - \alpha}$ and $\bar{\Phi}$ denote the $(1 - \alpha)$-th quantile and the survival function of the standard normal distribution, respectively. Furthermore, under the same asymptotic regime, for $W\sim N(0, 1)$,
\begin{equation}\label{eq:conditional_typeI-app}
\P\left[\sum_{i=1}^{m}G(\hat{u}^\marg(X_{2n+i}))\ge \xi c_{1 - \alpha}(G) - m(\xi - 1)\int_0^1 G(u)du\mid \mathcal{D} \right]\stackrel{d}{\rightarrow} \bar{\Phi}(\xi z_{1 - \alpha} + \sqrt{\gamma} W).
\end{equation}
\end{thm}

\begin{remark}
For Fisher's combination test, $G(u) = -2\log u$. Since $G(U)\sim \chi^2(2)$, condition (i) is clearly satisfied. To verify (ii), we note that $G(u)$ is decreasing and $|G'(u)| = 2 / u$ is decreasing. Thus, for $u\in [(j-1)/(n+1), j/(n+1)]$, for $k\in \{1, 2\}$,
\[0\le G^{k}(u) - G^{k}\left(\frac{j}{n+1}\right)\le \frac{k}{n+1}G^{k-1}\left(\frac{j}{n+1}\right)G'\left(\frac{j}{n+1}\right)\le \frac{8\log(n + 1)}{j}.\]
As a result,
\begin{align*}
&\bigg|\frac{1}{n+1}\sum_{j=1}^{n+1}G^k\left(j / (n + 1)\right) - \int_{0}^{1}G^{k}(u)du\bigg| \le \sum_{j=1}^{n+1}\bigg|\frac{1}{n+1}G^{k}\left(\frac{j}{n+1}\right) - \int_{(j-1)/(n+1)}^{j/(n+1)}G^{k}(u)du\bigg|\\
& \le \sum_{j=1}^{n+1}\int_{(j-1)/(n+1)}^{j/(n+1)}\big|G^k\left(j / (n + 1)\right) - G^{k}(u)\big|du \le \frac{1}{n+1}\sum_{j=1}^{n+1}\frac{8\log(n + 1)}{j} = O\left(\frac{\log^2 n}{n}\right).
\end{align*}
Thus, (ii) is proved. Finally, (iii) is satisfied because $G(j/(n+1))\le G(1/(n+1)) = O(\log n)$. Therefore, Theorem~\ref{thm:fisher} is a special case of Theorem~\ref{thm:combination_test} with $\xi = 1$. In general, it is easy to verify (i)--(iii) for various other combination functions~\cite{liptak1958combination, van1967combination, vovk2020combining}.
\end{remark}

\begin{remark}
By \eqref{eq:unconditional_typeI}, the limiting marginal type-I error is $\alpha$ when $\xi = \sqrt{1 + \gamma}$. This implies \eqref{eq:fisher_corrected} by noting that $\int_0^1 (-2\log u)du = 2$. By \eqref{eq:conditional_typeI-app}, since the random variable $\bar{\Phi}(\xi z_{1 - \alpha} + \sqrt{\gamma}W)$ has a positive density everywhere, the $(1 - \delta)$-th quantile of the conditional type-I error converges to the $(1 - \delta)$-th quantile of $\bar{\Phi}(\xi z_{1 - \alpha} + \sqrt{\gamma}W)$, which is $\bar{\Phi}(\xi z_{1 - \alpha} - \sqrt{\gamma}z_{1 - \delta})$. Thus, the conditional type-I error is controlled at level $\alpha$ with probability at least $1 - \delta$ asymptotically if $\xi = 1 + \sqrt{\gamma}z_{1 - \delta}/ z_{1 - \alpha}$. 
\end{remark}

\begin{remark}
To confirm our theory, we run Monte-Carlo simulations with $n = 10^{5}$ and $\gamma \in \{2^{-3}, 2^{-2},\ldots, 2^{3}\}$, estimating  the average type-I error across $10^{4}$ samples. Since $\hat{s}(X)$ is continuously distributed, we can assume that $\hat{s}(X) \sim \Unif([0, 1])$ without loss of generality, as we will show in the proof. Figure~\ref{fig:global_testing} presents the simulated and asymptotic type-I errors for both the unadjusted ($\xi = 1$) and adjusted ($\xi = \sqrt{1 + \gamma}$) Fisher's combination test given by \eqref{eq:fisher_corrected}. 
\begin{figure}[htp]
    \centering
    \includegraphics[width = 0.7\textwidth]{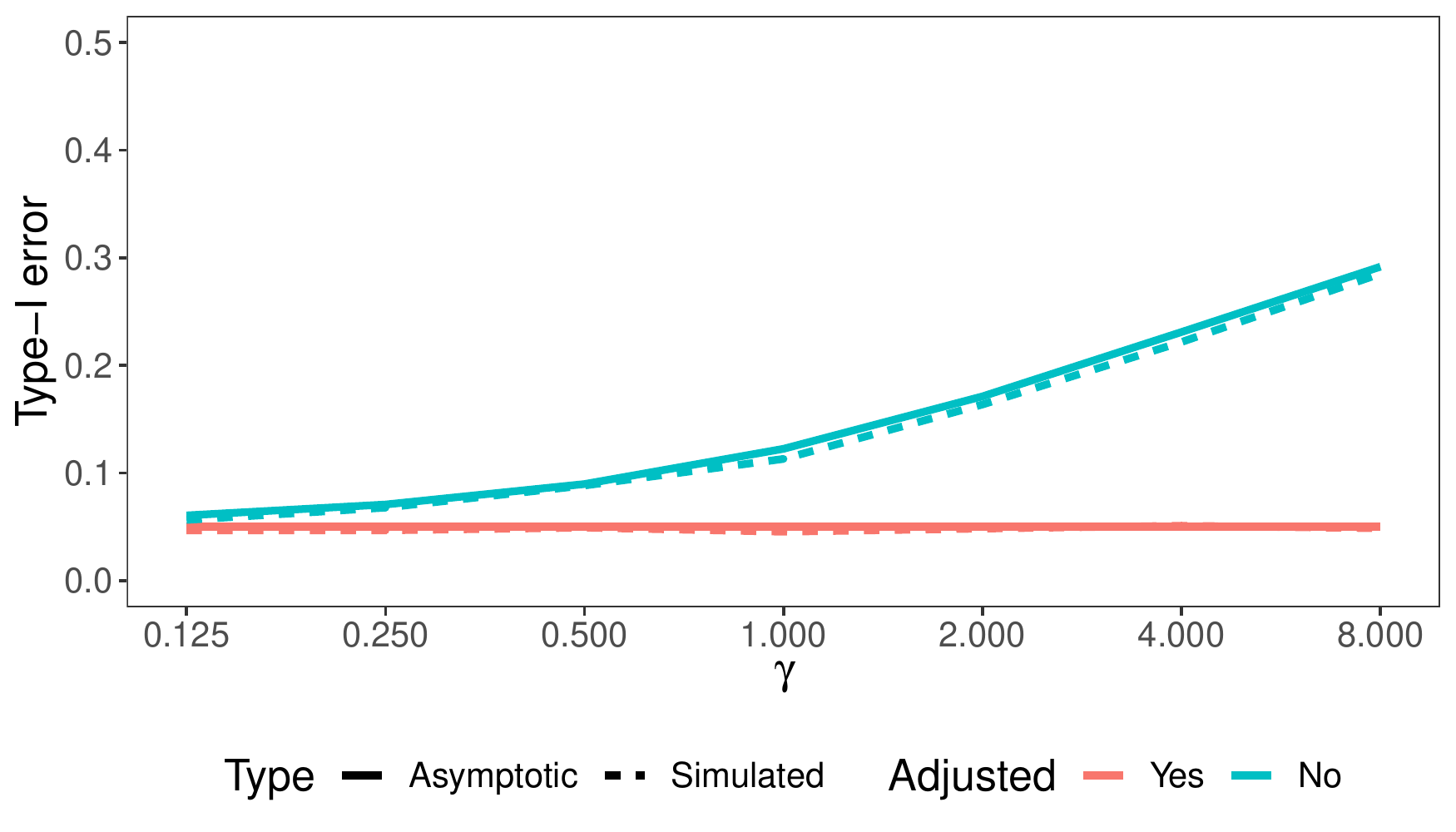}
    \caption{Type-I errors of unadjusted and adjusted Fisher's combination test.}
    \label{fig:global_testing}
\end{figure}
\end{remark}

\begin{remark}
If the $p_i$'s are dependent,~\cite{brown1975400} and~\cite{kost2002combining} approximate the null distribution by a rescaled chi-square distribution $c\chi^2(f)$, where $c$ and $f$ are chosen to match the mean and variance of the Fisher's combination statistic $S_{\mathrm{Fisher}}$. Specifically,
\[c = \frac{\mathrm{Var}[S_{\mathrm{Fisher}}]}{2\E[S_{\mathrm{Fisher}}]}, \quad f = \frac{2\E[S_{\mathrm{Fisher}}]^2}{\mathrm{Var}[S_{\mathrm{Fisher}}]}.\]
In our case, it is easy to see that 
\[\E[S_{\mathrm{Fisher}}]\approx 2m, \quad \mathrm{Var}[S_{\mathrm{Fisher}}]\approx 4m (1 + \gamma).\]
As a result, the null distribution is approximated by $(1 + \gamma)\chi^2(2m / (1 + \gamma))$. The central limit theorem implies that $\chi^2(f)\approx N(f, 2f)$ when $f$ is large. Thus, the critical value for this approximation is
\[(1 + \gamma)\chi^2(2m / (1 + \gamma); 1 - \alpha)\approx (1 + \gamma)\left(\frac{2m}{1 + \gamma} + \sqrt{\frac{2m}{1 + \gamma}}z_{1 - \alpha}\right) = 2m + \sqrt{2m(1 + \gamma)}z_{1 - \alpha}.\]
Similarly, the critical value for our correction \eqref{eq:fisher_corrected} is 
\[\sqrt{1 + \gamma}\chi^2(2m; 1 - \alpha) - 2(\sqrt{1 + \gamma} - 1)m\approx \sqrt{1 + \gamma}(2m + \sqrt{2m}z_{1 - \alpha}) - 2(\sqrt{1 + \gamma} - 1)m\approx 2m + \sqrt{2m(1 + \gamma)}z_{1 - \alpha}.\]
Therefore, both corrections are asymptotically equivalent.
\end{remark}

To prove Theorem~\ref{thm:combination_test}, we start by stating two lemmas. The first lemma is a general Berry-Esseen bound for sums of independent (but not necessarily identically distributed) random variables with potentially infinite third moments.

\begin{lemma}\label{lem:Berry_Esseen}[\cite{petrov1975independent}, p.~112, Theorem 5]
Let $X_1, X_2, \ldots, X_{n}$ be independent random variables such that $\E[X_j] = 0$, for all $j$. Assume also $\E[X_j^2 g(X_j)] < \infty$ for some function $g$ that is non-negative, even, and non-decreasing in the interval $x > 0$, with $x / g(x)$ being non-decreasing for $x > 0$. Write $B_{n} = \sum_{j}\Var[X_j]$. Then,
\[d_{K}\left(\mathcal{L}\left(\frac{1}{\sqrt{B_{n}}}\sum_{j=1}^{n}X_j\right), N(0, 1)\right)\le \frac{A}{B_n g(\sqrt{B_n})}\sum_{j=1}^{n}\E\left[X_j^2 g(X_j)\right],\]
where $A$ is a universal constant, $\mathcal{L}(\cdot)$ denotes the probability law, $d_{K}$ denotes the Kolmogorov-Smirnov distance (i.e., the $\ell_{\infty}$-norm of the difference of CDFs)
\end{lemma}

The second lemma is a well-known representation of the spacing between consecutive order statistics. 
\begin{lemma}[From~\cite{moran1947random}; see also Section 4 of~\cite{arnold2008first}]\label{lem:moran}
  Let $U_1, \ldots, U_n\stackrel{i.i.d.}{\sim}\mathrm{Unif}([0, 1])$ and $U_{(1)} \le U_{(2)}\le \ldots \le U_{(n)}$ be their order statistics. Then
\[
(U_{(1)} - U_{(0)}, \ldots, U_{(n+1)} - U_{(n)}) \stackrel{d}{=} \left(\frac{V_1}{\sum_{k=1}^{n+1}V_{k}}, \ldots, \frac{V_{n+1}}{\sum_{k=1}^{n+1}V_{k}}\right),
\]
where $U_{(0)} = 0, U_{(n + 1)}= 1$, and $V_{1},\ldots, V_{n+1}\stackrel{\mathrm{i.i.d.}}{\sim}\mathrm{Exp}(1)$.
\end{lemma}

\begin{proof}[\textbf{Proof of Theorem~\ref{thm:combination_test}}]
We first prove the limiting conditional type-I error \eqref{eq:conditional_typeI-app}. For convenience, we write $p_i$ instead of $\hat{u}^\marg(X_{2n+i})$ and $S_j$ instead of $\hat{s}(X_{n+j})$. Since $\hat{s}(X)$ is continuously distributed, 
\[
p_i = \frac{1 + |\{j\in \Dcal: S_{j} \le \hat{s}(X_{2n+i}) \}|}{n + 1} = \frac{1 + |\{j\in \Dcal: F_{S}(S_{j}) \le F_{S}(\hat{s}(X_{2n+i})) \}|}{n + 1}\]
where $F_{S}$ denotes the CDF of $\hat{s}(X)$ conditional on $\mathcal{D}$. As a result, we can assume $\hat{s}(X)\sim \Unif([0, 1])$ without loss of generality. Conditional on $\mathcal{D}$, $p_1, \ldots, p_{m}$ are i.i.d.~random variables with 
\[\P\left[p_i = \frac{j}{n+1}\mid \mathcal{D} \right] = S_{(j)} - S_{(j-1)}, \qquad j = 1,\ldots, n+1,
\]
where $S_{(1)} < S_{(2)} < \ldots < S_{(n)}$ denote the order statistics of $(S_1, \ldots, S_{n})$, and $S_{(0)} = 0, S_{(n + 1)}= 1$. By Lemma~\ref{lem:moran}, we can reformulate the distribution of $p_i$ conditional on $\mathcal{D}$ as 
\begin{equation}\label{eq:pi_dist}
    \P\left[p_i = \frac{j}{n+1}\mid \mathcal{D}\right] = \frac{V_{j}}{\sum_{k=1}^{n+1}V_{k}}, \qquad j = 1,\ldots, n+1.
\end{equation}
As a result, for $k \in \{1, 2\}$,
\begin{align}\label{eq:EGk}
\E\left[G^k(p_i)\mid \mathcal{D}\right] = \frac{\sum_{j=1}^{n+1}G^k\left(\frac{j}{n+1}\right)V_{j}}{\sum_{j=1}^{n+1}V_{j}} = \frac{(n+1)^{-1}\sum_{j=1}^{n+1}G^k\left(\frac{j}{n+1}\right)V_{j}}{(n+1)^{-1}\sum_{j=1}^{n+1}V_{j}}.
\end{align}
By the strong law of large number,
\begin{equation}\label{eq:LLN_sumV}
\frac{1}{n+1}\sum_{j=1}^{n+1}V_{j}\stackrel{\mathrm{a.s.}}{\rightarrow}\E[V_{1}] = 1.
\end{equation}
Let $g_{n} = \max_{j\in \{1, \ldots, n+1\}}G(j / (n + 1))$. Since $V_1, \ldots, V_{n+1}$ are independent, 
\begin{align}
&\Var\left[\frac{1}{n+1}\sum_{j=1}^{n+1}G^k\left(\frac{j}{n+1}\right)V_{j}\right] = \sum_{j=1}^{n+1}\frac{1}{(n + 1)^2}\E\left[G^{2k}\left(\frac{j}{n+1}\right)(V_{j} - 1)^2\right]\nonumber\\
& = \frac{1}{(n + 1)^2}\sum_{j=1}^{n + 1}G^{2k}\left(\frac{j}{n+1}\right)\le \frac{g_{n}^{2k - 2}}{(n + 1)}\frac{1}{n+1}\sum_{j=1}^{n + 1}G^{2}\left(\frac{j}{n+1}\right).\label{eq:G2k}
\end{align}
By condition (ii), 
\[\bigg|\frac{1}{n+1}\sum_{j=1}^{n + 1}G^{2}\left(\frac{j}{n+1}\right) - \int_0^1 G^2(u)du\bigg| = o(1),\]
and thus
\[\frac{1}{n+1}\sum_{j=1}^{n + 1}G^{2}\left(\frac{j}{n+1}\right) = O(1).\]
By condition (iii), $g_{n} = o(\sqrt{n})$. Together with \eqref{eq:G2k}, we obtain that for $k\in \{1, 2\}$,
\[\Var\left[\frac{1}{n+1}\sum_{j=1}^{n+1}G^k\left(\frac{j}{n+1}\right)V_{j}\right] = o(1).\]
By Chebyshev's inequality, 
\[\frac{1}{n+1}\sum_{i=1}^{n}G^{k}\left(\frac{j}{n+1}\right)(V_{j} - 1) = o_{P}(1).\]
Applying the condition (ii) again, we arrive at
\begin{equation*}
\frac{1}{n+1}\sum_{i=1}^{n}G^{k}\left(\frac{j}{n+1}\right)V_{j} - \int_{0}^{1}G^k(u)du = o_{P}(1).    
\end{equation*}
By \eqref{eq:EGk},
\begin{equation}\label{eq:EGk_limit}
\E\left[G^k(p_i)\mid \mathcal{D}\right] - \int_{0}^{1}G^k(u)du = o_{P}(1), \qquad k\in \{1, 2\}.
\end{equation}
Let $a_{n}$ be a deterministic sequence such that $a_{n} < 1 / 2$, and $U\sim \Unif([0, 1])$. Let also $\mathcal{E}_{n}$ be the event that $\mathcal{D}$ is such that
\begin{equation}\label{eq:event_Var}
    \mathcal{E}_{n} = \left\{\mathcal{D}: \frac{\Var[G(p_i)\mid \mathcal{D}]}{\Var[G(U)]}\in [1 - a_{n}, 1 + a_{n}]\right\}.
\end{equation}
Since $G$ is a non-constant function, $\Var[G(U)] > 0$. By \eqref{eq:EGk_limit}, we can choose $a_n = o(1)$ such that
\[\P\left[\mathcal{E}_{n}^c\right] = o(1).\]
Let 
\[W_{m} = \frac{\sum_{i=1}^{m}\left\{G(p_i) - \E[G(p_i)\mid \mathcal{D}]\right\}}{\sqrt{m\Var[G(p_i)\mid \mathcal{D}]}}.\]
By Lemma~\ref{lem:Berry_Esseen} with $g(x) = x$, 
\[d_{K}\left(\mathcal{L}\left(W_{m}\mid \mathcal{D}\right), N(0, 1)\right)\le \frac{A}{\sqrt{m}}\frac{\E\left[|G(p_i) - \E[G(p_i)\mid \mathcal{D}]|^3\right]}{\Var[G(p_i)\mid \mathcal{D}]^{3/2}},\]
where $A$ is a universal constant. Since $G(p_i)\le g_{n}$ almost surely, by condition (iii), 
\[\E\left[|G(p_i) - \E[G(p_i)\mid \mathcal{D}]|^3\right]\le 2g_{n}\Var[G(p_i)\mid \mathcal{D}].\]
Thus, 
\[d_{K}\left(\mathcal{L}\left(W_{m}\mid \mathcal{D}\right), N(0, 1)\right)\le \frac{2A}{\sqrt{m}}\frac{g_{n}}{\Var\left[G(p_i)\mid \mathcal{D}\right]^{1/2}}.\]
On the event $\mathcal{E}_{n}$, the condition (iii) and that $n = O(m)$ imply that
\begin{equation*}
  d_{K}\left(\mathcal{L}\left(W_{m}\mid \mathcal{D}\right), N(0, 1)\right)\le \frac{4Ag_{n}}{\sqrt{m\Var[G(U)]}} = o(1).  
\end{equation*}
Since the Kolmogorov distance is invariant under rescalings, we have
\[d_{K}\left(\mathcal{L}\left(\sqrt{\frac{\Var[G(p_i)\mid \mathcal{D}]}{\Var[G(U)]}} W_{m}\mid \mathcal{D}\right), N\left(0, \frac{\Var[G(p_i)\mid \mathcal{D}]}{\Var[G(U)]}\right)\right) = o(1).\]
Since $\Var[G(p_i)\mid \mathcal{D}]/\Var[G(U)]\in [1 - a_{n}, 1 + a_{n}]\rightarrow 1$,
\[d_{K}\left( N\left(0, \frac{\Var[G(p_i)\mid \mathcal{D}]}{\Var[G(U)]}\right), N(0, 1)\right) = o(1).\]
Let 
\begin{equation}\label{eq:conditional_CLT}
  K_{n}\triangleq d_{K}\left(\mathcal{L}\left(\sqrt{\frac{\Var[G(p_i)\mid \mathcal{D}]}{\Var[G(U)]}} W_{m}\mid \mathcal{D}\right), N(0, 1)\right).
\end{equation}
The above arguments show that $K_{n} = o(1)$ on the event $\mathcal{E}_{n}$. 

~\\
\noindent On the other hand, let 
\begin{equation*}
   c_{m} = \frac{c_{1 - \alpha}(G) - m\E[G(U)]}{\sqrt{m\Var[G(U)]}},
\end{equation*}
and 
\begin{equation*}
\tilde{W}_{n} = \frac{\sqrt{n + 1}(\E[G(p_i)\mid \mathcal{D}] - \E[G(U)])}{\sqrt{\Var[G(U)]}}.
\end{equation*}
Then

\begin{align*}
\P\left[\sum_{i=1}^{m}G(p_i)\ge \xi c_{1 - \alpha}(G) - m(\xi - 1)\E[G(U)]\mid \mathcal{D}\right] = \P\left[\sqrt{\frac{\Var[G(p_i)\mid \mathcal{D}]}{\Var[G(U)]}}W_{m} + \sqrt{\frac{m}{n + 1}}\tilde{W}_{n}\ge \xi c_{m}\mid \mathcal{D}\right].
\end{align*}

By \eqref{eq:conditional_CLT},
\[\bigg|\P\left[\sqrt{\frac{\Var[G(p_i)\mid \mathcal{D}]}{\Var[G(U)]}}W_{m} + \sqrt{\frac{m}{n + 1}}\tilde{W}_{n}\ge \xi c_{m}\mid \mathcal{D}\right] - \bar{\Phi}\left(\xi c_{m} - \sqrt{\frac{m}{n + 1}}\tilde{W}_{n}\right)\bigg| \le K_{n}.\]
Since $K_{n} = o(1)$ on $\mathcal{E}_{n}$ and $\P[\mathcal{E}_{n}^{c}] = o(1)$, we obtain that 
\begin{equation}\label{eq:step1}
\bigg|\P\left[\sum_{i=1}^{m}G(p_i)\ge c_{1 - \alpha}(G)\mid \mathcal{D}\right] -  \bar{\Phi}\left(\xi c_{m} - \sqrt{\frac{m}{n + 1}}\tilde{W}_{n}\right)\bigg| = o_{P}(1).
\end{equation}
Since $\bar{\Phi}$ is a continuous function and $m / n\rightarrow \gamma$, to prove \eqref{eq:conditional_typeI-app}, it remains to prove
\begin{equation}\label{eq:goal1}
c_{m}\stackrel{p}{\rightarrow} z_{1 - \alpha}, \quad \tilde{W}_{n}\stackrel{d}{\rightarrow} N(0, 1).
\end{equation}
Without loss of generality, we assume that $\eta \le 1$ in the condition (i). By Lemma~\ref{lem:Berry_Esseen} with $g(x) = x^{\eta}$, which clearly fulfills the criteria, we have that
\begin{equation}\label{eq:cm}
d_{K}\left(\frac{\sum_{j=1}^{m}G(U_i) - \E[G(U)] }{\sqrt{m\Var[G(U)]}}, N(0, 1)\right)\le \frac{A}{m^{\eta / 2}}\frac{\E[|G(U) - \E[G(U)]|^{2 + \eta}]}{\Var[G(U)]^{1 + \eta/2}} = o(1).
\end{equation}
By definition, $c_{m}$ is the $(1 - \alpha)$-th quantile of $\left(\sum_{j=1}^{m}G(U_i) - \E[G(U)]\right) / \sqrt{m\Var[G(U)]}$. By \eqref{eq:cm},
\[|\bar{\Phi}(c_{m}) - \alpha| = |\bar{\Phi}(c_{m}) - \bar{\Phi}(z_{1 - \alpha})| = o(1).\]
Since $\bar{\Phi}'(z_{1 - \alpha}) > 0$, it implies the first part of \eqref{eq:goal1}. 

~\\
\noindent To prove the second part of \eqref{eq:goal1}, we recall \eqref{eq:EGk} with $k = 1$ that
\[\tilde{W}_{n} = \frac{(n + 1
)^{-1/2}\sum_{j=1}^{n+1}\left\{G\left(\frac{j}{n+1}\right) - \E[G(U)]\right\}V_{j}}{\sqrt{\Var[G(U)]}\left(\sum_{j=1}^{n+1}V_{j}\right)/(n+1)}.\]
Set $X_j = (n + 1)^{-1/2}\left\{G\left(\frac{j}{n+1}\right) - \E[G(U)]\right\}(V_{j} - 1)$ and $g(x) = x$ in Lemma~\ref{lem:Berry_Esseen}. Then
\[B_{n} = \frac{1}{n+1}\sum_{j=1}^{n+1}\left\{G\left(\frac{j}{n+1}\right) - \E[G(U)]\right\}^2.\]
By the condition (ii), we have that
\begin{equation}\label{eq:Bn}
B_{n} = \frac{1}{n+1}\sum_{j=1}^{n}G^2\left(\frac{j}{n+1}\right) - \frac{2\E[G(U)]}{n+1}\sum_{j=1}^{n}G\left(\frac{j}{n+1}\right) + (\E[G(U)])^2\rightarrow \Var[G(U)].
\end{equation}
By the condition (i), (iii) and \eqref{eq:Bn}, 
\begin{align*}
\sum_{j=1}^{n+1}\E|X_j|^3 
&\le \frac{1}{(n+1)^{3/2}}\sum
_{j=1}^{n+1}\bigg|G\left(\frac{j}{n+1}\right) - \E[G(U)]\bigg|^3\\
&\le \frac{g_{n} + \E[G(U)]}{\sqrt{n+1}}\frac{1}{n+1}\sum_{j=1}^{n+1}\left(G\left(\frac{j}{n+1}\right) - \E[G(U)]\right)^2\\
&= \frac{g_{n} + \E[G(U)]}{\sqrt{n+1}}B_{n} = o(1).
\end{align*}
Let 
\[\tilde{W}'_{n} = \frac{1}{\sqrt{B_n}}\sum_{j=1}^{n+1}X_j = \frac{1}{\sqrt{(n + 1)B_{n}}}\sum_{j=1}^{n+1}\left\{G\left(\frac{j}{n+1}\right) - \E[G(U)]\right\}(V_{j} - 1).\]
Then Lemma~\ref{lem:Berry_Esseen} implies that
\begin{equation}\label{eq:CLT_tildeW_prime}
d_{K}\left(\tilde{W}'_{n}, N(0, 1)\right)\le \frac{A\sum_{j=1}^{n+1}\E|X_j|^3}{B_{n}^{3/2}} = o(1).
\end{equation}
By definition,
\[\tilde{W}_{n} = \left(\tilde{W}'_{n} + \frac{1}{\sqrt{(n + 1)B_{n}}}\sum_{j=1}^{n+1}\left\{G\left(\frac{j}{n+1}\right) - \E[G(U)]\right\}\right)\sqrt{\frac{B_{n}}{\Var[G(U)]}}\frac{1}{\left(\sum_{j=1}^{n}V_{j}\right) / (n+1)}.\]
The condition (ii) with $k = 1$ implies that 
\begin{equation}\label{eq:bias}
\frac{1}{\sqrt{(n + 1)}}\sum_{j=1}^{n+1}\left\{G\left(\frac{j}{n+1}\right)- \E[G(U)]\right\} = o(1).
\end{equation}
By \eqref{eq:LLN_sumV}, \eqref{eq:Bn}, \eqref{eq:CLT_tildeW_prime}, \eqref{eq:bias} and Slutsky's Lemma, we prove the second part of \eqref{eq:goal1}. Therefore, the limiting conditional type-I error \eqref{eq:conditional_typeI-app} is proved. 

Next, we prove the limiting marginal type-I error \eqref{eq:unconditional_typeI}. Since $$\P\left[\sum_{i=1}^{m}G(p_i)\ge \xi c_{1 - \alpha}(G) - m(\xi - 1)\E[G(U)] \mid \mathcal{D} \right]$$ is bounded almost surely, the convergence in distribution implies the convergence in expectation. Therefore,
\[\P\left[\sum_{i=1}^{m}G(p_i)\ge \xi c_{1 - \alpha}(G) - m(\xi - 1)\E[G(U)]\right] \rightarrow \E[\bar{\Phi}(\xi z_{1 - \alpha} + \sqrt{\gamma}W)].\]
Let $W'$ be an independent copy of $W$. Then
\[\bar{\Phi}(\xi z_{1 - \alpha} + \sqrt{\gamma}W) = \P\left[ W' \ge \xi z_{1 - \alpha} + \sqrt{\gamma}W\mid W\right].\]
As a result,
\[\E[\bar{\Phi}(\xi z_{1 - \alpha} + \sqrt{\gamma}W)] = \P\left[W' \ge \xi z_{1 - \alpha} - \sqrt{\gamma}W \right] = \P\left[W' - \sqrt{\gamma}W\ge \xi z_{1 - \alpha} \right].\]
The proof is completed by the fact that $W' - \sqrt{\gamma}W\sim N(0, 1 + \gamma)$.
 
\end{proof}

\subsection{Conformal p-values are PRDS}\label{app:PRDS}

\begin{proof}[Proof of Theorem~\ref{thm:prds}]
Let $Z = (S_{(1)},\dots,S_{(n)})$ be the order statistics of $(\hat{s}(X_i))_{i \in \{n+1,\dots,2n\}}$, the conformal scores evaluated on the calibration set.  Let $Y = (p_1,\dots,p_m)$ be the conformal p-values evaluated on the test set (i.e., $p_j = \hat{u}^\marg(X_{2n+j})$). Then,
\begin{align*}
    \P\left[ Y \in A \mid Y_i = y \right] &= \int \P\left[ Y \in A \mid Z = z\right] \P\left[Z = z \mid Y_i = y\right] dz \\
    &= \E_{Z \mid Y_i = y}\left[ \P\left[ Y \in A \mid Z\right] \right].
\end{align*}
With this representation, the conclusion will be implied by the following two lemmas.

\begin{lemma} \label{lem:prds1}
For a non-decreasing set $A$ and vectors $z, z'$ such that $z \preceq z'$, then 
\begin{equation*}
    \P\left[ Y \in A \mid Z = z \right] \ge \P \left[ Y \in A \mid Z = z' \right].
\end{equation*}
\end{lemma}

\begin{lemma} \label{lem:prds2}
For $y \ge y'$, if $i$ belongs to the set of inliers, there exists $Z_1 \sim Z \mid Y_i = y$ and $Z_2 \sim Z \mid Y_i = y'$ such that $\P\left[Z_1 \preceq Z_2\right] = 1$. 
\end{lemma}

In other words, Lemma~\ref{lem:prds1} states that the conformal p-values increase as the conformal scores on the calibration set decrease, while Lemma~\ref{lem:prds2} states that a larger conformal p-value indicates the calibration conformal scores are smaller.
The proof follows easily from these. Take any $y \ge y'$ and let $Z_1$ and $Z_2$ be as in the statement of Lemma~\ref{lem:prds2}. Then, for any $i$ belonging to the set of inliers,
\begin{align*}
\P\left[ Y \in A \mid Y_i = y \right] 
    &= \E_{Z_1}\left[\P\left[ Y \in A \mid Z = Z_1\right] \right] \\
    &\ge \E_{Z_2}\left[\P\left[ Y \in A \mid Z = Z_2 \right] \right] \\
    &= \P\left[Y \in A \mid Y_i = y'\right].
\end{align*}
The inequality follows from Lemma~\ref{lem:prds1} and the fact that $\P\left[ Z_1 \preceq Z_2 \right] = 1$, which comes from Lemma~\ref{lem:prds2}.
\end{proof}

Lemma~\ref{lem:prds1} follows immediately from the definition of marginal conformal p-values in~\eqref{eq:marginal-pvals-def}. Lemma~\ref{lem:prds2} is proved below.

\begin{proof}[Proof of Lemma~\ref{lem:prds2}, continuous case]
As in the proof of Theorem~\ref{thm:combination_test}, since $\hat{s}(X)$ is continuously distributed, we can assume without loss of generality that the scores $S_i$ follow the uniform distribution on $[0,1]$. Let $S'_{(1)} \le S'_{(2)}\le \ldots \le S'_{(n+1)}$ be the order statistics of $(\hat{s}(X_{n+1}), \ldots, \hat{s}(X_{2n+1}))$ and $R_{2n+1}$ be the rank of $\hat{s}(X_{2n+1})$ among these. By definition,
\[\bigg\{(S_{(1)}, \ldots, S_{(n)})\mid R_{2n+1} = k, S'_{(1)}, \ldots, S'_{(n+1)}\bigg\} = (S'_{(1)}, \ldots, S'_{(k-1)}, S'_{(k+1)}, \ldots, S'_{(n+1)}).\]
Since $\hat{s}
(X)$ is continuously distributed, $R_{2n+1}$ is independent of $(S'_{(1)}, S'_{(2)}, \ldots, S'_{(n+1)})$. As a result, for any positive integer $k \le n + 1$,
\[\bigg\{(S_{(1)}, \ldots, S_{(n)})\mid R_{2n+1} = k\bigg\}\stackrel{d}{=}(S'_{(1)}, \ldots, S'_{(k-1)}, S'_{(k+1)}, \ldots, S'_{(n+1)}).\]
The right-hand-side is clearly entry-wise non-increasing in $k$. Since $p_1 = R_{2n+1} / (n + 1)$, Lemma~\ref{lem:prds2} is proved for $i = 1$. The same proof carries over to other indices $i$ belonging to the set of inliers.

\end{proof}

\paragraph{Extension to non-continuous scores.} When $\hat{s}(X)$ has atoms, the set of conformity scores $\{\hat{s}(X_i): i\in \Dcal\}$ have ties with non-zero probability. In this case, we replace the marginal conformal p-value \eqref{eq:marg_validity_def} by a randomized version, i.e., 
\begin{equation}\label{eq:marginal_discrete}
p_j = \frac{|\{i \in \mathcal{D}^{\text{cal}} : \hat{s}(X_i) < \hat{s}(X_{2n+j})\}| + \lceil (1 + |\{i \in \mathcal{D}^{\text{cal}} : \hat{s}(X_i) = \hat{s}(X_{2n+j})\}|)U_j\rceil}{n+1},
\end{equation}
where $U_1, U_2, \ldots$ are i.i.d. random variables drawn from $\Unif([0, 1])$ which are independent of the data. Note that  \eqref{eq:marginal_discrete} is identical to \eqref{eq:marg_validity_def} almost surely when $\hat{s}(X)$ is continuously distributed. Now we prove that the marginal conformal p-values defined in \eqref{eq:marginal_discrete} satisfy the PRDS property. 

\begin{prop}[Theorem~\ref{thm:prds} for the non-continuous case]
Consider the setting of Theorem~\ref{thm:prds}, but where $\hat{s}(\cdot)$ is not assumed to be continuous. Define the randomized marginal p-values as in \eqref{eq:marginal_discrete}. Then, the marginal conformal p-values $(p_1,\dots,p_m)$ are PRDS.
\end{prop}

The proof follows as above, once we verify Lemma~\ref{lem:prds1} and Lemma~\ref{lem:prds2} in the more general setting.

\begin{proof}[Proof of Lemma~\ref{lem:prds1}, general case]
Let $U = (U_1, \ldots, U_m)$. By definition, $U$ is independent of $(Y, Z)$, and thus
\[\P[Y\in A\mid Z = z] = \P[Y\in A \mid Z = z, U], \quad \text{a.s.}.\]
Let $p_j(x; z, u)$ denote the mapping from $(X_{2n+j}, Z, U)$ to $p_j$. Then 
\[p_j(x; z, u) = \frac{m_{<}(x; z) + \lceil\{1 + m_{=}(x; z)\}u\rceil}{n + 1},\]
where
\[m_{<}(x, z) = \sum_{i=1}^{n}I(z_i < x), \quad m_{=}(x, z) = \sum_{i=1}^{n}I(z_i = x).\]
If $z \preceq z'$,
\begin{equation}\label{eq:m<m=}
m_{<}(x, z)\ge m_{<}(x, z'), \quad m_{<}(x, z) + m_{=}(x, z)\ge m_{<}(x, z') + m_{=}(x, z').
\end{equation}
We claim that the mapping $p_j(x; z, u)$ is non-increasing in $z$ for every $x$ and $u$. Equivalently, we will show that for any $x$ and $u\in [0, 1]$, 
\begin{equation}\label{eq:pj_monotone}
m_{<}(x, z) + \lceil \{1 + m_{=}(x, z)\} u\rceil\ge m_{<}(x, z') + \lceil \{1 + m_{=}(x, z')\}u\rceil.
\end{equation}
We consider three cases. 
\begin{enumerate}[{Case} 1: ]
\item if $m_{<}(x, z) = m_{<}(x, z')$, \eqref{eq:m<m=} implies that $m_{=}(x, z)\ge m_{=}(x, z')$. Thus, \eqref{eq:pj_monotone} is obviously true.
\item if $m_{<}(x, z) + m_{=}(x, z) = m_{<}(x, z') + m_{=}(x, z')$, let $a = 1 + m_{=}(x, z)$ and $b = m_{<}(x, z) - m_{<}(x, z')$. Then \eqref{eq:pj_monotone} is equivalent to
\[b\ge \lceil (a + b)u\rceil - \lceil au\rceil.\]
This can be proved using the fact that $\lceil (a + b)u\rceil \le \lceil au\rceil + \lceil bu\rceil$.
\item if $m_{<}(x, z) > m_{<}(x, z')$ and $m_{<}(x, z) + m_{=}(x, z) > m_{<}(x, z') + m_{=}(x, z')$, then 
$m_{<}(x, z) \ge m_{<}(x, z') + 1$ and $m_{<}(x, z) + m_{=}(x, z) \ge m_{<}(x, z') + m_{=}(x, z') + 1$ since
$m_{<}(x, z), m_{<}(x, z'), m_{=}(x, z)$, and $m_{=}(x, z')$ are all integers. Then
\begin{align*}
 m_{<}(x, z) + \lceil \{1 + m_{=}(x, z)\}u\rceil
&\ge m_{<}(x, z) + \{1 + m_{=}(x, z)\}u \\
&= m_{<}(x, z)(1 - u) + \{1 + m_{=}(x, z) + m_{<}(x, z)\}u\\
&\ge \{1 + m_{<}(x, z')\}(1 - u) + \{2 + m_{=}(x, z') + m_{<}(x, z')\}u\\
& = m_{<}(x, z') + \{1 + m_{=}(x, z')\}u + 1\\
& \ge m_{<}(x, z') + \lceil\{1 + m_{=}(x, z')\}u\rceil.
\end{align*}
\end{enumerate}
Therefore, \eqref{eq:pj_monotone} is proved. As a result, the mapping from $(X_{2n+1}, \ldots, X_{2n+m}, Z, U)$ to $Y$ is entry-wise non-increasing in $Z$ given $(X_{2n+j}, \ldots, X_{2n+m}, U)$. Since $\{X_{2n+j}: j = 1,\ldots,m\}$, $Z$, and $U$ are mutually independent, we arrive at
\[\P[Y\in A \mid Z = z, U]\ge \P[Y\in A \mid Z = z', U], \quad \text{a.s.}.\]
The independence between $U$ and $Z$ implies that $(U\mid Z = z) \stackrel{d}{=} (U\mid Z = z')$.  Lemma \ref{lem:prds1} then follows from the above inequality.
\end{proof}

\begin{proof}[Proof of Lemma~\ref{lem:prds2}, general case]
Let $R_{2n+j} = (n + 1)p_j$. Note that $R_{2n+j}$ can be interpreted as the rank with ties broken randomly. As in the proof for the continuous case, we first prove that
\begin{equation}\label{eq:discrete_PRDS_step1}
\bigg\{(S_{(1)}, \ldots, S_{(n)})\mid R_{2n+1} = k, S'_{(1)}, \ldots, S'_{(n+1)}\bigg\} = (S'_{(1)}, \ldots, S'_{(k-1)}, S'_{(k+1)}, \ldots, S'_{(n+1)}).
\end{equation}
Let $k_{-} = \max\{\ell: S'_{(\ell)} < S_{2n+1}\}$ and $k_{+} = \min\{\ell: S'_{(\ell)} > S_{2n+1}\}$. Then $S'_{\ell} = S_{2n+1}$ for any $k_{-} < \ell < k_{+}$. Since there exists at least one $\ell$ with $S'_{(\ell)} = S_{2n+1}$, i.e., the index corresponding to $S_{2n+1}$, we have $k_{+} - k_{-}\ge 2$. By definition, 
\[1 + |\{i \in \mathcal{D}^{\text{cal}} : \hat{s}(X_i) = \hat{s}(X_{2n+j})\}| = |\{i \in \mathcal{D}^{\text{cal}}\cup \{2n+1\} : \hat{s}(X_i) = \hat{s}(X_{2n+j})\}| = k_{+} - k_{-} - 1.\]
As a result,
\[k = k_{-} + \lceil (k_{+} - k_{-} - 1)U_{1}\rceil\in (k_{-}, k_{+}).\]
Therefore, $\hat{s}(X_{2n+1}) = S'_{(k)}$ and \eqref{eq:discrete_PRDS_step1} is proved.

It remains to prove that $R_{2n+1}$ is independent of $(S'_{(1)}, S'_{(2)}, \ldots, S'_{(n+1)})$. For any non-decreasing sequence $a_{1}\le \ldots\le a_{n+1}$, let $1 = n_0 < n_1 < \ldots < n_m = n+1$ be integers such that
\[a_{n_{j-1}} = \ldots = a_{n_j - 1} < a_{n_j}, \quad j = 1,\ldots, m - 1, \quad a_{n_{m-1}-1} < a_{n_{m-1}} = \ldots = a_{n_{m}}\]
Let $\pi: \{1, \ldots, n+1\}\mapsto \{1, \ldots, n + 1\}$ be a uniform random permutation. Since $X_{n+1}, \ldots, X_{2n+1}$ are i.i.d.,  Conditioning on the event that, 
\[\bigg\{\left(\hat{s}(X_{n+1}), \ldots, \hat{s}(X_{2n+1})\right)\mid (S'_{(1)}, \ldots, S'_{(n+1)}) = (a_1, \ldots, a_{n+1})\bigg\}\stackrel{d}{=}\left(a_{\pi(1)}, \ldots, a_{\pi(n+1)}\right).\]
For any $j = 1,\ldots, m - 1$, if $\pi(n+1)\in [n_{j-1}, n_{j})$, 
\[|\{i: a_{\pi(i)} = a_{\pi(n+1)}\}| = n_{j} - n_{j-1}, \quad |\{i: a_{\pi(i)} < a_{\pi(n+1)}\}| = n_{j-1} - 1,\]
and thus,
\[R_{2n+1} = n_{j-1} - 1 + \lceil (n_{j} - n_{j-1})U_{j}\rceil.\]
Similarly, if $\pi(n+1)\in [n_{m-1}, n_{m}]$,
\[R_{2n+1} = n_{m-1} - 1 + \lceil (n_{m} - n_{m-1} + 1)U_{j}\rceil.\]
For any $k$, let $j_k = \max\{j: n_{j}\le k\}$, and $\mathcal{I}_k$ be the set $\{n_{j_k - 1}, \ldots, n_{j_k} - 1\}$ if $j_k < m$ and $\{n_{j_k - 1}, \ldots, n_{j_k}\}$ otherwise. Then
\begin{align*}
&\P(R_{2n+1} = k\mid (S'_{(1)}, \ldots, S'_{(n+1)}) = (a_1, \ldots, a_{n+1}))\\
& = \P\left(\pi(n+1)\in \mathcal{I}_{k}, U_{1}\in \left(\frac{k - n_{j_k - 1}}{|\mathcal{I}_k|}, \frac{k + 1 - n_{j_k - 1}}{|\mathcal{I}_k|} \right]\right)\\
& = \P\left(\pi(n+1)\in \mathcal{I}_{k}\right)\P\left(U_{1}\in \left(\frac{k  - n_{j_k - 1}}{|\mathcal{I}_{k}|}, \frac{k + 1 - n_{j_k - 1}}{|\mathcal{I}_{k}|} \right]\right)\\
& = \frac{|\mathcal{I}_{k}|}{n+1} \frac{1}{|\mathcal{I}_{k}|} = \frac{1}{n+1}.
\end{align*}
Therefore, $R_{2n+1}$ is independent of $(S'_{(1)}, \ldots, S'_{(n+1)})$. The proof of Lemma \ref{lem:prds2} is then completed. 
\end{proof}



\subsection{Storey's correction does not break FDR control}\label{app:storey_BH}
 Given a p-value $p_i$ for the $i$-th null hypothesis, let $p_{(1)}\le \ldots \le p_{(m)}$ be the ordered statistics. Given a target FDR level $\alpha$ and a scalar $\lambda\in (0, 1)$, the rejection set of the Storey-BH procedure is
\[\mathcal{R} = \left\{i: p_i\le \frac{\alpha R}{m\hat{\pi}_0}, p_i < \lambda\right\},\]
where
\[\hat{\pi}_0 = \frac{1 + \sum_{i=1}^{m}I(p_i \ge \lambda)}{m (1 - \lambda)} \triangleq \frac{1 + A}{m(1 - \lambda)}\]
and
\[R = \max\left\{r: p_{(r)}\le \frac{\alpha r}{m\hat{\pi}_0}, p_{(r)} < \lambda\right\}.\]
The parameter $\lambda$ is often chosen as $0.5$, $\alpha$ or $1 - \alpha$.

We start with a novel FDR bound for this procedure applied to PRDS p-values.

\begin{thm}\label{thm:storey_BH_generic}
Assume that $(p_1, \ldots, p_n)$ is PRDS and each null p-value is super-uniform with an almost sure lower bound $p_{\min}\in [0, 1]$. Then
\[\E\left[\frac{|\mathcal{R} \cap \mathcal{H}_0|}{\max\{1, |\mathcal{R}|\} } \right]\le \alpha(1 - \lambda)\sum_{i\in \mathcal{H}_0}\E\left[\frac{1}{1 + A}\mid p_i\le p_{*}\right],\]
where
\[p_{*} = \max\left\{\frac{\alpha(1 - \lambda)}{m}, p_{\min}\right\}.\]
\end{thm}
\begin{proof}
Let
  \[V_i = I(H_i \text{ is rejected}) \le I\lb p_i \le \alpha(1 - \lambda)\frac{R}{1 + A}\rb.\]
  Then
  \begin{align*}
    \E\left[\frac{|\mathcal{R} \cap \mathcal{H}_0|}{\max\{1, |\mathcal{R}|\} } \right]
    &= \sum_{i\in \mathcal{H}_0}\E\left[\frac{V_i}{R\vee 1}\right] = \sum_{i\in \mathcal{H}_0}\sum_{r = 1}^{m}\frac{1}{r}\P\lb p_i \le \alpha(1 - \lambda)\frac{r}{1 + A}, R = r\rb\\
    & = \sum_{i\in \mathcal{H}_0}\sum_{r = 1}^{m}\sum_{a = 1}^{m}\frac{1}{r}\P\lb p_i \le \alpha(1 - \lambda)\frac{r}{1 + a}, R = r, A = a\rb.
  \end{align*}
   Let $r_0(a) = \max\{1, \lceil(1 + a)p_{\min} / (1 - \lambda)\alpha\rceil\}$. By definition, the summand for a given $a$ is non-zero only if $r\ge r_{0}(a)$. Thus,
  \begin{align*}
    \E\left[\frac{|\mathcal{R} \cap \mathcal{H}_0|}{\max\{1, |\mathcal{R}|\} } \right]
    & = \sum_{i\in \mathcal{H}_0}\sum_{a = 1}^{m}\sum_{r = r_{0}(a)}^{m}\frac{1}{r}\P\lb p_i \le \alpha(1 - \lambda)\frac{r}{1 + a}\rb\P\lb R = r, A = a\mid p_i \le \alpha(1 - \lambda)\frac{r}{1 + a}\rb\\
    & \stackrel{(i)}{\le} \sum_{i\in \mathcal{H}_0}\sum_{a = 1}^{m}\sum_{r = r_0(a)}^{m}\frac{1}{r}\cdot \alpha(1 - \lambda)\frac{r}{1 + a}\P\lb R = r, A = a\mid p_i \le \alpha(1 - \lambda)\frac{r}{1 + a}\rb\\
    & = \alpha(1 - \lambda)\sum_{i\in \mathcal{H}_0}\sum_{a = 1}^{m}\sum_{r = r_{0}(a)}^{m}\frac{1}{1 + a}\P\lb R = r, A = a\mid p_i \le \alpha(1 - \lambda)\frac{r}{1 + a}\rb,
  \end{align*}
  where (i) uses the super-uniformity of the null p-value. Let $\T$ denote the set of all possible values that $r / (1 + a)$ can take such that $\P(p_i \le \alpha(1 - \lambda)r / (1 + a)) > 0$, i.e.
  \[\T = \left\{\frac{r}{1 + a}: a\in \{1, \ldots, m\}, r\in \{r_{0}(a), \ldots, m\}, a + r \le m\right\}.\]
  Clearly, $\T$ is a finite set. Let $t_1\le t_2 \le \ldots \le t_M$ denote the elements of $\T$. It is easy to see that
  \begin{equation}
    \label{eq:t1}
    \alpha(1 - \lambda)t_1 \ge \max\left\{p_{\min}, \frac{\alpha(1 - \lambda)}{m}\right\} = p_{*}.
  \end{equation}
  Then,
  \begin{align*}
    \E\left[\frac{|\mathcal{R} \cap \mathcal{H}_0|}{\max\{1, |\mathcal{R}|\} } \right] 
    & \leq \alpha(1 - \lambda)\sum_{i\in \mathcal{H}_0}\sum_{a = 1}^{m}\sum_{r = r_{0}(a)}^{m}\frac{1}{1 + a}\P\lb R = r, A = a\mid p_i \le \alpha(1 - \lambda)\frac{r}{1 + a}\rb \\
    & = \alpha(1 - \lambda)\sum_{i\in \mathcal{H}_0} \sum_{j=1}^{M}\sum_{a = 1}^{m} \sum_{r = r_{0}(a)}^{m} \frac{1}{1 + a}\P\lb R = r, A = a\mid p_i \le \alpha(1 - \lambda)\frac{r}{1 + a}\rb \I\left[t_j = \frac{r}{1+a}\right]\\
    & = \alpha(1 - \lambda)\sum_{i\in \mathcal{H}_0} \sum_{j=1}^{M} \sum_{a = 1}^{m} \sum_{r = r_{0}(a)}^{m} \frac{1}{1 + a}\P\lb R = t_j(1+a), A = a\mid p_i \le \alpha(1 - \lambda) t_j \rb \I\left[t_j = \frac{r}{1+a}\right]\\
& = \alpha(1 - \lambda)\sum_{i\in \mathcal{H}_0}\sum_{j=1}^{M}\sum_{a = 1}^{m}\frac{1}{1 + a}\P\lb R = (1 + a)t_j, A = a\mid p_i \le \alpha(1 - \lambda)t_j\rb\\
         & = \alpha(1 - \lambda)\sum_{i\in \mathcal{H}_0}\sum_{j=1}^{M}\E\left[ \frac{I\{R = (1 + A)t_j\}}{1 + A}\mid p_i \le \alpha(1 - \lambda)t_j\right]\\
         & = \alpha(1 - \lambda)\sum_{i\in \mathcal{H}_0}\sum_{j=1}^{M}\bigg\{\E\left[ H_j(p)\mid p_i \le \alpha(1 - \lambda)t_j\right] - \E\left[ H_{j+1}(p)\mid p_i \le \alpha(1 - \lambda)t_j\right]\bigg\}\\
         & = \alpha(1 - \lambda)\sum_{i\in \mathcal{H}_0}\bigg\{\E\left[ H_1(p)\mid p_i \le \alpha(1 - \lambda)t_1\right] \\
    & \quad - \sum_{j=1}^{M-1}\lb\E\left[ H_{j+1}(p)\mid p_i \le \alpha(1 - \lambda)t_j\right] - \E\left[ H_{j+1}(p)\mid p_i \le \alpha(1 - \lambda)t_{j+1}\right]\rb\bigg\},
  \end{align*}
  where $p = (p_1,\ldots,p_m)$ and
  \[H_j(p) = \frac{I\{R\ge (1 + A)t_j\}}{1 + A}, \quad H_{M+1}(p) = 0.\]  
  Since $A$ is an increasing function of $p$ (element-wise) and $R$ is a decreasing function of $p$ (element-wise), $H_j(p)$ is decreasing in $p$ (element-wise). The PRDS property implies that for any $j = 1, \ldots, M - 1$,
  \[\E\left[ H_{j+1}(p)\mid p_i \le \alpha(1 - \lambda)t_j\right] - \E\left[ H_{j+1}(p)\mid p_i \le \alpha(1 - \lambda)t_{j+1}\right]\ge 0.\]
  Therefore,
    \begin{align*}
      \E\left[\frac{|\mathcal{R} \cap \mathcal{H}_0|}{\max\{1, |\mathcal{R}|\} } \right] & \le \alpha(1 - \lambda)\sum_{i\in \mathcal{H}_0}\E\left[ H_1(p)\mid p_i \le \alpha(1 - \lambda)t_1\right]\\
           & \le \alpha(1 - \lambda)\sum_{i\in \mathcal{H}_0}\E\left[ \frac{1}{1 + A}\mid p_i \le \alpha(1 - \lambda)t_1\right]\\
      &\le \alpha(1 - \lambda)\sum_{i\in \mathcal{H}_0}\E\left[ \frac{1}{1 + A}\mid p_i \le p_{*}\right],
    \end{align*}
    where the last step follows from \eqref{eq:t1}, the PRDS property, and the fact that $p\mapsto 1 / (1 + A)$ is decreasing element-wise.
\end{proof}

To prove Theorem~\ref{thm:storey_BH}, we present an additional lemma. 
\begin{lemma}\label{lem:inverse_binomial}[Lemma 1 from~\cite{benjamini2006adaptive}]
  If $Y\sim \mathrm{Binom}(k - 1, p)$, then $\E[1 / (1 + Y)]\le 1 / kp$.
\end{lemma}

\begin{proof}[Proof of Theorem~\ref{thm:storey_BH}]
As in the proof of Theorem~\ref{thm:combination_test}, since $\hat{s}(X)$ is continuously distributed, we can assume $\hat{s}(X)\sim \Unif([0, 1])$ without loss of generality. We write $p_i$ instead of $\hat{u}^\marg(X_{2n+i})$ and $S_j$ instead of $\hat{s}(X_{n+j})$. Then
\[p_j = \frac{1 + \sum_{i=1}^{n}I(S_{i}\le S_{n+j})}{n + 1}.\]
Then $p_j\ge 1 / (n + 1)$ almost surely. Let $m_0 = |\mathcal{H}_0|$ and we assume that $\mathcal{H}_0 = \{1, \ldots, m_0\}$ without loss of generality. Since $p = (p_1, \ldots, p_m)$ are PRDS and exchangeable, Theorem~\ref{thm:generic} implies that
\[\E\left[\frac{|\mathcal{R} \cap \mathcal{H}_0|}{\max\{1, |\mathcal{R}|\} } \right]\le \alpha(1 - \lambda)m_0\E\left[\frac{1}{1 + A}\mid p_1\le \max\left\{\frac{1}{n+1}, \frac{\alpha(1 - \lambda)}{m}\right\}\right].\]
Since $1 / (1 + A)$ is decreasing in $p$, using the PRDS property again, we have
\begin{equation}
  \label{eq:fdr_conformal}
  \E\left[\frac{|\mathcal{R} \cap \mathcal{H}_0|}{\max\{1, |\mathcal{R}|\} } \right]\le \alpha(1 - \lambda)m_0\E\left[\frac{1}{1 + A}\mid p_1\le \frac{1}{n+1}\right] = \alpha(1 - \lambda)m_0\E\left[\frac{1}{1 + A}\mid p_1 = \frac{1}{n+1}\right].
\end{equation}
Let $A_0 = \sum_{j=2}^{m_0}I(p_j \ge \lambda)$. Then
\[\E\left[\frac{1}{1 + A}\mid p_1 = \frac{1}{n+1}\right]\le \E\left[\frac{1}{1 + A_0}\mid p_1 = \frac{1}{n+1}\right].\]
Let $S_{(1)}\le S_{(2)}\le \ldots \le S_{(n+1)}$ denote the order statistics of $S_1, \ldots, S_{n+1}$ and $R_{n+1}$ denote the rank of $S_{n+1}$. Since $S_1\sim \mathrm{Unif}([0, 1])$, there is no tie almost surely.

Now we compute
\begin{equation}
  \label{eq:A0}
  \E\left[\frac{1}{1 + A_0}\mid p_1 = \frac{1}{n+1}, S_{(1)}, \ldots, S_{(n+1)}\right] = \E\left[\frac{1}{1 + A_0}\mid R_{n+1} = 1, S_{(1)}, \ldots, S_{(n+1)}\right].
\end{equation}
By definition, 
\[p_2, \ldots, p_{m_0} \mid S_1, \ldots, S_{n+1} \stackrel{i.i.d.}{\sim}\frac{1 + \sum_{j=1}^{n}I(S_j\le U)}{n + 1} \]
where $U\sim \mathrm{Unif}([0, 1])$. Note that there is a bijection between $(S_1, \ldots, S_{n+1})$ and $(S_{(1)}, \ldots, S_{(n+1)}, R_{1}, \ldots, R_{n+1})$ for vectors without ties. The above distributional equivalence can be rewritten as
\[p_2, \ldots, p_{m_0} \mid R_1, \ldots, R_{n+1}, S_{(1)}, \ldots, S_{(n+1)}  \stackrel{i.i.d.}{\sim}\frac{1 + \sum_{j=1}^{n+1}I(S_{(j)}\le U) - I(S_{(R_{n+1})} \le U)}{n + 1}.\]
Since the RHS does not depend on $(R_1, \ldots, R_n)$, $(p_2, \ldots, p_{m_0})$ is independent of $(R_1, \ldots, R_n)$ conditional on $(R_{n+1}, S_{(1)}, \ldots, S_{(n+1)})$. As a result,
\[p_2, \ldots, p_{m_0} \mid R_{n+1} = 1, S_{(1)}, \ldots, S_{(n+1)} \stackrel{i.i.d.}{\sim}\frac{1 + \sum_{j=2}^{n+1}I(S_{(j)}\le U)}{n + 1}.\]
Recall $K = (n + 1)\lambda\in \mathbb{Z}$. Then
\begin{align*}
  \P\lb p_2 \ge \lambda \mid R_{n+1} = 1, S_{(1)}, \ldots, S_{(n+1)}\rb
  & = \P\lb \sum_{j=2}^{n+1}I(S_{(j)}\le U)\ge K - 1 \mid S_{(2)}, \ldots, S_{(n+1)}\rb\\
  & = \P\lb U \ge S_{(K)}\mid S_{(2)}, \ldots, S_{(n+1)}\rb\\
  & = 1 - S_{(K)}.
\end{align*}
Therefore,
\[I(p_2 \ge \lambda), \ldots, I(p_{m_0} \ge \lambda) \mid R_{n+1} = 1, S_{(1)}, \ldots, S_{(n+1)} \stackrel{i.i.d.}{\sim}\mathrm{Ber}\lb 1 - S_{(K)}\rb.\]
This implies that
\[A_0\mid R_{n+1} = 1, S_{(1)}, \ldots, S_{(n+1)}\sim \mathrm{Binom}\lb m_0 - 1, 1 - S_{(K)}\rb.\]
By Lemma~\ref{lem:inverse_binomial},
\[\E\left[\frac{1}{1 + A_0}\mid R_{n+1} = 1, S_{(1)}, \ldots, S_{(n+1)}\right]\le \frac{1}{m_0 \{1 - S_{(K)}\}}.\]
Since $R_{n+1}$ is independent of $(S_{(1)}, \ldots, S_{(n+1)})$,
\begin{equation}
  \label{eq:EA0}
  \E\left[\frac{1}{1 + A_0}\mid R_{n+1} = 1\right] \le \E\left[\frac{1}{m_0 \{1 - S_{(K)}\}}\right].
\end{equation}
By symmetry and the property of order statistics,
\[1 - S_{(K)}\stackrel{d}{=} S_{(n + 2 - K)}\sim \mathrm{Beta}(n + 2 - K, K).\]
Thus,
\begin{align}
  \E\left[\frac{1}{1 - S_{(K)}}\right]
  &= \int_{0}^{1}\frac{1}{x}\frac{\Gamma(n + 2)}{\Gamma(n + 2 - K)\Gamma(K)}x^{n + 1 - K}(1 - x)^{K - 1}dx\nonumber\\
  & = \int_{0}^{1}\frac{\Gamma(n + 2)}{\Gamma(n + 2 - K)\Gamma(K)}x^{n - K}(1 - x)^{K - 1}dx\nonumber\\
  & = \frac{\Gamma(n + 2)\Gamma(n + 1 - K)}{\Gamma(n + 2 - K)\Gamma(n + 1)}\nonumber\\
  & = \frac{n + 1}{n + 1 - K} = \frac{1}{1 - \lambda}.\label{eq:ES}
\end{align}
Putting \eqref{eq:fdr_conformal}, \eqref{eq:EA0} and \eqref{eq:ES} together, we prove the result.
\end{proof}

\subsection{Conditional p-value adjustment}

\begin{proof}[Proof of Theorem~\ref{thm:generic}]
Let $S_{i} = \hat{s}(X_{n+i})$ for $i = 1, \ldots, n$ with $F^{-}(t) = \P[S_{i} < t\mid \mathcal{D}^\textnormal{train}]$, and $S_{(1)}\le S_{(2)}\le \ldots \le S_{(n)}$ be the order statistics. Note that here we condition on the training data in $\mathcal{D}^\textnormal{train}$, which makes the $S_{i}$ are independent of another. Then it is easy to see that 
\[(F^{-}(S_{(1)}), \ldots, F^{-}(S_{(n)}))\preceq (U_{(1)}, \ldots, U_{(n)}),\]
where $\preceq$ denotes the entry-wise stochastic dominance in the sense that $(A_1, \ldots, A_n) \preceq (B_1, \ldots, B_n)$ iff
\[\P[A_1 \le z_1, \ldots, A_n \le z_n] \ge \P[B_1 \le z_1, \ldots, B_n \le z_n], \quad \forall (z_1, \ldots, z_n)\in \R^n.\]
When $F$ is continuous, the equality in distribution holds. Let $\mathcal{E}_{n}$ denote the event on which $F^-(S_{(i)})\le b_{i}$ for all $i = 1, \ldots, n$. Then 
\[\P[ \mathcal{E}_{n}] \ge 1 - \delta.\]
Now we prove the following claim, which directly yields the theorem:
\begin{equation}\label{eq:goal_generic}
\P\left[\hat{u}^\cond(X_{2n+1})\le t\mid \mathcal{D}\right]\le t, \quad \forall t\in [0, 1], \quad \text{if } \mathcal{E}_{n} \text{ occurs}.
\end{equation}
Note that the image of $\hat{u}^\cond$ is $\{b_1, \ldots, b_n, 1\}$, it remains to prove \eqref{eq:goal_generic} with $t \in \{b_1, \ldots, b_n, 1\}$. When $t = 1$, it clearly holds. When $t = b_i$,
\[\hat{u}^\cond(X_{2n+1})\le b_i\Longleftrightarrow \hat{u}^\marg(X_{2n+1})\le \frac{i}{n+1}\Longleftrightarrow \hat{s}(X_{2n+1}) < S_{(i)}.\]
Thus, 
\[\P\left[\hat{u}^\cond(X_{2n+1})\le b_i\mid \mathcal{D}\right] = \P\left[\hat{s}(X_{2n+1}) < S_{(i)}\mid \mathcal{D}\right] = F^-(S_{(i)}).\]
By definition of $\mathcal{E}_{n}$, \eqref{eq:goal_generic} holds for all $t\in \{b_1, \ldots, b_n\}$.

\end{proof}

\subsection{Simultaneous confidence bounds for the false positive rate}

\begin{proof}[Proof of Proposition~\ref{prop:ucb-fpr}]
Note that $h(i / n) = b_{\lceil i + i/n\rceil} = b_{i+1}$ where we let $b_{n+1} = 1$ for  convenience. Then, the event that $F(Z_{(i)})\le h((i - 1) / n) = b_{i}$ for all $i \in \{1, \ldots, n\}$ occurs with probability at least $1 - \delta$, where $Z_{(1)}\le \ldots\le Z_{(n)}$ are the order statistics. Under this event, for any $z \in [Z_{(i-1)}, Z_{(i)})$, where we let $Z_{(0)} = \infty$ and $Z_{(n+1)} = \infty$ for convenience, $\hat{F}_{n}(z) = (i - 1) / n$ and thus
\[F(z)\le F(Z_{(i)})\le b_{i} = h\left(\hat{F}_{n}(z)\right).\]
On the other hand, if $h: [0, 1]\rightarrow [0, 1]$ is a function such that $h(\hat{F}_{n}(z))$ is a uniform upper confidence band of $F$ for any CDF $F$, then \eqref{eq:confidence_band} holds with $b_{i} = h(i / n)$.
\end{proof}

\section{Power analysis of Fisher's combination test}\label{sec:power_fisher}

In this section, we investigate the effective $\alpha$-level of Fisher's combination test applied to calibration-conditional conformal p-values $\hat{u}_i^\cond \equiv h\circ \hat{u}_i^\marg$, for different adjustment functions $h$. To be self-contained, we summarize the three calibration-conditional adjustments as follows.
\begin{itemize}
\item Asymptotic adjustment:
  \begin{equation}
    \label{eq:asym}
    h^\asym\lb \frac{i}{n+1}\rb = \min\left\{\frac{i}{n} + c_n(\delta)\frac{\sqrt{i(n - i)}}{n\sqrt{n}}, 1\right\}, 
  \end{equation}
  \[\text{where } c_n(\delta) = \frac{-\log[-\log(1-\delta)]+2\log\log n + (1/2) \log \log \log n - (1/2) \log\pi }{\sqrt{2 \log \log n}}.\]
\item DKWM adjustment:
  \[h^\dkwm\lb \frac{i}{n+1}\rb = \min\left\{\frac{i}{n+1} + \sqrt{\frac{\log(2/\delta)}{2n}}, 1\right\}.\]
\item Simes adjustment:
  \[h^\simes\lb \frac{n+1 - i}{n+1}\rb = 1 - \delta^{1/k}\left(\frac{i\cdots (i - k + 1)}{n\cdots (n - k + 1)}\right)^{1/k},\]
  where $k$ is chosen to be $0.5 n$ in the experiments.
\end{itemize}

Throughout the section, we will treat $\delta\in (0, 1)$ as a constant that does not vary with $n$ or $m$, though it is not hard to recover the dependence on $\delta$ from the proofs. As a result, the big-O notation could hide constants that solely depend on $\delta$. 

\subsection{Asymptotic adjustment}\label{subsec:fisher_asym}
\subsubsection{Monotonicity of $h^\asym$}\label{subsubsec:fisher_asym_monotonicity}
  For notational convenience, we write $a_n$ for $c_n(\delta) / \sqrt{n}$ throughout the subsection. It should be kept in mind that $a_{n}$ depends on $\delta$. We first prove that $h^\asym$ is non-decreasing.
  \begin{prop}\label{prop:hasym_increasing}
  $h^\asym$ is non-decreasing for any $n$ and $\delta$. Furthermore,
  \[h^\asym\lb \frac{i}{n+1}\rb = 1\text{ for any }i\ge t_n, \text{ where }t_n = \left\lceil \frac{n}{a_n^2 + 1}\right\rceil.\]
\end{prop}
\begin{proof}
  Let $g_n(x) = x + a_n\sqrt{x(1 - x)}$. By definition,
  \[h^\asym\lb \frac{i}{n+1}\rb = g_n\lb\frac{i}{n}\rb.\]
  It is left to prove $\min\{g(x), 1\}$ is non-decreasing on $[0, 1]$. Taking the derivative, we obtain that
  \[g_n'(x) = 1 + a_n \frac{1 - 2x}{2\sqrt{x(1 - x)}}.\]
  Clearly, $g_n'(x)\ge 0$ for any $x \le 1/2$. When $x > 1/2$,
  \begin{align*}
    g_n'(x) > 0 & \Longleftrightarrow 2\sqrt{x(1 - x)} \ge a_n(2x - 1)\\
              & \Longleftrightarrow 4x(1 - x) \ge a_n^2(2x - 1)^2\\
              & \Longleftrightarrow (a_n^2 + 1)(2x - 1)^2 \le 1\\
              & \Longleftrightarrow x \le \frac{1}{2}\lb 1 + \sqrt{\frac{1}{a_n^2 + 1}}\rb\equiv d_n.
  \end{align*}
  As a result, $g_n(x)$ is increasing on $[0, d_n]$ and decreasing on $[d_n, 1]$. On the other hand,
  \[g_n\lb\frac{1}{a_n^2 + 1}\rb = 1,\]
  and
  \[\frac{1}{a_n^2 + 1} - d_n = \lb \sqrt{\frac{1}{a_n^2 + 1}} + \frac{1}{2}\rb\lb \sqrt{\frac{1}{a_n^2 + 1}}- 1\rb < 0.\]
  Therefore, $g_n(x)$ is increasing on $[0, 1/(a_n^2 + 1)]$ with $g_n(1/(a_n^2 + 1)) = 1$. On $[1/(a_n^2 + 1), 1]$,
  \[g_n(x)\ge \min\{g_n(1/(a_n^2 + 1)), g_n(1)\} = 1.\]
  Therefore, $\min\{g_n(x), 1\}$ is increasing and, when $i\ge t_n$,
  \[h^\asym\lb\frac{i}{n+1}\rb = 1.\]
\end{proof}

\subsubsection{Mean of Fisher's combination statistic}
As a stepping stone to analyze the effective $\alpha$-level, we will first compute the mean of Fisher's combination statistic under the null, which roughly measures the conservatism of the test.
\begin{lemma}\label{lem:integral}
  Let $\int f(x)dx$ denote the indefinite integral of $f(x)$. Then
  \[\int \frac{dx}{\sqrt{x(1-x)} + a_n(1 - x)} = \frac{2}{a_n^2 + 1}\left\{ \arcsin(\sqrt{x}) - a_n \log\lb \sqrt{x} + a_n\sqrt{1 - x}\rb\right\} + \mathrm{Const}.\]
\end{lemma}
\begin{proof}
  \begin{align*}
    &\int \frac{dx}{\sqrt{x(1-x)} + a_n(1 - x)} \stackrel{x = \sin^2\theta}{=}\int \frac{d\sin^2\theta}{\sin\theta \cos\theta + a_n\cos^2\theta} = 2\int \frac{\sin\theta}{\sin\theta + a_n \cos\theta}d\theta.
  \end{align*}
  Let
  \[f_1(\theta) = \int \frac{\sin\theta}{\sin\theta + a_n \cos\theta}d\theta, \quad f_2(\theta) = \int \frac{\cos\theta}{\sin\theta + a_n \cos\theta}d\theta.\]
  Then
  \[f_1(\theta) + a_n f_2(\theta) = \int 1d\theta = \theta + \mathrm{Const},\]
  and
  \[f_2(\theta) - a_n f_1(\theta) = \int \frac{\cos\theta - a_n \sin\theta}{\sin\theta + a_n \cos\theta}d\theta = \log(\sin\theta + a_n \cos\theta) + \mathrm{Const}.\]
  Therefore,
  \[f_1(\theta) = \frac{1}{a_n^2 + 1}\lb \theta - a_n\log(\sin\theta + a_n \cos\theta)\rb + \mathrm{Const},\]
  which implies the lemma.
\end{proof}

\begin{thm}[Part (a) of Theorem~\ref{theorem:fisher-mean}]\label{thm:fisher_asym}
  \[\E_{H_0}[-\log (h^\asym \circ \hat{u}^\marg)] = 1 - \frac{\pi}{2}\frac{c_n(\delta)}{\sqrt{n}} + O\lb \frac{(\log n)(\log \log n)}{n}\rb.\]
\end{thm}
\begin{proof}
  Let $g_n(x) = x + a_n \sqrt{x(1 - x)}$. Since $-\log(1) = 0$, 
  \[\E_{H_0}[-\log (h^\asym \circ \hat{u}^\marg)] = \frac{1}{n+1}\sum_{i=1}^{n+1}h^\asym\lb\frac{i}{n+1}\rb = \frac{1}{n+1}\sum_{i=1}^{t_n - 1} -\log \left\{g_n\lb\frac{i}{n}\rb\right\}.\]
  We have shown in the proof of Proposition \ref{prop:hasym_increasing} that $g_n(x)$ is increasing on $[0, t_n / n]$. Thus,
  \begin{equation}
    \label{eq:integral_bound}
    \frac{n+1}{n}\E_{H_0}[-\log (h^\asym \circ \hat{u}^\marg)]\in \left[\int_{1/n}^{t_n / n}-\log \{g_n(x)\}dx, \int_{0}^{(t_n - 1) / n}-\log \{g_n(x)\}dx\right].
  \end{equation}
  Now we calculate the indefinite integral of $-\log\{g_n(x)\}$. Using integration by parts,
  \begin{align*}
    &\int -\log \{g_n(x)\}dx\\
    &= -x\log \{g_n(x)\} + \int x\frac{1 + a_n (1 - 2x) / 2\sqrt{x(1- x)}}{x + a_n\sqrt{x(1 - x)}}dx\\
    &= -x\log \{g_n(x)\} + \int x\frac{ 2\sqrt{x(1- x)} + a_n (1 - 2x)}{2x\sqrt{x(1- x)} + 2a_n x(1 - x)}dx\\
    & = -x\log \{g_n(x)\} + \int\frac{2\sqrt{x(1 - x)} + a_n (1 - 2x)}{2\sqrt{x(1 - x)} + 2a_n(1 - x)}dx\\
    & = -x\log \{g_n(x)\} + x - \frac{a_n}{2}\int\frac{1}{\sqrt{x(1 - x)} + a_n(1 - x)}dx\\
    & = -x\log \{g_n(x)\} + x - \frac{a_n}{a_n^2 + 1}\arcsin(\sqrt{x}) + \frac{a_n^2}{a_n^2 + 1}\log\lb \sqrt{x} + a_n\sqrt{1 - x}\rb + \mathrm{Const} \,\,\text{(Lemma \ref{lem:integral})}\\
    & \equiv G_n(x) + \mathrm{Const}.
  \end{align*}
  Applying this to \eqref{eq:integral_bound}, by Newton-Leibniz formula,
  \[\int_{1/n}^{t_n / n}-\log \{g_n(x)\}dx = G_n(t_n / n) - G_n(1 / n).\]
  By definition,
  \[1 - \frac{t_n}{n} = O(a_n^2), \quad \frac{1}{n} = O\lb a_n^2\rb.\]
  Then
  \begin{align*}
    -\frac{t_n}{n}\log\left\{\frac{t_n}{n} + a_n \sqrt{\frac{t_n}{n}}\lb 1 - \frac{t_n}{n}\rb\right\} &= O\lb 1 - \frac{t_n}{n} - a_n \sqrt{\frac{t_n}{n}}\lb 1 - \frac{t_n}{n}\rb\rb\\
    & = O(a_n^2) = O\lb\frac{(\log n)(\log \log n)}{n}\rb,
  \end{align*}
  and
  \begin{equation*}
    -\frac{1}{n}\log\left\{\frac{1}{n} + a_n\sqrt{\frac{n-1}{n^2}}\right\} = O\lb\frac{(\log n)(\log\log n)}{n}\rb.
  \end{equation*}
  Thus,
  \begin{equation}
    \label{eq:integral_1}
    -x\log \{g_n(x)\}\bigg|_{1/n}^{t_n/n} = O\lb a_n^2 + \frac{(\log n)(\log\log n)}{n}\rb =O\lb\frac{(\log n)(\log\log n)}{n}\rb.
  \end{equation}
  Similarly,
  \begin{equation}
    \label{eq:integral_2}
    x\bigg|_{1/n}^{t_n/n} = 1 + O\lb a_n^2\rb = 1 + O\lb\frac{(\log n)(\log\log n)}{n}\rb.
  \end{equation}
  Next, 
  \begin{equation*}
    \frac{a_n}{a_n^2 + 1}\arcsin(\sqrt{1/n}) = O\lb\frac{a_n}{\sqrt{n}}\rb = O\lb\frac{(\log n)(\log\log n)}{n}\rb,
  \end{equation*}
  and
  \begin{align}
    -\frac{a_n}{a_n^2 + 1}\arcsin(\sqrt{t_n/n}) &= -\frac{a_n}{a_n^2+1}\lb \arcsin(1) - \int_{\sqrt{t_n/n}}^{1}\frac{dx}{\sqrt{1 - x^2}}\rb\nonumber\\
                              &\in -\frac{\pi}{2}\frac{a_n}{a_n^2 + 1} + \frac{a_n}{a_n^2 + 1}\sqrt{1 - \frac{t_n}{n}}\cdot \left[ 1, \frac{1}{1 - t_n / n}\right]\nonumber\\
    & = -\frac{\pi}{2}a_n + O(a_n^2) = -\frac{\pi}{2}a_n + O\lb\frac{(\log n)(\log\log n)}{n}\rb.\nonumber\label{eq:integral_4}
  \end{align}
  Thus,
  \begin{equation}
    \label{eq:integral_3}
    -\frac{a_n}{a_n^2 + 1}\arcsin(\sqrt{x})\bigg|_{1/n}^{t_n/n} = -\frac{\pi}{2}a_n + O\lb \frac{(\log n)(\log\log n)}{n}\rb.
  \end{equation}
  Finally,
  \begin{align*}
    &\frac{a_n^2}{a_n^2 + 1}\log\lb \sqrt{\frac{t_n}{n}} + a_n\sqrt{1 - \frac{t_n}{n}}\rb = O\lb a_n^2\left\{1 - \sqrt{\frac{t_n}{n}} - a_n\sqrt{1 - \frac{t_n}{n}}\right\}\rb\\
    &= O(a_n^4) = O\lb\frac{(\log n)(\log\log n)}{n}\rb,
  \end{align*}
  and
  \begin{align*}
    &\frac{a_n^2}{a_n^2 + 1}\log\lb \sqrt{\frac{1}{n}} + a_n\sqrt{1 - \frac{1}{n}}\rb = O(a_n^2 \log n) = O\lb\frac{(\log n)(\log\log n)}{n}\rb.
  \end{align*}
  As a result,
  \begin{equation}
    \label{eq:integral_4}
    \frac{a_n^2}{a_n^2 + 1}\log\lb \sqrt{x} + a_n\sqrt{1 - x}\rb\bigg|_{1/n}^{t_n/n} = O\lb\frac{(\log n)(\log\log n)}{n}\rb.
  \end{equation}
  Putting \eqref{eq:integral_1} - \eqref{eq:integral_4} together, we obtain that
  \[\int_{1/n}^{t_n / n}-\log \{g_n(x)\}dx = 1 - \frac{\pi}{2}a_n + O\lb\frac{(\log n)(\log \log n)}{n}\rb.\]
  Similarly,
  \[\int_{0}^{(t_n - 1) / n}-\log \{g_n(x)\}dx = 1 - \frac{\pi}{2}a_n + O\lb\frac{(\log n)(\log \log n)}{n}\rb.\]
  The proof is then completed.
\end{proof}

\subsubsection{Effective $\alpha$-level}
The next result shows the distributional approximation for Fisher's combination statistic. The proof is involved and thus relegated to Section \ref{subsec:lem_fisher_asym_CLT}.
\begin{lemma}\label{lem:fisher_asym_CLT}
  Under the global null, as $m, n\rightarrow \infty$,
  \begin{align*}
    &\bigg|\P_{H_0}\lb \sum_{i=1}^{m}-\log\lb h^{\asym}\circ \hat{u}^\marg_i\rb\ge \sqrt{m}\sqrt{1 + \frac{m}{n}}z_{1- \alpha} + m \left\{ 1 - \frac{\pi}{2}\frac{c_n(\delta)}{\sqrt{n}}\right\}\rb - \alpha\bigg|
     = O\lb \frac{1}{\sqrt{m}} + \frac{\log^6 n}{\sqrt{n}}\rb.
  \end{align*}
\end{lemma}

We apply Lemma \ref{lem:fisher_asym_CLT} to derive the effective $\alpha$-level in two practically relevant regimes.

\begin{thm}\label{thm:fisher_asym_effective_alpha}
  Assume that $m = \gamma n$ for some $\gamma \in (0, \infty)$. The type-I error of Fisher's combination test applied to $h^\asym\circ \hat{u}_i^\marg$'s
  \begin{align*}
    &\P_{H_0}\lb\sum_{i=1}^{m} -2\log(h^\asym \circ \hat{u}_i^\marg)\ge \chi^2(2m; 1 - \alpha)\rb\\
    &= (1 + o(1))\cdot C(\alpha; \gamma)\exp\left\{-\frac{\pi^2 \gamma}{4(1 + \gamma)}\log\log n - \frac{\pi\sqrt{\gamma}z_{1-\alpha}}{\sqrt{2}(1 + \gamma)}\sqrt{\log \log n} - \frac{1}{2}\log\log\log n\right\}\\
    & = \frac{1}{(\log n)^{e(\gamma)(1 + o(1))}},
  \end{align*}
  where
  \[C(\alpha; \gamma) = \sqrt{\frac{2(1 + \gamma)\exp\{-z_{1-\alpha}^2 / (1 + \gamma)\}}{\pi^2\gamma}}, \quad e(\gamma) = \frac{\pi^2 \gamma}{4(1 + \gamma)}.\]
\end{thm}
\begin{remark}
  When $\gamma = 1, \alpha = 0.05$,
  \[C(\alpha; \gamma)\approx 0.32, \quad e(\gamma)\approx 1.23,\]
  and thus the type-I error is approximately
  \[0.32(\log n)^{-1.23} \cdot \exp\left\{-1.83\sqrt{\log\log n} - 0.5\log\log \log n\right\}.\]
\end{remark}
\begin{proof}
  Choose $\alpha_n\in (0, 1)$ such that
  \[\sqrt{1 + \gamma}z_{1 - \alpha_{n}} - \frac{\pi\sqrt{\gamma}}{2} c_{n}(\delta) = \frac{\chi^2(2m; 1 - \alpha) - 2m}{2\sqrt{m}} \triangleq A_m.\]
  The standard CLT implies that
  \begin{equation}
    \label{eq:chi_normal_approx}
    A_{m} = z_{1 - \alpha} + O\lb\frac{1}{\sqrt{m}}\rb = z_{1 - \alpha} + O\lb\frac{1}{\sqrt{n}}\rb.
  \end{equation}
  Then
  \begin{equation}
    \label{eq:alphan}
    z_{1 - \alpha_{n}} = \frac{A_{m}}{\sqrt{1 + \gamma}} + \frac{\pi}{2}\sqrt{\frac{\gamma}{1 + \gamma}}c_{n}(\delta) = \frac{z_{1-\alpha}}{\sqrt{1+\gamma}} + \frac{\pi}{\sqrt{2}}\sqrt{\frac{\gamma}{1 + \gamma}}\sqrt{\log\log n} + o(1).
  \end{equation}
  Since $z_{1 - \alpha_{n}}\rightarrow \infty$,
  \begin{align*}
    \alpha_{n}
    &= 1 - \Phi(z_{1 - \alpha_{n}}) = (1 + o(1))\cdot \frac{1}{z_{1 - \alpha_{n}}}\exp\left\{-\frac{z_{1 - \alpha_{n}}^2}{2}\right\}\\
    & = (1 + o(1))\cdot \frac{C(\alpha; \gamma)}{\sqrt{\log \log n}}\exp\left\{-\frac{\pi^2 \gamma}{4(1 + \gamma)}\log\log n - \frac{\pi\sqrt{\gamma}z_{1-\alpha}}{\sqrt{2}(1 + \gamma)}\sqrt{\log \log n}\right\}\\
    & = (1 + o(1))\cdot C(\alpha; \gamma)\exp\left\{-\frac{\pi^2 \gamma}{4(1 + \gamma)}\log\log n - \frac{\pi\sqrt{\gamma}z_{1-\alpha}}{\sqrt{2}(1 + \gamma)}\sqrt{\log \log n} - \frac{1}{2}\log\log\log n\right\}.
  \end{align*}
  Clearly, $\alpha_{n}$ decays more slowly than any polynomial of $1/n$. By Lemma \ref{lem:fisher_asym_CLT},
  \begin{align*}
    &\P_{H_0}\lb\sum_{i=1}^{m} -2\log(h^\asym \circ \hat{u}_i^\marg)\ge \chi^2(2m; 1 - \alpha)\rb\\
    & = (1 + o(1))\cdot C(\alpha; \gamma)\exp\left\{-\frac{\pi^2 \gamma}{4(1 + \gamma)}\log\log n - \frac{\pi\sqrt{\gamma}z_{1-\alpha}}{\sqrt{2}(1 + \gamma)}\sqrt{\log \log n} - \frac{1}{2}\log\log\log n\right\}.
  \end{align*}
\end{proof}

\begin{thm}\label{thm:fisher_asym_effective_alpha_decay}
  Assume that $m \rightarrow \infty$ and $m = o(n / \log \log n)$. The type-I error of Fisher's combination test applied to $h^\asym\circ \hat{u}_i^\marg$'s
  \begin{align*}
    &\P_{H_0}\lb\sum_{i=1}^{m} -2\log(h^\asym \circ \hat{u}_i^\marg)\ge \chi^2(2m; 1 - \alpha)\rb = \alpha + o(1).
  \end{align*}
\end{thm}
\begin{proof}
  Adopting the notation utilized in the proof of Theorem \ref{thm:fisher_asym_effective_alpha}, by \eqref{eq:alphan},
  \begin{align*}
    z_{1 - \alpha_{n}} &= \frac{A_{m}}{\sqrt{1 + m/n}} + \frac{\pi}{2}\sqrt{\frac{m/n}{1 + m/n}}c_{n}(\delta)\\
                       & = \frac{z_{1-\alpha}}{\sqrt{1+m/n}} + \frac{\pi}{\sqrt{2}}\sqrt{\frac{m/n}{1 + m/n}}\sqrt{\log\log n}(1 + o(1))\\
    & = z_{1- \alpha} + o(1).
  \end{align*}
  This implies $\alpha_{n} = \alpha + o(1)$. By Lemma \ref{lem:fisher_asym_CLT},
  \[\bigg|\P_{H_0}\lb\sum_{i=1}^{m} -2\log(h^\asym \circ \hat{u}_i^\marg)\ge \chi^2(2m; 1 - \alpha)\rb - \alpha_n\bigg| = O\lb \frac{1}{\sqrt{m}} + \frac{\log^6 n}{\sqrt{n}}\rb = o(1).\]
  The proof is then completed.
\end{proof}

\subsubsection{Proof of Lemma \ref{lem:fisher_asym_CLT}}\label{subsec:lem_fisher_asym_CLT}
\begin{lemma}\label{lem:cond_moments}
  There exist universal constants $C, c > 0$ and a constant $C(\delta) > 0$ that only depends on $\delta$ such that, with probability $1 - \exp\{-c(\log n)^2\}$,
  \begin{equation}
    \label{eq:k=12}
    \sum_{k=1}^{2}\bigg|\E_{H_0}\left[|-\log (h^\asym \circ \hat{u}^\marg)|^k\mid \D\right] - k\bigg| \le \frac{C(\delta)\log^6 n}{\sqrt{n}},
  \end{equation}
  and
  \begin{equation}
    \label{eq:k=34}
    \sum_{k=3}^{4}\bigg|\E_{H_0}\left[|-\log (h^\asym \circ \hat{u}^\marg)|^k\mid \D\right]\bigg|\le C.
  \end{equation}
\end{lemma}
\begin{proof}
  Following \cite{wainwright2019high}, we call a random variable $V$ \emph{sub-exponential} with parameters $(\nu, b)$ if
  \[\E[e^{\lambda(V - \E[V])}]\le e^{\nu^2\lambda^2 / 2}, \quad \text{for all }\lambda < 1 / b.\]
  If $V \sim \mathrm{Exp}(1)$, then for any $\lambda < 1/2$, it is easy to verify that
  \[\E[e^{\lambda(V - \E[V])}] = \frac{e^{-\lambda}}{1 - \lambda}\le e^{2\lambda^2}.\]
  Thus, $V$ is sub-exponential with parameters $(2, 2)$.

  Let $G(p) = -\log (h^\asym(p))$. By \eqref{eq:EGk} in the proof of Theorem \ref{thm:combination_test},
  \begin{equation}
    \label{eq:E_log_D}
    \E_{H_0}\left[|-\log (h^\asym \circ \hat{u}^\marg)|^k\mid \D\right] = \frac{\frac{1}{n+1}\sum_{i=1}^{n+1}G^k\left(\frac{i}{n+1}\right)V_{i}}{\frac{1}{n+1}\sum_{i=1}^{n+1}V_{i}},
  \end{equation}
  where $V_1, \ldots, V_{n+1}\stackrel{i.i.d.}{\sim} \mathrm{Exp}(1)$. Let $\event_n$ denote the event that
  \begin{align}
    &\sum_{k=0}^{4}\bigg|\frac{1}{n+1}\sum_{i=1}^{n+1}G^k\left(\frac{i}{n+1}\right)(V_{i} - 1)\bigg|\le \frac{\log^6 n}{\sqrt{n}}.    \label{eq:event_En}
  \end{align}
  Note that
  \[G\lb \frac{i}{n+1}\rb\le \log n, \quad i = 1, \ldots, n+1.\]
  Then $G^k\left(\frac{i}{n+1}\right)(V_{i} - 1)$ is exponential with parameters $(2\log^4 n, 2\log^4 n)$ for $k\le 4$. By Bernstein's inequality \citep[e.g.,][Proposition 2.9]{wainwright2019high},
  \[\P\lb\bigg|\frac{1}{n+1}\sum_{i=1}^{n+1}G^k\left(\frac{i}{n+1}\right)(V_{i} - 1)\bigg| \ge t\rb\le 2\exp\lb-\min\left\{\frac{(n+1)t^2}{8\log^8 n}, \frac{(n+1)t}{4\log^4 n}\right\}\rb.\]
  Let $t = \frac{\log^6 n}{5\sqrt{n}}$, we obtain that
  \[\P\lb\bigg|\frac{1}{n+1}\sum_{i=1}^{n+1}G^k\left(\frac{i}{n+1}\right)(V_{i} - 1)\bigg| \ge \frac{\log^6 n}{5\sqrt{n}}\rb\le 2\exp \left\{-O(\log^2 n)\right\}.\]
  Applying the union bound, there exists a universal constant $c > 0$ such that
  \begin{equation}
    \label{eq:prob_event_c}
    \P\lb \event_n^c\rb \le \exp \left\{-c\log^2 n\right\}.
  \end{equation}

Adopting the notation in Section \ref{subsec:fisher_asym}, we have
  \[\frac{1}{n+1}\sum_{i=1}^{n+1}G^k\left(\frac{i}{n+1}\right) = \frac{1}{n+1}\sum_{i=1}^{t_n - 1} \lb -\log \left\{g_n\lb\frac{i}{n}\rb\right\}\rb^{k}.\]
  We have shown in the proof of Proposition \ref{prop:hasym_increasing} that $g_n(x)\le 1$ is increasing on $[0, t_n / n]$. Then,
  \begin{align}
    &\frac{1}{n}\sum_{i=1}^{t_n - 1} \log^2 \left\{g_n\lb\frac{i}{n}\rb\right\} \in \left[\int_{1/n}^{t_n / n}\log^2 \{g_n(x)\}dx, \frac{\log^2\left\{g_n(1 / n)\right\}}{n} + \int_{1/n}^{t_n / n}\log^2 \{g_n(x)\}dx\right].    \label{eq:integral_log2_bound}
  \end{align}
  For any $x\in (0, 1)$,
  \[|\log(g_n(x)) - \log(x)|\le \frac{|g_n(x) - x|}{x} = a_{n}\sqrt{\frac{1 - x}{x}}.\]
  Then
  \begin{align*}
    &\bigg|\int_{1/n}^{t_n / n}\log^2 \{g_n(x)\}dx - \int_{1/n}^{t_n / n}\log^2 xdx\bigg|\\
    & \le 2a_n \int_{1/n}^{t_n / n} \sqrt{\frac{1-x}{x}}(-\log x) dx + a_n^2 \int_{1/n}^{t_n / n} \frac{1-x}{x} dx\\
    & \le 2a_{n}\sqrt{\int_{1/n}^{t_n / n} \frac{1-x}{x} dx}\sqrt{\int_{1/n}^{t_n / n} \log^2 x dx} + a_n^2 \int_{1/n}^{t_n / n} \frac{1-x}{x} dx,
  \end{align*}
  where the last line uses the Cauchy-Schwarz inequality. Clearly,
  \[\int_{1/n}^{t_n / n} \frac{1-x}{x} dx = \log t_n - \frac{t_n - 1}{n}\le \log n,\]
  and
  \[\int_{1/n}^{t_n / n} \log^2 x dx = x\log^2 x - 2x\log x + 2x\bigg|_{1/n}^{t_n/n}.\]
  Using the same reasoning as in \eqref{eq:integral_1},
  \begin{equation}
    \label{eq:log2x_integral}
    \int_{1/n}^{t_n / n} \log^2 x dx = 2 + O(a_n^2) = 2 + o_{\delta}\lb\frac{1}{\sqrt{n}}\rb,
  \end{equation}
  where we use $o_{\delta}$ and $O_{\delta}$ herein to hide the dependence on $\delta$. Therefore,
  \begin{align*}
    \bigg|\int_{1/n}^{t_n / n}\log^2 \{g_n(x)\}dx - \int_{1/n}^{t_n / n}\log^2 xdx\bigg| = O(a_n \sqrt{\log n} + a_n^2\log n) = O_{\delta}\lb\frac{\log^6 n}{\sqrt{n}}\rb.
  \end{align*}
    Putting pieces together with \eqref{eq:integral_log2_bound}, we have
    \begin{align*}
   \int_{1/n}^{t_n / n}\log^2 \{g_n(x)\}dx = 2 + O_{\delta}\lb\frac{\log^6 n}{\sqrt{n}}\rb.
    \end{align*}
  On the other hand, by Theorem \ref{thm:fisher_asym},
  \[\bigg|\frac{1}{n+1}\sum_{i=1}^{n+1} G\lb\frac{i}{n+1}\rb - 1\bigg| = O\lb \frac{\sqrt{\log\log n}}{\sqrt{n}}\rb =  O\lb \frac{\log^6 n}{\sqrt{n}}\rb.\]
  As a result, there exists a constant $C_1(\delta)$ that only depends on $\delta$ such that
  \begin{equation}
    \label{eq:k=12_rate}
    \bigg|\frac{1}{n+1}\sum_{i=1}^{n+1}G^2\left(\frac{i}{n+1}\right) - 2\bigg| + \bigg|\frac{1}{n+1}\sum_{i=1}^{n+1}G\left(\frac{i}{n+1}\right) - 1\bigg|\le C_1(\delta)\frac{\log^6 n}{\sqrt{n}}.
  \end{equation}
    Putting \eqref{eq:E_log_D}, \eqref{eq:event_En}, and \eqref{eq:k=12_rate} together, on the event $\event_{n}$, for $k = 1,2$, 
    \[\E_{H_0}\left[\lb -\log (h^\asym \circ \hat{u}^\marg)\rb^{k}\mid \D\right]\in \left[\frac{k - C_1(\delta) \log^6 n / \sqrt{n}}{1 + \log^6 n / \sqrt{n}}, \frac{k + C_1(\delta) \log^6 n / \sqrt{n}}{1 - \log^6 n / \sqrt{n}}\right].\]
    As a result, for sufficiently large $n$, 
  \[\sum_{k=1}^{2}\bigg|\E_{H_0}\left[|-\log (h^\asym \circ \hat{u}^\marg)|^k\mid \D\right] - k\bigg| \le 2(C_1(\delta) + 2)\frac{\log^6 n}{\sqrt{n}}.\]
  This completes the proof of \eqref{eq:k=12}.

  To prove \eqref{eq:k=34}, note that
  \begin{align*}
    &\frac{1}{n+1}\sum_{i=1}^{n+1}G^k\left(\frac{i}{n+1}\right) = \frac{1}{n+1}\sum_{i=1}^{t_n - 1} \lb -\log \left\{g_n\lb\frac{i}{n}\rb\right\}\rb^{k}\\
& \le \int_{0}^{t_n/n}\lb-\log g_n(x)\rb^k dx\le \int_{0}^{1}\lb-\log x\rb^k dx = k!.
  \end{align*}
  By \eqref{eq:E_log_D} and \eqref{eq:event_En}, on the event $\event_{n}$, for $k = 3,4$,
  \[\E_{H_0}\left[\lb -\log (h^\asym \circ \hat{u}^\marg)\rb^{k}\mid \D\right]\le \frac{k! + \log^6 n / \sqrt{n}}{1 - \log^6 n / \sqrt{n}}.\]
  This proves \eqref{eq:k=34}.
\end{proof}

\begin{lemma}\label{lem:uncond_moments}
  There exist a universal constant $C > 0$ and a constant $C(\delta) > 0$ that only depends on $\delta$ such that,
  \begin{equation}
    \label{eq:k=12_uncond}
    \sum_{k=1}^{2}\bigg|\E_{H_0}\left[|-\log (h^\asym \circ \hat{u}^\marg)|^k\right] - k\bigg| \le \frac{C(\delta)\log^6 n}{\sqrt{n}},
  \end{equation}
  and
  \begin{equation}
    \label{eq:k=34_uncond}
    \sum_{k=3}^{4}\bigg|\E_{H_0}\left[|-\log (h^\asym \circ \hat{u}^\marg)|^k\right]\bigg|\le C.
  \end{equation}

\end{lemma}
\begin{proof}
  Let $\event_n$ be the event defined in Lemma \ref{lem:cond_moments}. By Lemma \ref{lem:cond_moments}, $\P(\event_n^c) \le \exp\{-c(\log n)^2\}$. On $\event_n^c$.
  \[\E_{H_0}\left[|-\log (h^\asym \circ \hat{u}^\marg)|^k\mid \D\right]\le (\log n)^{k}.\]
  Thus, for $k\le 4$, 
  \[\E_{H_0}\left[|-\log (h^\asym \circ \hat{u}^\marg)|^k I(\event^c)\right]\le \exp\{-c(\log n)^2\}(\log n)^4 = O\lb\frac{1}{\sqrt{n}}\rb.\]
  By the triangle inequality, for $k = 1, 2$,
  \begin{align*}
    &\bigg|\E_{H_0}\left[|-\log (h^\asym \circ \hat{u}^\marg)|^k - k\right]\bigg|\\
    & \le \bigg|\E_{H_0}\left[\lb |-\log (h^\asym \circ \hat{u}^\marg)|^k - k\rb I(\event_n)\right]\bigg| + k\P(\event_n^c) + \bigg|\E_{H_0}\left[|-\log (h^\asym \circ \hat{u}^\marg)|^k I(\event^c)\right]\bigg|\\
    & = \E\left\{\bigg|\E_{H_0}\left[|-\log (h^\asym \circ \hat{u}^\marg)|^k - k\mid \D\right]\bigg|I(\event_n)\right\} + O\lb\frac{1}{\sqrt{n}}\rb\\
    & \le \frac{C(\delta)\log^6 n}{\sqrt{n}} + O\lb\frac{1}{\sqrt{n}}\rb
  \end{align*}
  where the last line uses \eqref{eq:k=12}. Using a similar argument, we can also prove \eqref{eq:k=34_uncond}.
\end{proof}

\begin{lemma}\label{lem:dK_normal}
  For any $\mu\in \R$ and $\sigma^2 > 0$,
  \[d_{K}(N(\mu, \sigma^2), N(0, 1))\le |\mu| + |1/\sigma - 1|\cdot \max\{1, \sigma\}.\]
\end{lemma}
\begin{proof}
  Let $W\sim N(0, 1)$. By the triangle inequality,
  \begin{align*}
    &d_{K}(N(\mu, \sigma^2), N(0, 1))\\
    &\le d_{K}(N(\mu, \sigma^2), N(\mu, 1)) + d_{K}(N(\mu, 1), N(0, 1))\\
    &= d_{K}(N(0, \sigma^2), N(0, 1)) + d_{K}(N(\mu, 1), N(0, 1))\\
    & = \sup_{x}|\Phi(x / \sigma) - \Phi(x)| + \sup_{x}|\Phi(x - \mu) - \Phi(x)|
  \end{align*}
  For any $x\in \R$, 
  \begin{align*}
    |\Phi(x / \sigma) - \Phi(x)| & \le \bigg|\frac{1}{\sigma} - 1\bigg|\cdot |x| \cdot \frac{1}{\sqrt{2\pi}}\exp\left\{-\frac{x^2}{2\max\{1, \sigma\}^2}\right\}\\
                                 & \le \bigg|\frac{1}{\sigma} - 1\bigg|\cdot \max\{1, \sigma\} \cdot \sup_{y\in \R}\frac{1}{\sqrt{2\pi}}|y|\exp\left\{-\frac{y^2}{2}\right\}\\
    & \le \bigg|\frac{1}{\sigma} - 1\bigg|\cdot \max\{1, \sigma\}.
  \end{align*}
  Similarly,
  \[|\Phi(x - \mu) - \Phi(x)|\le |\mu| \cdot \sup_{y\in \R}\frac{1}{\sqrt{2\pi}}\exp\left\{-\frac{y^2}{2}\right\}\le |\mu|.\]
\end{proof}

\begin{lemma}\label{lem:dW_CLT}[\citet{ross2011fundamentals}, Theorem 3.2 with $D = 1$]
  Let $X_1, X_2, \ldots, X_{n}$ be independent random variables such that $\E[X_j] = 0$ and $\E[X_j^4] <\infty$, for all $j$. Write $B_{n} = \sum_{j}\Var[X_j]$. Then
  \[d_{W}\left(\mathcal{L}\left(\frac{1}{\sqrt{B_n}}\sum_{j=1}^{n}X_j\right), N(0, 1)\right)\le \frac{1}{B_n^{3/2}}\sum_{j=1}^{n}\E[|X_j|^3] + \frac{\sqrt{28}}{\sqrt{\pi}B_n}\sqrt{\sum_{j=1}^{n}\E[X_j^4]}.\]
\end{lemma}

\begin{proof}[\textbf{Proof of Lemma \ref{lem:fisher_asym_CLT}}]
Write $p_i$ for $\hat{u}^\marg_i$. Throughout the proof we suppress $H_0$ from the expectation $\E_{H_0}$ and $\P_{H_0}$ because we will only consider the global null. This proof refines that of Theorem \ref{thm:combination_test}. Let $p_i = \hat{u}_i^\marg$, $G(p) = -\log(h^\asym(p))$, and $\event_{n}$ be the event defined in \eqref{eq:event_En}. By Lemma \ref{lem:cond_moments}, on the event $\event_n$,
\[\sum_{k=1}^{2}\bigg|\E[G^k(p_i)\mid \D] - k\bigg| \le \frac{C(\delta)\log^6 n}{\sqrt{n}}, \quad \text{and}\quad \sum_{k=3}^{4}\bigg|\E[G^k(p_i)\mid \D]\bigg| \le C.\]
Recalling that we treat $\delta$ as a constant and ignore all terms that solely depend on $\delta$ in the big-O notation, we will simply write $C$ for $C(\delta)$ for the rest of the proof.

Analogous to the proof of Theorem \ref{thm:combination_test}, we define
  \[W_{m} = \frac{1}{\sqrt{m}}\sum_{i=1}^{m}\left\{G(p_i) - \E[G(p_i)\mid \mathcal{D}]\right\}, \quad \tilde{W}_{n} = \sqrt{n + 1}\left\{\E[G(p_i)\mid \mathcal{D}] - \E[G(p_i)]\right\}.\]
  By Lemma~\ref{lem:Berry_Esseen} with $g(x) = |x|$,
  \begin{align*}
    d_{K}\left(\mathcal{L}\left(\frac{W_{m}}{\sqrt{\Var[G(p_i)\mid \mathcal{D}]}}\mid \mathcal{D}\right), N(0, 1)\right)&\le \frac{A}{\sqrt{m}}\frac{\E\left[|G(p_i) - \E[G(p_i)\mid \mathcal{D}]|^3\right]}{\Var[G(p_i)\mid \mathcal{D}]^{3/2}}\\
    & \le \frac{8A}{\sqrt{m}}\frac{\E\left[|G(p_i)|^3\mid \mathcal{D}\right]}{\Var[G(p_i)\mid \mathcal{D}]^{3/2}},  
  \end{align*}
where $A$ is a universal constant. By definition of $\event_{n}$,
\begin{equation}
  \label{eq:var_cond}
  \bigg|\Var[G(p_i)\mid \mathcal{D}] - 1\bigg|\le \frac{2C\log^6 n}{\sqrt{n}}, \quad \text{on the event }\event_n.
\end{equation}
When $n$ is sufficiently large,
\begin{equation*}
  d_{K}\left(\mathcal{L}\left(\frac{W_{m}}{\Var[G(p_i)\mid \mathcal{D}]}\mid \mathcal{D}\right), N(0, 1)\right)\le \frac{1}{\sqrt{m}}\cdot\frac{8CA}{\lb 1 - \frac{2C\log^6 n}{\sqrt{n}}\rb^{3/2}}\le \frac{16CA}{\sqrt{m}}, \quad \text{on the event }\event_n.
\end{equation*}
The scale-invariance of the Kolmogorov distance then implies
\[d_{K}\left(\mathcal{L}\left(W_{m}\mid \mathcal{D}\right), N(0, \Var[G(p_i)\mid \mathcal{D}])\right)\le \frac{16CA}{\sqrt{m}}, \quad \text{on the event }\event_n.
\]
By \eqref{eq:var_cond}, Lemma \ref{lem:dK_normal} and the triangle inequality, for a sufficiently large $n$,
\begin{equation}
  \label{eq:dK_normal_Wm}
  d_{K}\left(\mathcal{L}\left(W_{m}\mid \mathcal{D}\right), N(0, 1)\right)\le \frac{16CA}{\sqrt{m}} + \frac{6C\log^6 n}{\sqrt{n}},\quad \text{on the event }\event_n.
\end{equation}
Since $\tilde{W}_{n}$ is a function of $\mathcal{D}$, for any $t\in \R$,
\[\bigg|\P\lb W_{m} + \sqrt{\frac{m}{n+1}}\tilde{W}_{n}\ge t \mid \mathcal{D}\rb - \bar{\Phi}\lb t - \sqrt{\frac{m}{n+1}}\tilde{W}_{n}\rb\bigg|\le \frac{16CA}{\sqrt{m}} + \frac{6C\log^6 n}{\sqrt{n}},\quad \text{on the event }\event_n.\]
By \eqref{eq:prob_event_c}, when $n$ is sufficiently large,
\begin{align}
  &\sup_{t\in \R}\bigg|\P\lb W_{m} + \sqrt{\frac{m}{n+1}}\tilde{W}_{n}\ge t \rb - \E\left[\bar{\Phi}\lb t - \sqrt{\frac{m}{n+1}}\tilde{W}_{n}\rb\right]\bigg|\nonumber\\
  & \le \frac{16CA}{\sqrt{m}} + \frac{6C\log^6 n}{\sqrt{n}} + \P(\event_n^c)
    = O\lb \frac{1}{\sqrt{m}} + \frac{\log^6 n}{\sqrt{n}}\rb.
    \label{eq:P_Wm_Wn}
\end{align}
Recall from \eqref{eq:E_log_D} with $k = 1$ that
\[\tilde{W}_{n} = \frac{(n + 1)^{-1/2}\sum_{j=1}^{n+1}\left\{G\left(\frac{j}{n+1}\right) - \E[G(p_i)]\right\}V_{j}}{\frac{1}{n+1}\sum_{j=1}^{n+1}V_{j}}\triangleq \frac{\tilde{W}_{1n}}{\tilde{W}_{2n}}.\]
Let
\[B_{n} = \frac{1}{n+1}\sum_{j=1}^{n+1}\left\{G\left(\frac{j}{n+1}\right) - \E[G(p_i)]\right\}^2.\]
Since $\E[G(p_i)] = \frac{1}{n+1}\sum_{j=1}^{n+1}G\left(\frac{j}{n+1}\right)$,
\[B_{n} = \E\left[|-\log (h^\asym(p_i))|^2\right] - \lb\E\left[|-\log (h^\asym(p_i))|\right]\rb^2.\]
By \eqref{eq:k=12} in Lemma \ref{lem:uncond_moments},
\begin{equation*}
    |B_{n} - 1| \le 2 + \frac{C\log^6 n}{\sqrt{n}} - \lb 1 - \frac{C\log^6 n}{\sqrt{n}}\rb^2\le 1 + \frac{3C\log^6 n}{\sqrt{n}}.
\end{equation*}
Similarly, Lemma \ref{lem:uncond_moments} implies that
\[\sum_{k=3}^{4}\frac{1}{n+1}\sum_{j=1}^{n}\bigg|G\left(\frac{j}{n+1}\right) - \E[G(p_i)]\bigg|^k \le \sum_{k=3}^{4}2^{k}\frac{1}{n+1}\sum_{j=1}^{n}\bigg|G\left(\frac{j}{n+1}\right)\bigg|^k = O(1),\]
where the first step applies Jensen's inequality that
\[\Bigg|\frac{G\left(\frac{j}{n+1}\right) - \E[G(p_i)]}{2}\Bigg|^k\le \frac{1}{2}\lb\bigg|G\left(\frac{j}{n+1}\right)\bigg|^{k} + |\E[G(p_i)]|^k\rb,\]
and
\[|\E[G(p_i)]|^k\le \bigg|\frac{1}{n+1}\sum_{j=1}^{n+1}G\lb\frac{j}{n+1}\rb\bigg|^{k}\le \frac{1}{n+1}\sum_{j=1}^{n+1}\bigg|G\lb\frac{j}{n+1}\rb\bigg|^{k}.\]
By Lemma \ref{lem:dW_CLT},
\begin{align*}
  & d_{W}\left(\mathcal{L}\left(\frac{\tilde{W}_{1n}}{\sqrt{B_n}}\right), N(0, 1)\right)\\
  & \le \frac{\E[V_1^3]}{(n+1)^{3/2}B_n^{3/2}}\sum_{j=1}^{n}\bigg|G\left(\frac{j}{n+1}\right) - \E[G(p_i)]\bigg|^3 + \frac{\sqrt{28}\E[V_1^4]}{\sqrt{\pi}(n+1) B_n}\sqrt{\sum_{j=1}^{n}\bigg|G\left(\frac{j}{n+1}\right) - \E[G(p_i)]\bigg|^4}\\
  & = O\lb\frac{1}{\sqrt{n}}\rb.
\end{align*}
As a result,
\[d_{W}\left(\mathcal{L}(\tilde{W}_{1n}), N(0, B_n)\right) = \sqrt{B_n}\cdot d_{W}\left(\mathcal{L}\left(\frac{\tilde{W}_{1n}}{\sqrt{B_n}}\right), N(0, 1)\right) = O\lb\frac{1}{\sqrt{n}}\rb.\]
Let $Z$ denotes a standard normal random variable. Using the coupling definition of the Wasserstein distance,
\[d_{W}(N(0, B_n), N(0, 1)) \le \E |\sqrt{B_n} Z - Z| = |\sqrt{B_n} - 1|\cdot \E|Z| = O(B_n - 1) = O\lb\frac{\log^6 n}{\sqrt{n}}\rb.\]
By the triangle inequality,
\begin{equation}
  \label{eq:W1n_CLT}
  d_{W}\left(\mathcal{L}(\tilde{W}_{1n}), N(0, 1)\right) \le d_{W}\left(\mathcal{L}(\tilde{W}_{1n}), N(0, B_n)\right) + d_{W}\left(N(0, B_n), N(0, 1)\right) = O\lb\frac{\log^6 n}{\sqrt{n}}\rb.
\end{equation}
On the other hand, using the coupling definition of Wasserstein distance again,
\begin{align}
  d_{W}\left(\tilde{W}_{n}, \tilde{W}_{1n}\right)&\le \E\left|\frac{\tilde{W}_{1n}}{\tilde{W}_{2n}} - \tilde{W}_{1n}\right| = \E\left[|\tilde{W}_{1n}|\cdot \left|\frac{1}{\tilde{W}_{2n}} - 1\right|\right].\label{eq:dW_Wn_W1n}
\end{align}
Since $V_i\sim \mathrm{Exp}(1)\sim \Gamma(1, 1)$,
\[(n+1)\tilde{W}_{2n}\sim \Gamma(n+1, 1).\]
Using the properties of inverse-Gamma distributions,
\[\E\left[\frac{1}{\tilde{W}_{2n}}\right] = 1 + \frac{1}{n}, \quad \Var\left[\frac{1}{\tilde{W}_{2n}}\right] = \frac{(n + 1)^2}{n^2 (n - 1)}.\]
By the triangle inequality and Cauchy-Schwarz inequality,
\begin{align*}
  \E\left[|\tilde{W}_{1n}|\cdot \left|\frac{1}{\tilde{W}_{2n}} - 1\right|\right]
  & \le \frac{1}{n}\E\left[|\tilde{W}_{1n}|\right] + \E\left[|\tilde{W}_{1n}| \cdot \left|\frac{1}{\tilde{W}_{2n}} - \E\left[\frac{1}{\tilde{W}_{2n}}\right]\right|\right]\\
  & \le \frac{1}{n}\sqrt{\E[\tilde{W}_{1n}^2]} + \sqrt{\E[\tilde{W}_{1n}^2]}\cdot \sqrt{\Var \left[\frac{1}{\tilde{W}_{2n}}\right]}\\
  & = O\lb \frac{\sqrt{\E[\tilde{W}_{1n}^2]}}{\sqrt{n}}\rb.
\end{align*}
Note that
\[\E[\tilde{W}_{1n}] = (n + 1)^{-1/2}\sum_{j=1}^{n+1}\left\{G\left(\frac{j}{n+1}\right) - \E[G(p_i)]\right\} = 0,\]
and by Lemma \ref{lem:uncond_moments}, 
\[\Var[\tilde{W}_{1n}] = \frac{1}{n+1}\sum_{j=1}^{n+1}\left\{G\left(\frac{j}{n+1}\right) - \E[G(p_i)]\right\}^2\le \frac{1}{n+1}\sum_{j=1}^{n+1}G^2\left(\frac{j}{n+1}\right) = O(1).\]
Thus,
\[\E\left[|\tilde{W}_{1n}|\cdot \left|\frac{1}{\tilde{W}_{2n}} - 1\right|\right] = O\lb \frac{\sqrt{\Var[\tilde{W}_{1n}]}}{\sqrt{n}}\rb = O\lb\frac{1}{\sqrt{n}}\rb.\]
By the triangle inequality, \eqref{eq:W1n_CLT}, and \eqref{eq:dW_Wn_W1n}
\begin{equation}
  \label{eq:Wn_CLT}
  d_{W}\lb\tilde{W}_{n}, N(0, 1)\rb\le d_{W}\left(\tilde{W}_{n}, \tilde{W}_{1n}\right) + d_{W}\left(\tilde{W}_{1n}, N(0, 1)\right) = O\lb\frac{\log^6 n}{\sqrt{n}}\rb.
\end{equation}
Note that
\[|\bar{\Phi}'(x)| = \frac{1}{\sqrt{2\pi}}\exp\left\{-\frac{x^2}{2}\right\} \le 1.\]
Again, Let $Z$ denotes a standard normal random variable. Using the Kantorovich-Rubinstein dual representation of the Wasserstein's distance,
\begin{equation}
  \label{eq:barPhi_diff}
  \sup_{t\in \R}\bigg|\bar{\Phi}\lb t - \sqrt{\frac{m}{n+1}}\tilde{W}_n\rb - \bar{\Phi}\lb t - \sqrt{\frac{m}{n+1}}Z\rb\bigg| \le \sqrt{\frac{m}{n+1}} d_{W}(\tilde{W}_{n}, Z) = O\lb\frac{\sqrt{m}\log^6 n}{n}\rb.
\end{equation}
Similar to the last step in the proof of Theorem \ref{thm:combination_test}, letting $\tilde{Z}$ denote a standard normal random variable that is independent of $Z$,
\begin{equation}
  \label{eq:barPhi}
  \E\left[\bar{\Phi}\lb t - \sqrt{\frac{m}{n+1}}Z\rb\right] = \P\lb \sqrt{\frac{m}{n+1}}Z + \tilde{Z}\ge t\rb = \P\lb N\lb 0, 1 + \frac{m}{n+1}\rb\ge t\rb = \bar{\Phi}\lb \frac{t}{\sqrt{1 + \frac{m}{n+1}}}\rb.
\end{equation}
Combining \eqref{eq:P_Wm_Wn}, \eqref{eq:barPhi_diff}, and \eqref{eq:barPhi},
\begin{equation}\label{eq:approx_final_1}
  \sup_{t\in \R}\bigg|\P\lb W_{m} + \sqrt{\frac{m}{n+1}}\tilde{W}_{n}\ge t \rb - \bar{\Phi}\lb \frac{t}{\sqrt{1 + \frac{m}{n+1}}}\rb\bigg| = O\lb \frac{1}{\sqrt{m}} + \frac{\log^6 n}{\sqrt{n}}\rb.
\end{equation}
Note that $W_{m} + \sqrt{\frac{m}{n+1}}\tilde{W}_{n} = (1/\sqrt{m})\sum_{i=1}^{m}\{G(p_i) - \E[G(p_i)]\}$. Let
\[t = \sqrt{1 + \frac{m}{n}}z_{1 - \alpha} + \sqrt{m}\left\{ 1 - \frac{\pi}{2}\frac{c_n(\delta)}{\sqrt{n}} - \E[G(p_i)]\right\}.\]
Then \eqref{eq:approx_final_1} implies that
\[\sup_{\alpha\in (0, 1)}\bigg|\P\lb \sum_{i=1}^{m}G(p_i)\ge \sqrt{m}\sqrt{1 + \frac{m}{n}}z_{1- \alpha} + m \left\{ 1 - \frac{\pi}{2}\frac{c_n(\delta)}{\sqrt{n}}\right\}\rb - \bar{\Phi}\lb z_{1-\alpha} + b_{m, n}\rb\bigg| = O\lb \frac{1}{\sqrt{m}} + \frac{\log^6 n}{\sqrt{n}}\rb,\]
where
\[b_{m, n} = \lb\sqrt{\frac{1 + m/n}{1 + m / (n+1)}} - 1\rb z_{1-\alpha} + \sqrt{\frac{m(n + 1)}{n + m + 1}}\left\{ 1 - \frac{\pi}{2}\frac{c_n(\delta)}{\sqrt{n}} - \E[G(p_i)]\right\}.\]
As shown in the previous steps, $\bar{\Phi}$ is $1$-Lipschitz. Thus,
\[\sup_{\alpha\in (0, 1)}\bigg|\P\lb \sum_{i=1}^{m}G(p_i)\ge \sqrt{m}\sqrt{1 + \frac{m}{n}}z_{1- \alpha} + m \left\{ 1 - \frac{\pi}{2}\frac{c_n(\delta)}{\sqrt{n}}\right\}\rb - \alpha\bigg| = O\lb \frac{1}{\sqrt{m}} + \frac{\log^6 n}{\sqrt{n}} + b_{m, n}\rb.\]
By Theorem \ref{thm:fisher_asym},
\[\bigg|1 - \frac{\pi}{2}\frac{c_n(\delta)}{\sqrt{n}} - \E[G(p_i)]\bigg| = O\lb \frac{(\log n)(\log \log n)}{n}\rb.\]
Therefore,
\[b_{m, n} = O\lb \frac{1}{n} + \sqrt{\frac{m}{m+n}}\frac{(\log n)(\log \log n)}{\sqrt{n}}\rb = O\lb\frac{\log^6 n}{\sqrt{n}}\rb.\]
The proof is then completed.
\end{proof}

\subsection{DKWM adjustment}
\subsubsection{Mean of Fisher's combination statistic}
  For notational convenience, let
  \[b_{n} = \sqrt{\frac{\log(2/\delta)}{2n}}.\]
  It should be kept in mind that $b_n$ depends on $\delta$.
\begin{thm}[Part (b) of Theorem~\ref{theorem:fisher-mean}]\label{thm:fisher_DKWM}
  \[\E_{H_0}[-\log (h^\dkwm \circ \hat{u}^\marg)] = 1 - b_n\log\lb\frac{e}{b_n}\rb + O\lb\frac{\log n}{n}\rb.\]
\end{thm}

\begin{proof}
  For notational convenience, let $t_n = \lfloor (n + 1)(1 - b_n)\rfloor$. Since the mapping $x\mapsto -\log(x + b_n)$ is monotone decreasing,
  \begin{align}
    \E_{H_0}[-\log (h^\dkwm \circ \hat{u}^\marg)]
    &= \frac{1}{n+1}\sum_{i=1}^{n+1}-\log\lb\min\left\{\frac{i}{n+1} + b_{n}, 1\right\}\rb\nonumber\\
    &= \frac{1}{n+1}\sum_{i=1}^{t_n}-\log\lb\frac{i}{n+1} + b_{n}\rb\nonumber\\
    & \in \left[\int_{b_n + 1/(n + 1)}^{b_n + t_n / (n + 1)}(-\log x)dx, \int_{b_n}^{b_n + t_n / (n + 1)}(-\log x)dx\right].\label{eq:range_EH0_dkwm}
  \end{align}
  Note that the indefinite integral of $(-\log x)$ is 
  \[\int (-\log x)dx = -x\log x + x \triangleq h_1(x).\]
  It is easy to see that
  \[h_1\lb b_n\rb = b_n\log\lb\frac{e}{b_n}\rb,\]
  \[h_1\lb b_n + \frac{1}{n+1}\rb = h_1(b_n) + O\lb\frac{h_1'(b_n)}{n+1}\rb = b_n\log\lb\frac{e}{b_n}\rb + O\lb\frac{\log n}{n}\rb,\]
  and
  \[\bigg|h_1\lb b_n + \frac{t_n}{n+1}\rb - 1\bigg| = O\lb 1 - b_n - \frac{t_n}{n+1}\rb = O\lb\frac{1}{n}\rb.\]
  By Newton-Leibniz formula,
  \[\int_{b_n + 1/(n+1)}^{b_n + t_n / (n + 1)}(-\log x)dx, \int_{b_n}^{b_n + t_n / (n + 1)}(-\log x)dx = 1 - b_n\log\lb\frac{e}{b_n}\rb + O\lb\frac{\log n}{n}\rb.\]
  The proof is then closed by \eqref{eq:range_EH0_dkwm}.
\end{proof}

\subsubsection{Conditional moments of the adjusted p-values}
\begin{lemma}\label{lem:cond_moments_dkwm}
  There exist universal constants $C, c > 0$ and a constant $C(\delta) > 0$ that only depends on $\delta$ such that, with probability $1 - \exp\{-c(\log n)^2\}$,
  \begin{equation*}
    \sum_{k=1}^{2}\bigg|\E_{H_0}\left[|-\log (h^\dkwm \circ \hat{u}^\marg)|^k\mid \D\right] - k\bigg| \le \frac{C(\delta)\log^6 n}{\sqrt{n}},
  \end{equation*}
  and
  \begin{equation*}
    \sum_{k=3}^{4}\bigg|\E_{H_0}\left[|-\log (h^\dkwm \circ \hat{u}^\marg)|^k\mid \D\right]\bigg|\le C.
  \end{equation*}
\end{lemma}
\begin{proof}
  Let $G(p) = -\log (h^\dkwm(p))$. Then
  \begin{equation}
    \label{eq:G_upper_dkwm}
    G\lb\frac{i}{n+1}\rb\le -\log\lb \frac{1}{n+1} + \sqrt{\frac{\log(2/ \delta)}{2n}}\rb\le \log n,
  \end{equation}
  where we use the fact that
  \[\sqrt{\frac{\log(2/ \delta)}{2n}}\ge \sqrt{\frac{\log 2}{2n}}\ge \frac{1}{n(n+1)}.\]
  Using the same argument as in the proof of Lemma \ref{lem:cond_moments},
  \[\P\lb \event_n^c\rb \le \exp \left\{-c\log^2 n\right\},\]
  where $\event_n$ denotes the event that
  \[\sum_{k=0}^{4}\bigg|\frac{1}{n+1}\sum_{i=1}^{n+1}G^k\left(\frac{i}{n+1}\right)(V_{i} - 1)\bigg|\le \frac{\log^6 n}{\sqrt{n}}.\] 
Adopting the notation in Theorem \ref{thm:fisher_DKWM}, for any $k \ge 2$,
  \begin{align*}
    \frac{1}{n+1}\sum_{i=1}^{n+1}G^k\left(\frac{i}{n+1}\right) = \frac{1}{n+1}\sum_{i=1}^{t_n} \lb -\log \left\{\frac{i}{n+1} + b_n\right\}\rb^{k}.
  \end{align*}
  Since the mapping $x\mapsto \log^2(x + b_n)$ is decreasing,
  \begin{align*}
    \frac{1}{n+1}\sum_{i=1}^{n+1}G^2\left(\frac{i}{n+1}\right)&\in \left[\int_{1/(n + 1)}^{t_n / (n + 1)}\log^2 (x + b_n)dx, \int_{0}^{t_n / (n + 1)}\log^2 (x + b_n)dx\right]\\
                                                              & \in \left[\int_{b_n + 1/(n + 1)}^{b_n + t_n / (n + 1)}\log^2 xdx, \int_{b_n}^{b_n + t_n / (n + 1)}\log^2 xdx\right]
  \end{align*}
  Note that the indefinite integral of $\log^2 x$ is 
  \[\int \log^2 xdx = x\log^2 x - 2x\log x + 2x \triangleq h_2(x).\]
  It is easy to see that
  \[h_2\lb b_n\rb, h_2\lb b_n + \frac{1}{n+1}\rb = O_{\delta}\lb\frac{\log^2 n}{\sqrt{n}}\rb,\]
  where we use $O_{\delta}$ herein to hide the dependence on $\delta$, and
  \[\bigg|h_2\lb b_n + \frac{t_n}{n+1}\rb - 2\bigg| = O\lb 1 - b_n - \frac{t_n}{n+1}\rb = O\lb\frac{1}{n}\rb.\]
  By Newton-Leibniz formula,
  \[\int_{b_n+1/(n+1)}^{b_n + t_n / (n + 1)}\log^2 xdx, \int_{b_n}^{b_n + t_n / (n + 1)}\log^2 xdx = 2 + O_{\delta}\lb\frac{\log^2 n}{\sqrt{n}}\rb.\]
  This implies
  \[\bigg|\frac{1}{n+1}\sum_{i=1}^{n+1}G^2\left(\frac{i}{n+1}\right) - 2 \bigg| = O_{\delta}\lb\frac{\log^2 n}{\sqrt{n}}\rb.\]
  We have shown in Theorem \ref{thm:fisher_DKWM} that
  \[\bigg|\frac{1}{n+1}\sum_{i=1}^{n+1}G\left(\frac{i}{n+1}\right) - 1 \bigg| = O_{\delta}\lb\frac{\log n}{\sqrt{n}}\rb.\]
  Analogous to the proof of Lemma \ref{lem:cond_moments}, on the event $\event_{n}$, 
  \[\sum_{k=1}^{2}\bigg|\E_{H_0}\left[|-\log (h^\asym \circ \hat{u}^\marg)|^k\mid \D\right] - k\bigg| = O_{\delta}\lb\frac{\log^2 n}{\sqrt{n}}\rb = O_{\delta}\lb\frac{\log^6 n}{\sqrt{n}}\rb.\]
  
  To prove the bound for higher-order moments, note that
  \begin{align*}
    \frac{1}{n+1}\sum_{i=1}^{n+1}G^k\left(\frac{i}{n+1}\right) &= \frac{1}{n+1}\sum_{i=1}^{t_n} \lb -\log \left\{\frac{i}{n+1} + b_n\right\}\rb^{k} \le \int_{0}^{1}(-\log x)^k dx = k!.
  \end{align*}
  By \eqref{eq:E_log_D} and the definition of $\event_{n}$, on the event $\event_{n}$, for $k = 3,4$,
  \[\E\left[\lb -\log (h^\asym \circ \hat{u}^\marg)\rb^{k}\mid \D\right]\le \frac{k! + \log^6 n / \sqrt{n}}{1 - \log^6 n / \sqrt{n}} = O(1).\]
\end{proof}

\subsubsection{Tail approximation for Fisher's combination statistic}
\begin{lemma}\label{lem:fisher_dkwm_tail}
  Under the global null, there exists a universal constant $C > 0$ and $n_0(\delta)$ that only depends on $\delta$, such that for any $n\ge n_0(\delta)$ and $t > 0$,
  \begin{align*}
    &\P_{H_0}\lb \sum_{i=1}^{m}\left\{-\log\lb h^{\dkwm} \circ \hat{u}^\marg_i\rb - \E\left[-\log\lb h^{\dkwm} \circ \hat{u}^\marg\rb\right]\right\}\ge t\rb\\
    & \le \exp\left\{-C(\log n)^2 - C\min\left\{\frac{n}{m}, 1\right\}\min\left\{\frac{t^2}{m}, \frac{t}{\log n}\right\}\right\}.
  \end{align*}
\end{lemma}
\begin{proof}
   Let $G(p) = -\log (h^\dkwm(p))$ and write $p_i$ for $\hat{u}^\marg_i$. Throughout the proof we suppress $H_0$ from the expectation $\E_{H_0}$ and $\P_{H_0}$ because we will only consider the global null. By \eqref{eq:E_log_D} , we can write 
  \[\E[G(p_i)\mid \D] = \frac{\sum_{i=1}^{n+1}G\left(\frac{i}{n+1}\right)V_{i}}{\sum_{i=1}^{n+1}V_{i}}.\]
  Then, 
  \begin{align}
    &\P\lb \sum_{i=1}^{m}\{G(p_i) - \E[G(p_i)]\}\ge t\rb\nonumber\\
    & \le \P\lb \sum_{i=1}^{m}\{G(p_i) - \E[G(p_i)\mid \D]\}\ge \frac{t}{3}\rb + \P\lb \E[G(p_i)\mid \D] - \E[G(p_i)]\ge \frac{2t}{3m}\rb\nonumber\\
    & \le \P\lb \sum_{i=1}^{m}\{G(p_i) - \E[G(p_i)\mid \D]\}\ge \frac{t}{3}\rb + \P\lb \E[G(p_i)\mid \D] - \E[G(p_i)]\ge \frac{2t}{3m}\rb\nonumber\\
    & \le \P\lb \sum_{i=1}^{m}\{G(p_i) - \E[G(p_i)\mid \D]\}\ge \frac{t}{3}\rb + \P\lb \sum_{i=1}^{n+1}G\left(\frac{i}{n+1}\right)(V_{i} - 1)\ge \frac{t(n+1)}{3m}\rb\nonumber\\
    & \qquad+ \P\lb \sum_{i=1}^{n+1}V_{i} \le \frac{n+1}{2}\rb.\label{eq:dkwm_tail_prob}
  \end{align}
  By \eqref{eq:G_upper_dkwm} and Bernstein's inequality \citep[e.g.,][equation (2.10)]{boucheron2013concentration},
  \begin{align*}
    &\P\lb \sum_{i=1}^{m}\{G(p_i) - \E[G(p_i)\mid \D]\}\ge \frac{t}{3}\mid \D\rb\\
    &\le \exp\left\{-\frac{t^2}{18(m\Var[G(p_i)\mid \D] + t\log n / 9)}\right\}\\
    & \le \exp \left\{-\frac{t^2}{18m\Var[G(p_i)\mid \D]}\right\} + \exp\left\{-\frac{t}{2\log n}\right\}.
  \end{align*}
  By Lemma \ref{lem:cond_moments}, with probability $1 - \exp\{-c(\log n)^2\}$, 
  \[\Var[G(p_i)\mid \D] \le 2,\]
  when $C(\delta)\log^6 n / \sqrt{n}\le 1$. Thus,
  \begin{align}
    \label{eq:term1_dkwm_tail_prob}
    \P\lb \sum_{i=1}^{m}\{G(p_i) - \E[G(p_i)\mid \D]\}\ge \frac{t}{3}\rb\le \exp\{-c(\log n)^2\} + \exp \left\{-\frac{t^2}{36m}\right\} + \exp\left\{-\frac{t}{2\log n}\right\}.
  \end{align}
  Moving to the second term of \eqref{eq:dkwm_tail_prob}, we can use the fact that $V_i$ is sub-exponential with parameters $(2,2)$ as shown in the proof of Lemma \ref{lem:cond_moments} and apply Bernstein's inequality for sums of exponential variables \citep[e.g.,][Proposition 2.9]{wainwright2019high},
  \begin{align*}
    &\P\lb \sum_{i=1}^{n+1}G\left(\frac{i}{n+1}\right)(V_{i} - 1)\ge \frac{t(n+1)}{3m}\rb\\
    & \le \exp \left\{-\frac{n+1}{m^2}\frac{t^2}{36 \frac{1}{n+1}\sum_{i=1}^{n+1}G^2\lb\frac{i}{n+1}\rb}\right\} + \exp\left\{-\frac{n+1}{m}\frac{t}{4\log n}\right\}.
  \end{align*}
  By Lemma \ref{lem:cond_moments_dkwm} and a similar argument for Lemma \ref{lem:uncond_moments}, 
  \[\frac{1}{n+1}\sum_{i=1}^{n+1}G^2\lb\frac{i}{n+1}\rb = \E[G(p_i)^2] \le 3,\]
  when $C(\delta)\log^6 n / \sqrt{n}\le 1$. Thus, 
  \begin{align}
    \label{eq:term2_dkwm_tail_prob}
    \P\lb \sum_{i=1}^{n+1}G\left(\frac{i}{n+1}\right)(V_{i} - 1)\ge \frac{t(n+1)}{3m}\rb \le \exp \left\{-\frac{n+1}{m^2}\frac{t^2}{108}\right\} + \exp\left\{-\frac{n+1}{m}\frac{t}{4\log n}\right\}.
  \end{align}
  As for the third term of \eqref{eq:dkwm_tail_prob}, we can apply Bernstein's inequality for sums of sub-exponential random variables again and obtain that
  \begin{align}
    \label{eq:term3_dkwm_tail_prob}
    &\P\lb \sum_{i=1}^{n+1}V_{i} \le \frac{n+1}{2}\rb \le \P\lb \sum_{i=1}^{n+1}(1 - V_{i})\ge \frac{n+1}{2}\rb\le 2\exp\left\{-\frac{n+1}{16}\right\}.
  \end{align}
  Piecing \eqref{eq:term1_dkwm_tail_prob} - \eqref{eq:term3_dkwm_tail_prob} together, the lemma is proved. 
\end{proof}

\begin{lemma}\label{lem:fisher_dkwm_CLT}
  Under the global null, as $m, n\rightarrow \infty$,
  \begin{align*}
    &\bigg|\P_{H_0}\lb \sum_{i=1}^{m}-\log\lb h^{\dkwm}\circ \hat{u}^\marg_i\rb\ge \sqrt{m}\sqrt{1 + \frac{m}{n}}z_{1- \alpha} + m \left\{ 1 - b_n\log\lb\frac{e}{b_n}\rb\right\}\rb - \alpha\bigg|\\
    & \quad = O\lb \frac{1}{\sqrt{m}} + \frac{\log^6 n}{\sqrt{n}}\rb.    
  \end{align*}
\end{lemma}
\begin{proof}
  Let $G(p) = -\log\lb h^{\dkwm}(p)\rb$. Note that the proof of Lemma \ref{lem:fisher_asym_CLT} only replies on two facts that (1) $-\log\lb h^{\dkwm}\circ \hat{u}^\marg_i\rb\le \log n$, and (2) Lemma \ref{lem:cond_moments} holds for the first four conditional moments. Here, both continue to hold for DKWM-adjusted p-values. Following the steps in Section \ref{subsec:lem_fisher_asym_CLT}, we can show that \eqref{eq:approx_final_1} continues to hold, i.e.,
  \begin{equation}\label{eq:approx_dkwm}
  \sup_{t\in \R}\bigg|\P_{H_0}\lb \frac{1}{\sqrt{m}}\sum_{i=1}^{m}\{G(p_i) - \E[G(p_i)]\} \ge t \rb - \bar{\Phi}\lb \frac{t}{\sqrt{1 + \frac{m}{n+1}}}\rb\bigg| = O\lb \frac{1}{\sqrt{m}} + \frac{\log^6 n}{\sqrt{n}}\rb.
\end{equation}
By Theorem \ref{thm:fisher_DKWM},
\[\E[G(p_i)] = 1 - b_n\log\lb\frac{e}{b_n}\rb + O\lb\frac{\log n}{n}\rb.\]
Using the same argument below \eqref{eq:approx_final_1} in the proof of Lemma \ref{lem:fisher_asym_CLT}, we can prove the result.
\end{proof}

\subsubsection{Effective $\alpha$-level}
\begin{thm}\label{thm:fisher_dkwm_effective_alpha}
  \begin{enumerate}[(a)]
  \item  Assume that $m = \gamma n$ for some constant $\gamma \in (0, \infty)$. The type-I error of Fisher's combination test applied to $h^\dkwm\circ \hat{u}_i^\marg$'s
    \[\P_{H_0}\lb\sum_{i=1}^{m} -2\log(h^\dkwm \circ \hat{u}_i^\marg)\ge \chi^2(2m; 1 - \alpha)\rb \le \exp\left\{-C(\gamma, \delta)\log^2 n\right\},\]
    for some constant $C(\gamma, \delta) > 0$ that only depends on $\gamma$ and $\delta$. 
  \item Assume that $m \rightarrow \infty$ and $m = o(n / \log^2 n)$. The type-I error of Fisher's combination test applied to $h^\dkwm\circ \hat{u}_i^\marg$'s
  \begin{align*}
    &\P_{H_0}\lb\sum_{i=1}^{m} -2\log(h^\dkwm \circ \hat{u}_i^\marg)\ge \chi^2(2m; 1 - \alpha)\rb = \alpha + o(1).
  \end{align*}
  \end{enumerate}
\end{thm}

\begin{proof}
Throughout the proof we suppress $H_0$ from the expectation $\E_{H_0}$ and $\P_{H_0}$ because we will only consider the global null.   
  \begin{enumerate}[(a)]
  \item Let
    \[t_{n, m} = \frac{\chi^2(2m; 1 - \alpha)}{2} - m\E[-\log(h^\dkwm \circ \hat{u}_i^\marg)].\]
    Then, by Lemma \ref{lem:fisher_dkwm_tail},
    \begin{align*}
      &\P_{H_0}\lb \sum_{i=1}^{m} -2\log(h^\dkwm \circ \hat{u}_i^\marg)\ge \chi^2(2m; 1 - \alpha)\rb\\
      & = \P_{H_0}\lb \sum_{i=1}^{m} \{-\log(h^\dkwm \circ \hat{u}_i^\marg) - \E[-\log(h^\dkwm \circ \hat{u}_i^\marg)]\}\ge t_{n, m}\rb\\
      & \le \exp\left\{-C(\log n)^2 - C\min\left\{\frac{n}{m}, 1\right\}\min\left\{\frac{t_{n,m}^2}{m}, \frac{t_{n,m}}{\log n}\right\}\right\}.
    \end{align*}
    By \eqref{eq:chi_normal_approx},
    \[\frac{\chi^2(2m; 1 - \alpha) - 2m}{2} = \sqrt{m}z_{1-\alpha} + O\lb 1\rb.\]
    By Theorem \ref{thm:fisher_DKWM},
    \begin{align*}
      t_{n, m} &= mb_{n}\log\lb \frac{e}{b_n}\rb + \sqrt{m}z_{1-\alpha} + O\lb 1 + \frac{m\log n}{n}\rb\\
               & = mb_{n}\log\lb \frac{e}{b_n}\rb + \sqrt{m}z_{1-\alpha} + O\lb\log n\rb.
    \end{align*}
    Since $m = \gamma n$, there exists a constant $C(\gamma, \delta) > 0$ that only depends on $\gamma$ and $\delta$ such that
    \[t_{n, m}\ge C(\gamma, \delta)\sqrt{n}\log n.\]
    The proof is then completed.
  \item Choose $\alpha_n\in (0, 1)$ such that
    \[\sqrt{1 + \frac{m}{n}}z_{1 - \alpha_{n}} - \sqrt{m}b_n\log\lb\frac{e}{b_n}\rb = \frac{\chi^2(2m; 1 - \alpha) - 2m}{2\sqrt{m}} \triangleq A_m.\]
    By Lemma \ref{lem:fisher_dkwm_CLT},
    \begin{align*}
      &\bigg|\P_{H_0}\lb\sum_{i=1}^{m} -2\log(h^\dkwm \circ \hat{u}_i^\marg)\ge \chi^2(2m; 1 - \alpha)\rb - \alpha_n\bigg|\\
      & = \bigg|\P_{H_0}\lb\sum_{i=1}^{m} -\log(h^\dkwm \circ \hat{u}_i^\marg)\ge \sqrt{m}\sqrt{1 + \frac{m}{n}}z_{1- \alpha_n} + m \left\{ 1 - b_n\log\lb\frac{e}{b_n}\rb\right\}\rb - \alpha_n\bigg|\\
      & = O\lb \frac{1}{\sqrt{m}} + \frac{\log^6 n}{\sqrt{n}}\rb = o(1).
    \end{align*}
    Note that $A_{m} = z_{1-\alpha} + o(1)$ and $m = o(n / \log^2 n)$,
    \[z_{1-\alpha_n} = z_{1- \alpha} + o(1) + O(\sqrt{m} b_n\log n) = z_{1 - \alpha} + o(1).\]
    This implies $\alpha_n = \alpha + o(1)$ and hence
    \[\bigg|\P_{H_0}\lb\sum_{i=1}^{m} -2\log(h^\dkwm \circ \hat{u}_i^\marg)\ge \chi^2(2m; 1 - \alpha)\rb - \alpha\bigg| = o(1).\]
  \end{enumerate}
\end{proof}

\subsection{Simes adjustment}
\subsubsection{Mean of Fisher's combination statistic}
\begin{thm}[Part (c) of Theorem~\ref{theorem:fisher-mean}]\label{thm:fisher_simes}
  Assume that $k = \lceil\zeta n\rceil$ for some $\zeta > 0$. Then
  \[\E_{H_0}[-\log (h^\simes \circ \hat{u}^\marg)] \le 1 - \zeta - (1 - \zeta)\log (1 - \zeta) + O\lb\frac{\log n}{n}\rb.\]
\end{thm}
\begin{proof}
  Since $(i - j + 1) / (n - j + 1)\le i / n$ for any $j = 1, \ldots, k$, 
  \[h^\simes\lb 1 - \frac{i}{n+1}\rb\ge 1 - \delta^{1/k}\frac{i}{n}.\]
  Moreover,
  \[h^\simes\lb 1 - \frac{i}{n+1}\rb = 1\quad \text{if }i < k.\]
  Then
  \begin{align*}
    \E_{H_0}[-\log (h^\simes \circ \hat{u}^\marg)]
    &\le \frac{1}{n+1}\sum_{i=k}^{n}-\log\lb 1 - \delta^{1/k}\frac{i}{n}\rb\le \frac{1}{n}\sum_{i=k}^{n}-\log\lb 1 - \delta^{1/k}\frac{i}{n}\rb\\
    & \le \frac{\log (1-\delta^{1/k})}{n} + \frac{1}{n}\sum_{i=k}^{n-1}-\log\lb 1 - \frac{i}{n}\rb,
  \end{align*}
  where the last line uses the fact that $\delta \le 1$. Since the mapping $x\mapsto -\log(1 - x)$ is increasing,
  \begin{align*}
    \frac{1}{n}\sum_{i=k}^{n-1}-\log\lb 1 - \frac{i}{n}\rb &\le \int_{k/n}^{1}-\log(1 - x)dx \le \int_{\zeta}^{1}-\log(1 - x)dx\\
                                                           & \le \int_{0}^{1 - \zeta}(-\log x)dx = 1 - \zeta - (1 - \zeta)\log (1 - \zeta).
  \end{align*}
  Thus,
  \[\E_{H_0}[-\log (h^\simes \circ \hat{u}^\marg)]\le 1 - \zeta - (1 - \zeta)\log (1 - \zeta) + O\lb \frac{\log(1 - \delta^{1/k})}{n}\rb.\]
  The proof is completed by noting that
  \begin{align}
    -\log\lb 1 - \delta^{1/k}\rb & = -\log\lb 1 - \exp\left\{-\frac{\log (1 / \delta)}{k}\right\}\rb \le -\log\lb \frac{\log (1 / \delta)}{k}\rb\le \log n - \log\log\lb\frac{1}{\delta}\rb\label{eq:simes_max}\\
    & = O(\log n)\nonumber.
  \end{align}
\end{proof}

\subsubsection{Conditional variance of adjusted p-values}
\begin{lemma}\label{lem:cond_var_simes}
  There exist a universal constant $c > 0$ and a constant $n_0(\delta)$ that only depend on $\delta$ such that, for any $n\ge n_0(\delta)$,
  \begin{equation*}
    \P\lb\E_{H_0}[\log^2 (h^\simes \circ \hat{u}^\marg)\mid \D]\le \log n\rb \le 1 - \exp\left\{ \frac{cn}{\log n}\right\}.
  \end{equation*}
\end{lemma}
\begin{proof}
  Let $G(p) = -\log (h^\simes(p))$. By \eqref{eq:E_log_D},
  \[\E_{H_0}\left[\log^2 (h^\asym \circ \hat{u}^\marg)\mid \D\right] = \frac{\frac{1}{n+1}\sum_{i=1}^{n+1}G^2\left(\frac{i}{n+1}\right)V_{i}}{\frac{1}{n+1}\sum_{i=1}^{n+1}V_{i}}.\]
  By \eqref{eq:simes_max},
  \begin{equation*}
    G\lb\frac{i}{n+1}\rb\le -\log(1 - \delta^{1/k})\le \log n -  \log\log\frac{1}{\delta}.
  \end{equation*}
  Analogous to the proof of Theorem \ref{thm:fisher_simes},
  \begin{equation}
    \label{eq:G4_simes}
    \frac{1}{n+1}\sum_{i=1}^{n+1}G^4\lb\frac{i}{n+1}\rb\le \frac{\log^4(1 - \delta^{1/k})}{n} + \int_{0}^{1}(\log^4 x)dx = 24 + O_{\delta}\lb\frac{\log^4 n}{n}\rb,
  \end{equation}
  where we use $O_{\delta}$ herein to hide the dependence on $\delta$. When $n\ge n_0(\delta)$ for some sufficiently large $n_0(\delta)$ that only depends on $\delta$,
  \begin{equation}
    \label{eq:simes_max_element}
    G\lb\frac{i}{n+1}\rb\le -\log(1 - \delta^{1/k})\le 2\log n, \quad \frac{1}{n+1}\sum_{i=1}^{n+1}G^4\lb\frac{i}{n+1}\rb\le 25.
  \end{equation}
  By Bernstein's inequality \citep[e.g.,][Proposition 2.9]{wainwright2019high},
  \[\P\lb \frac{1}{n+1}\sum_{i=1}^{n+1}G^2\left(\frac{i}{n+1}\right)(V_{i} - 1)\ge t\rb\le \exp\lb-\min\left\{\frac{(n+1)t^2}{50}, \frac{(n+1)t}{8(\log n)^2}\right\}\rb.\]
  Then, there exists a universal constant $c > 0$ such that
  \[\P\lb \frac{1}{n+1}\sum_{i=1}^{n+1}G^2\left(\frac{i}{n+1}\right)(V_{i} - 1)\ge \frac{\log n}{3}\rb\le \exp\left\{-\frac{cn}{\log n}\right\}.\]
  By \eqref{eq:G4_simes} and Cauchy-Schwarz inequality,
  \[\frac{1}{n+1}\sum_{i=1}^{n+1}G^2\lb\frac{i}{n+1}\rb = O(1).\]
  Then, for any sufficiently large $n$,
  \[\P\lb \frac{1}{n+1}\sum_{i=1}^{n+1}G^2\left(\frac{i}{n+1}\right)V_{i}\ge \frac{\log n}{2}\rb\le \exp\left\{-\frac{cn}{\log n}\right\}.\]
  By \eqref{eq:term3_dkwm_tail_prob}, 
  \[\P\lb \frac{1}{n+1}\sum_{i=1}^{n+1}V_i \le \frac{n+1}{2}\rb\le \exp\lb-\frac{n+1}{16}\rb.\]
  The result is then proved by combining the above two inequalities.
\end{proof}

\subsubsection{Effective $\alpha$-level}
\begin{lemma}\label{lem:fisher_simes_tail}
  Under the global null, there exists a universal constant $C > 0$ and $n_0(\delta)$ that only depends on $\delta$, such that for any $n\ge n_0(\delta)$ and $t > 0$,
  \begin{align*}
    &\P_{H_0}\lb \sum_{i=1}^{m}\left\{-\log\lb h^{\simes} \circ \hat{u}^\marg_i\rb - \E\left[-\log\lb h^{\simes} \circ \hat{u}^\marg\rb\right]\right\}\ge t\rb\\
    & \le \exp\left\{-C\frac{n}{\log n} - C\min\left\{\frac{n}{m}, 1\right\}\min\left\{\frac{t^2}{m\log n}, \frac{t}{\log n}\right\}\right\}.
  \end{align*}
\end{lemma}
\begin{proof}
  Let $G(p) = -\log (h^\simes(p))$ and write $p_i$ for $\hat{u}^\marg_i$. Throughout the proof we suppress $H_0$ from the expectation $\E_{H_0}$ and $\P_{H_0}$ because we will only consider the global null. Analogous to \eqref{eq:dkwm_tail_prob},
  \begin{align}
    &\P\lb \sum_{i=1}^{m}\{G(p_i) - \E[G(p_i)]\}\ge t\rb\le     \P\lb \sum_{i=1}^{m}\{G(p_i) - \E[G(p_i)\mid \D]\}\ge \frac{t}{3}\rb\nonumber\\
    & \quad + \P\lb \sum_{i=1}^{n+1}G\left(\frac{i}{n+1}\right)(V_{i} - 1)\ge \frac{t(n+1)}{3m}\rb + \P\lb \sum_{i=1}^{n+1}V_{i} \le \frac{n+1}{2}\rb\label{eq:simes_tail_prob}.
  \end{align}
  By Lemma \ref{lem:cond_var_simes}, with probability $1 - \exp\{-cn/\log n\}$, 
  \[\Var[G(p_i)\mid \D] \le \log n,\]
  when $n$ is sufficiently large. By \eqref{eq:simes_max_element},
  \[G\lb\frac{i}{n+1}\rb\le 2\log n,\]
  when $n$ is sufficiently large.  
  \begin{align*}
    \P\lb \sum_{i=1}^{m}\{G(p_i) - \E[G(p_i)\mid \D]\}\ge \frac{t}{3}\rb\le \exp\left\{-\frac{cn}{\log n}\right\} + \exp \left\{-\frac{t^2}{18m\log n}\right\} + \exp\left\{-\frac{t}{4\log n}\right\}.
  \end{align*}
  Similar to \eqref{eq:term2_dkwm_tail_prob} and \eqref{eq:term3_dkwm_tail_prob}, we have
  \begin{align*}
    &\P\lb \sum_{i=1}^{n+1}G\left(\frac{i}{n+1}\right)(V_{i} - 1)\ge \frac{t(n+1)}{3m}\rb + \P\lb \sum_{i=1}^{n+1}V_{i} \le \frac{n+1}{2}\rb\\
    & \le \exp\left\{-c\frac{n}{m}\min\left\{\frac{t^2}{m}, \frac{t}{\log n}\right\} - cn\right\},
  \end{align*}
  for some universal constant $c > 0$. The result is then proved by \eqref{eq:simes_tail_prob}.
\end{proof}

\begin{thm}\label{thm:fisher_simes_effective_alpha}
  Assume $m / \log n\rightarrow \infty$. The type-I error of Fisher's combination test applied to $h^\simes\circ \hat{u}_i^\marg$'s
  \begin{align*}
    &\P_{H_0}\lb\sum_{i=1}^{m} -2\log(h^\simes \circ \hat{u}_i^\marg)\ge \chi^2(2m; 1 - \alpha)\rb \le \exp\left\{-C(\zeta, \delta)\frac{\min\{m, n\}}{\log n}\right\},
  \end{align*}
  for some constant $C(\zeta, \delta) > 0$ that only depends on $\zeta$ and $\delta$.
\end{thm}
\begin{proof}
  Let
  \[t_{n, m} = \frac{\chi^2(2m; 1 - \alpha)}{2} - m\E[-\log(h^\simes \circ \hat{u}_i^\marg)].\]
  Then, by Lemma \ref{lem:fisher_simes_tail},
  \begin{align*}
    &\P_{H_0}\lb \sum_{i=1}^{m} -2\log(h^\simes \circ \hat{u}_i^\marg)\ge \chi^2(2m; 1 - \alpha)\rb\\
    & = \P_{H_0}\lb \sum_{i=1}^{m} \{-\log(h^\simes \circ \hat{u}_i^\marg) - \E[-\log(h^\simes \circ \hat{u}_i^\marg)]\}\ge t_{n, m}\rb\\
    & \le \exp\left\{-C\frac{n}{\log n} - C\min\left\{\frac{n}{m}, 1\right\}\min\left\{\frac{t_{n,m}^2}{m\log n}, \frac{t_{n,m}}{\log n}\right\}\right\}.
  \end{align*}
  By \eqref{eq:chi_normal_approx},
  \[\frac{\chi^2(2m; 1 - \alpha) - 2m}{2} = \sqrt{m}z_{1-\alpha} + O\lb 1\rb.\]
  By Theorem \ref{thm:fisher_simes},
  \begin{align*}
    t_{n, m} &= m(\zeta + (1 - \zeta)\log(1 - \zeta)) + o(m)
  \end{align*}
  Thus, there exists $C(\zeta, \delta)> 0$ such that
  \[\P_{H_0}\lb \sum_{i=1}^{m} -2\log(h^\simes \circ \hat{u}_i^\marg)\ge \chi^2(2m; 1 - \alpha)\rb \le \exp\left\{-C(\zeta, \delta)\lb\frac{n}{\log n} + \min\left\{\frac{n}{m}, 1\right\}\frac{m}{\log n}\rb\right\}.\]
  
\end{proof}

\section{Numerical comparisons of different adjustment functions} \label{app:comparison}

In addition to the adjustment functions derived from the generalized Simes inequality and the DKWM inequality, we consider here another class of simultaneous bounds based on the so-called \emph{boundary crossing probability}~\cite{dempster1959generalized, durbin1973distribution, kotel1983computing, siegmund1986boundary}---the probability that $F(z)$ ever crosses $h(\hat{F}_{n}(z))$ for a fixed function $h(\cdot)$. This probability is generally difficult to compute analytically, but the special case of a linear $h(\cdot)$ is an exception.
Assuming that $F$ is the CDF of $\Unif([0, 1])$, let $\hat{F}_{n}(z)$ is the empirical CDF of $S_{1}, \ldots, S_{n}\stackrel{\mathrm{i.i.d.}}{\sim}\Unif([0, 1])$. Then,~\cite{dempster1959generalized} proved that
\begin{equation*}
  \P\left[ \hat{F}_{n}(z)\le b + \frac{1 - b}{1 - a}z, \,\, \forall z\in (0, 1)\right] = 1 - \Delta_{\mathrm{Dempster}}(a, b; n),
 \end{equation*}
 for any $a, b\in (0, 1)$, where
 \begin{equation}\label{eq:Delta_Dempster}
  \Delta_{\mathrm{Dempster}}(a, b; n) 
  \defeq a\sum_{j=0}^{\lfloor n(1 - b)\rfloor}\frac{n!}{j!(n-j)!}\left( a + \frac{1 - a}{1 - b}\frac{j}{n}\right)^{j-1}\left( 1 - a - \frac{1 - a}{1 - b}\frac{j}{n}\right)^{n-j}.
\end{equation}
 If we replace $S_i$ with $1 - S_i$, then $\hat{F}_{n}(z)$ becomes $1 - \hat{F}_{n}(1 - z)$. Further, replacing $z$ by $1 - z$ leads to
\begin{equation}  \label{eq:dempster}
  \P\left[ z \le \frac{1 - a}{1 - b}\hat{F}_{n}(z) + a, \,\, \forall z\in (0, 1)\right] = 1 - \Delta_{\mathrm{Dempster}}(a, b; n).
\end{equation}
For any pair $(a, b)$ with $\Delta_{\mathrm{Dempster}}(a, b; n) = \delta$, we obtain a function $h(z) = a + (1 - a)z / (1 - b)$ satisfying \eqref{eq:uniform_CDF}, which yields the following sequence satisfying \eqref{eq:confidence_band}:
\[b_{i} = a + \frac{1 - a}{1 - b}\frac{i}{n}.\]

Given any $a$, it is easy to compute the corresponding $b$ such that $\Delta_{\mathrm{Dempster}}(a, b; n) = \delta$ via a binary search. 
Note that this leads to adjusted p-values that cannot be lower than $b_1 = a + (1 - a) / (1 - b)n$. 
To ensure a fair comparison with the method based on the generalized Simes inequality, we choose $a$ via another binary search such that the resulting $b_1$ matches that given by the Simes inequality for a particular value of $k$.
If there exists no value of $a$ yielding the same $b_1$ as the Simes method, we set $a$ as to minimize $b_1$. 
Figure~\ref{fig:compare_bseq} compares the adjustment functions yielded by the generalized Simes inequality, the DKWM inequality, and the Dempster exact linear-boundary crossing probability with $k\in \{n/4, n/2\}$ and $n\in \{300, 1000, 3000, 10000\}$ for small marginal p-values within $[0, 0.05]$. It is clear that the Simes adjustment function is the best option in most scenarios, except when $n = 10000$ and $\hat{u}^\marg(X) > 0.03$, in which case the DKWM bound is tighter. Nonetheless, for the purpose of multiple testing, we would rarely expect  p-values above $0.03$ to be significant.

\begin{figure}[!htb]
    \centering
    \includegraphics[width=0.7\textwidth]{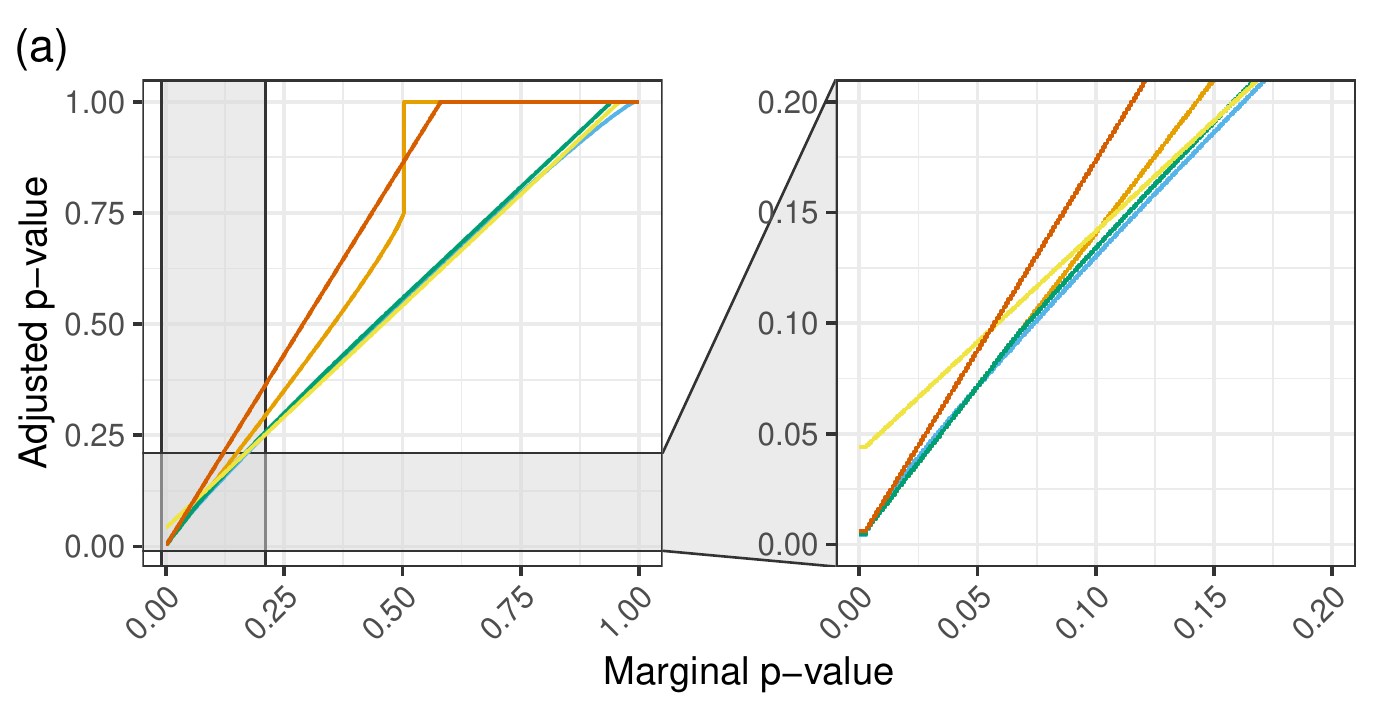}
    \includegraphics[width=0.7\textwidth]{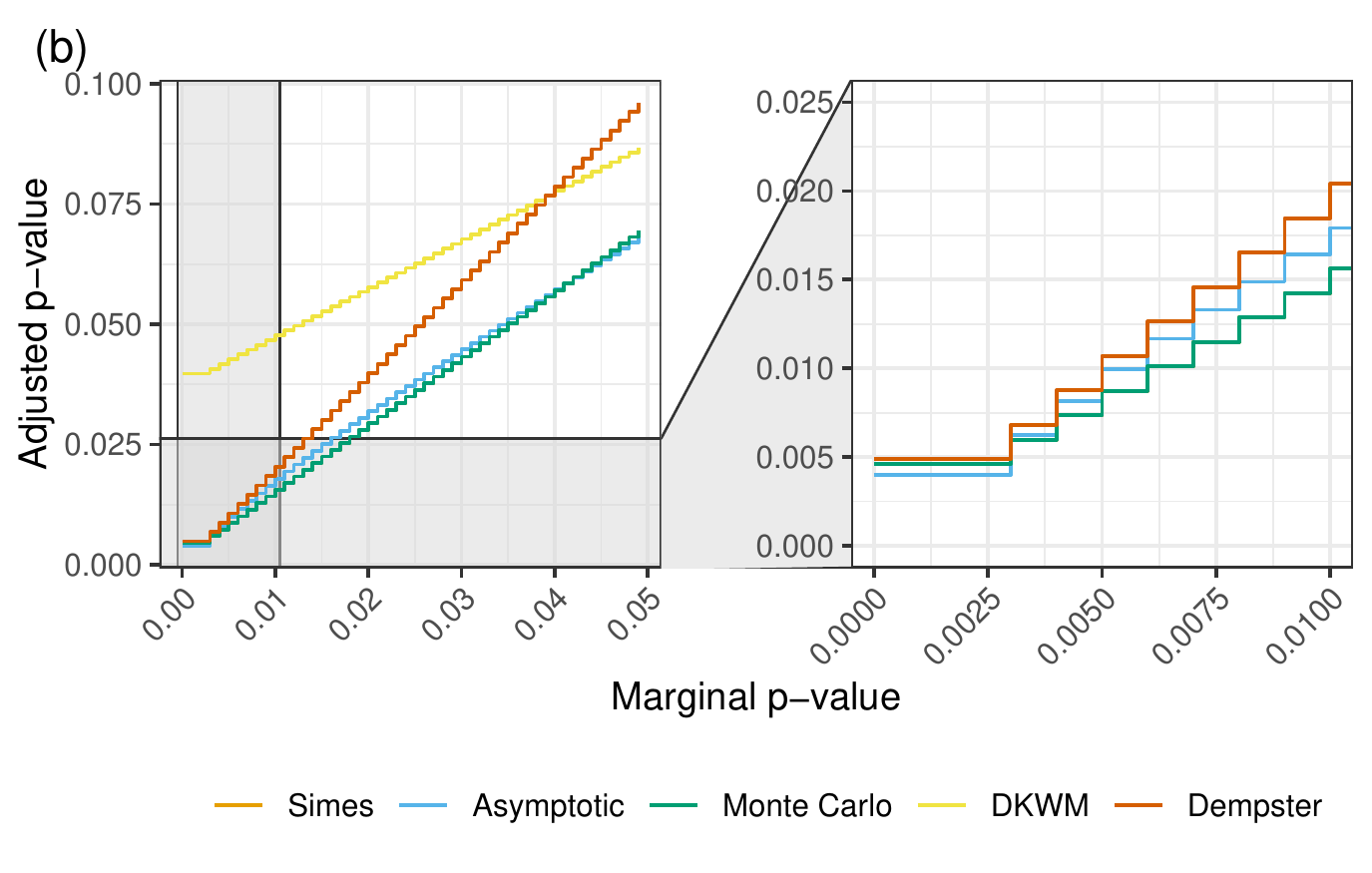}
    \caption{Comparison of different adjustment functions, with $n=1000$ and $\delta=0.1$.}
    \label{fig:compare_bseq}
\end{figure}

\begin{figure}[!htb]
    \centering
    \includegraphics[width=0.7\textwidth]{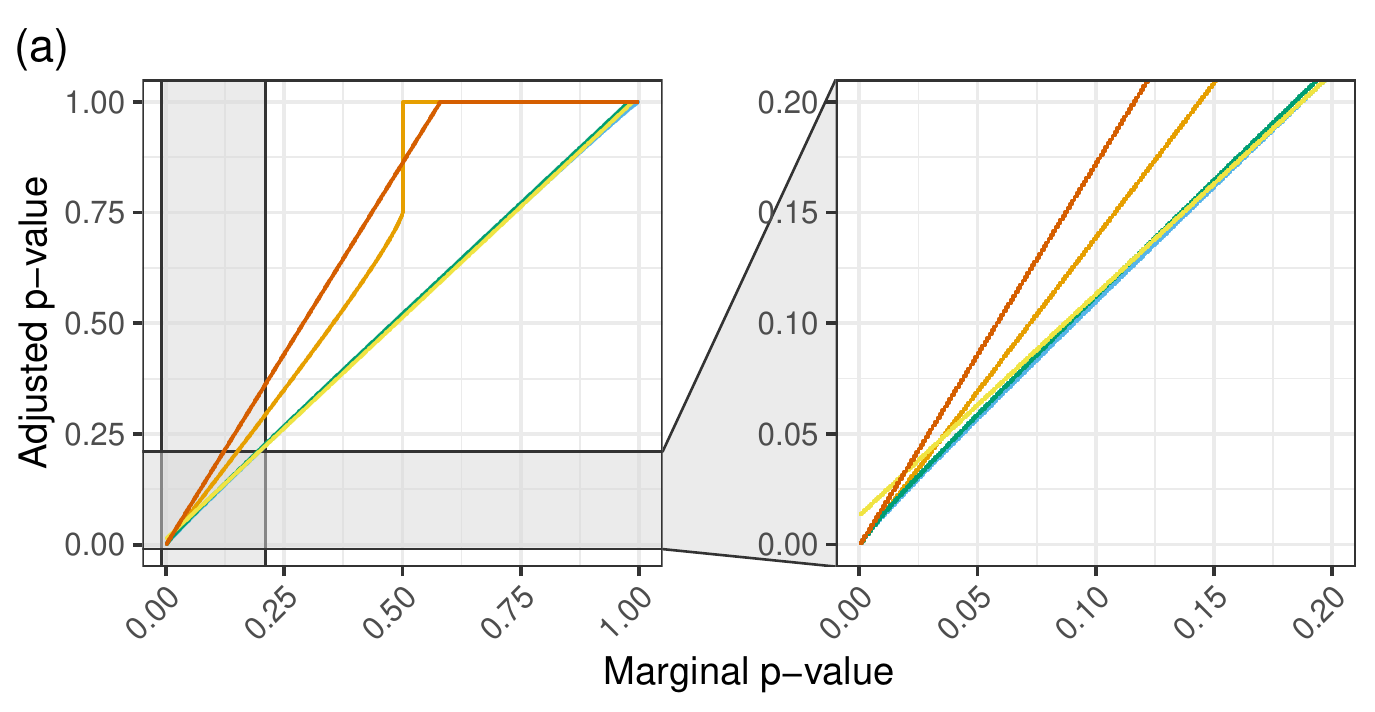}
    \includegraphics[width=0.7\textwidth]{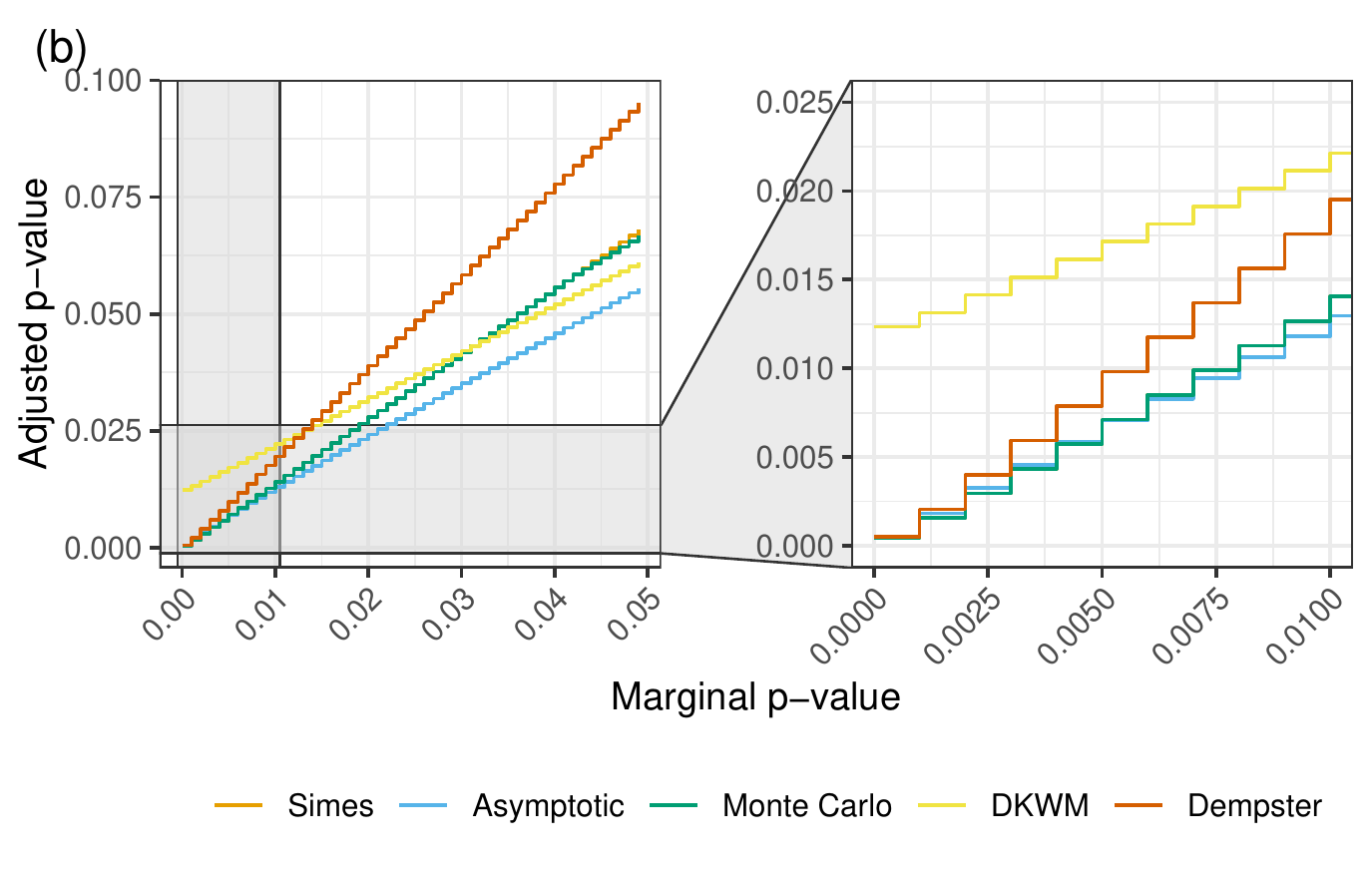}
    \caption{Comparison of different adjustment functions, with $n=10000$ and $\delta=0.1$. Other details are as in Figure~\ref{fig:compare_bseq}.}
    \label{fig:compare_bseq_1000}
\end{figure}

\begin{figure}[!htb]
    \centering
    \includegraphics[width=0.75\textwidth]{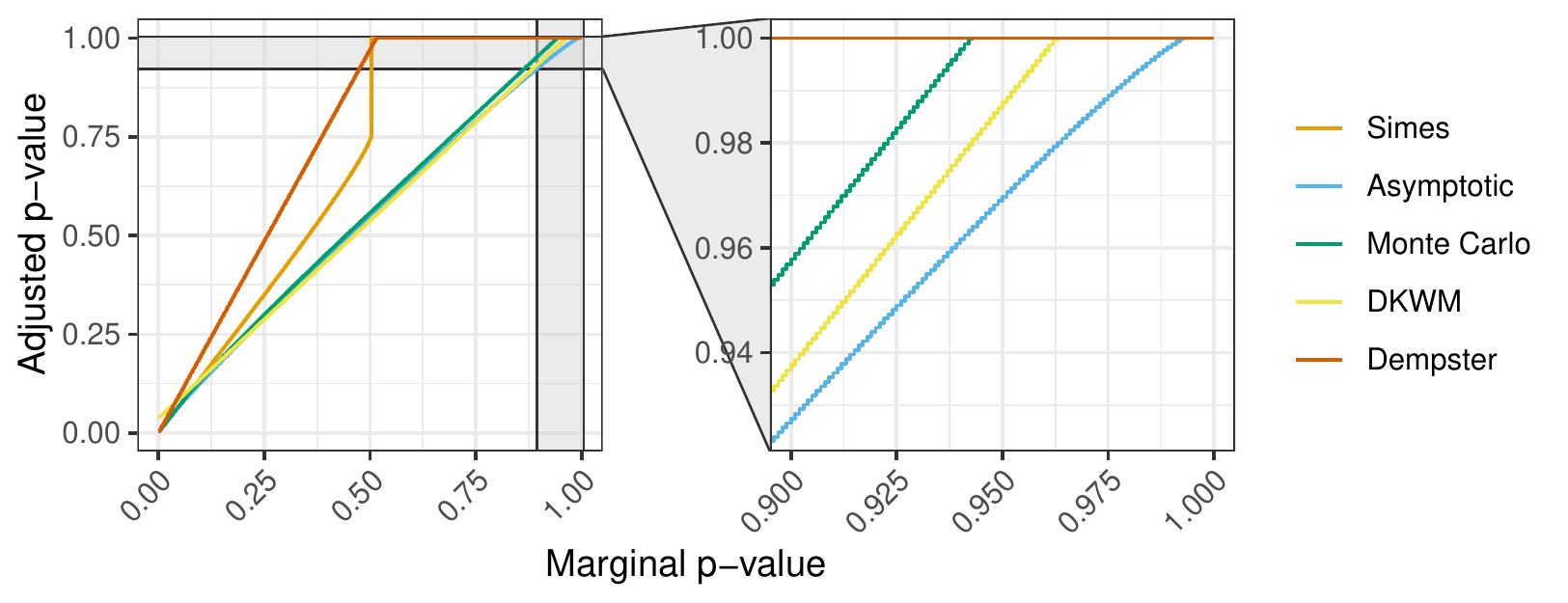}
    \caption{Comparison of different adjustment functions, with $n=1000$ and $\delta=0.1$.}
    \label{fig:compare_mcmc_opposite}
\end{figure}

\clearpage

\section{Numerical outlier detection experiments} \label{app:sim}

\subsection{Outlier detection on simulated data} \label{app:sim_data_exp}

\begin{figure}[!htb]
  \centering
  { \includegraphics[width=0.9\textwidth]{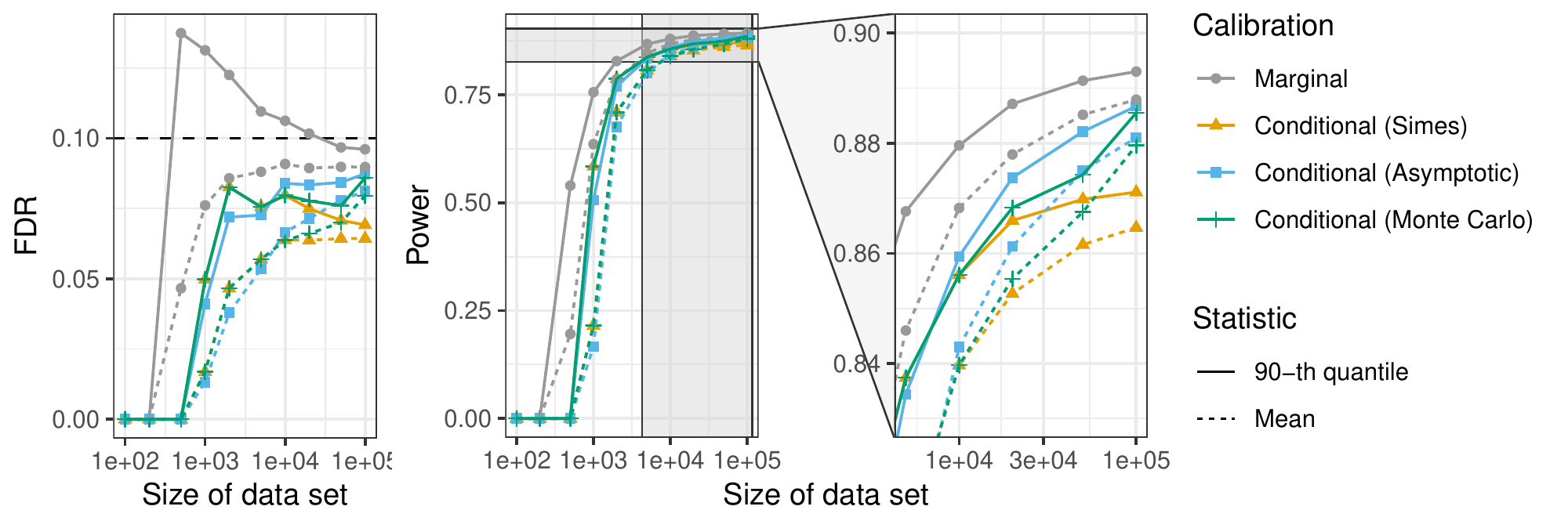}}
\caption{FDR and power in a simulated outlier detection problem as a function of the number of samples in the data set (half of which are utilized for calibration). Other details are as in Figure~\ref{fig:sim-bh}.
}
  \label{fig:sim-bh-n}
\end{figure}

\begin{figure}[!htb]
  \centering
  { \includegraphics[width=0.9\textwidth]{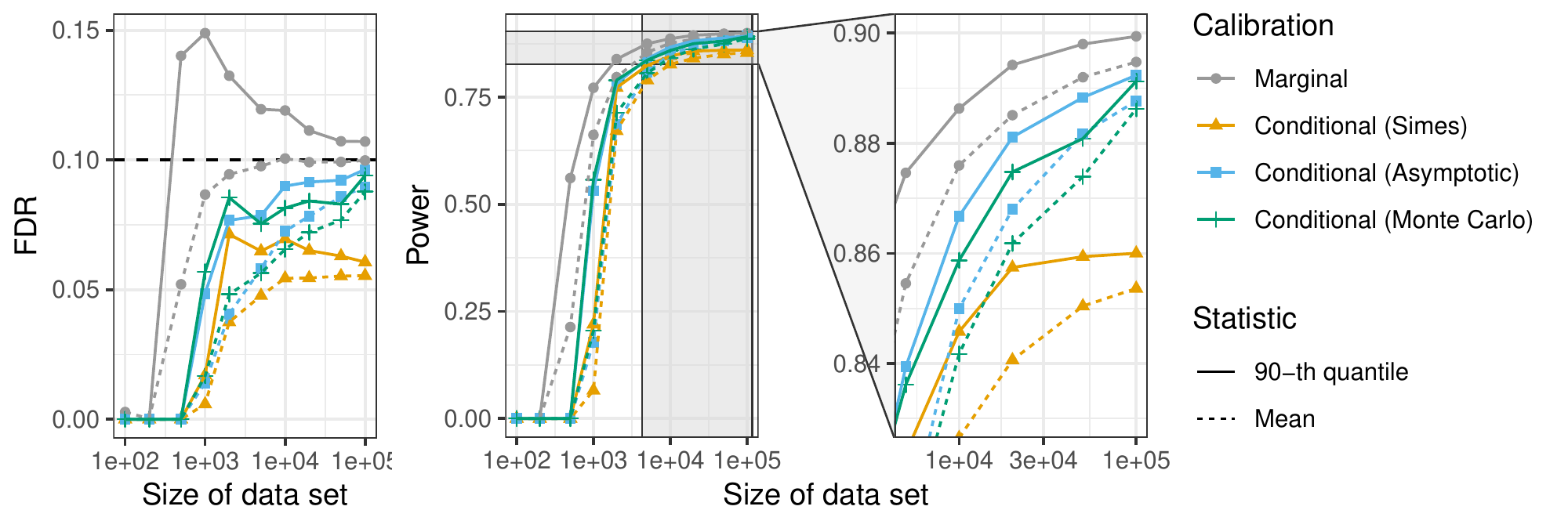}}
\caption{FDR and power in a simulated outlier detection problem, using the BH procedure with Storey's correction. Other details are as in Figure~\ref{fig:sim-bh-n}.
}
  \label{fig:sim-bh-n-storey}
\end{figure}

\begin{figure}[!htb]
  \centering
  { \includegraphics[width=0.9\textwidth]{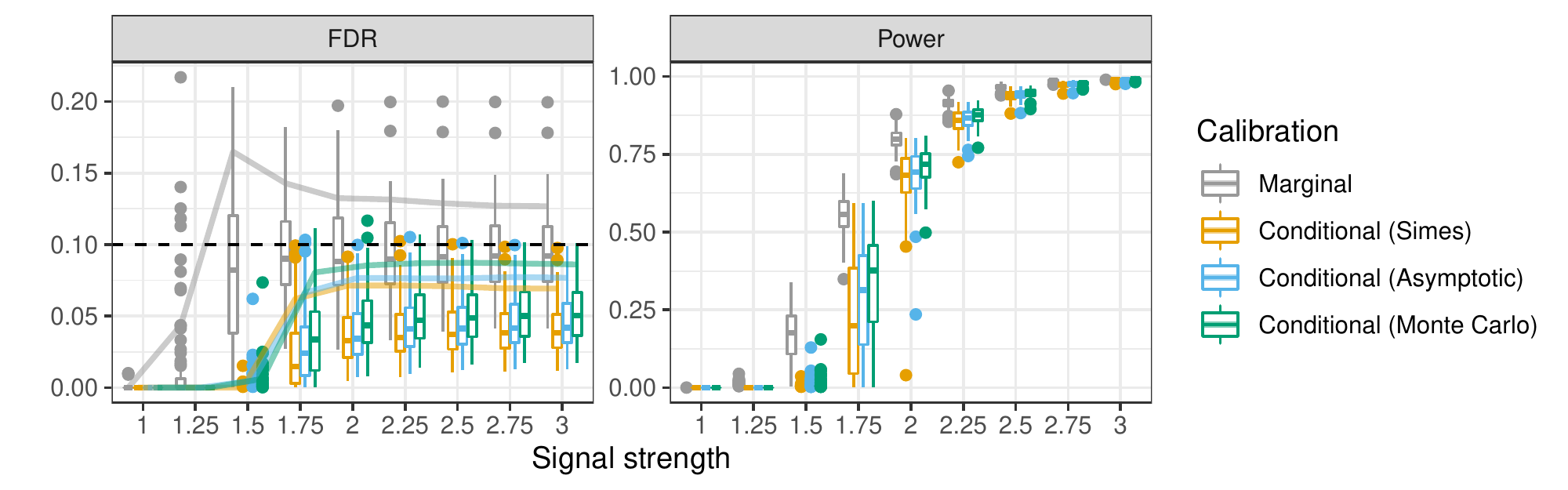}}
\caption{FDR and power in a simulated outlier detection problem, using the BH procedure with Storey's correction. Other details are as in Figure~\ref{fig:sim-bh}.
}
  \label{fig:sim-bh-storey}
\end{figure}

\begin{figure}[!htb]
  \centering
  { \includegraphics[width=0.9\textwidth]{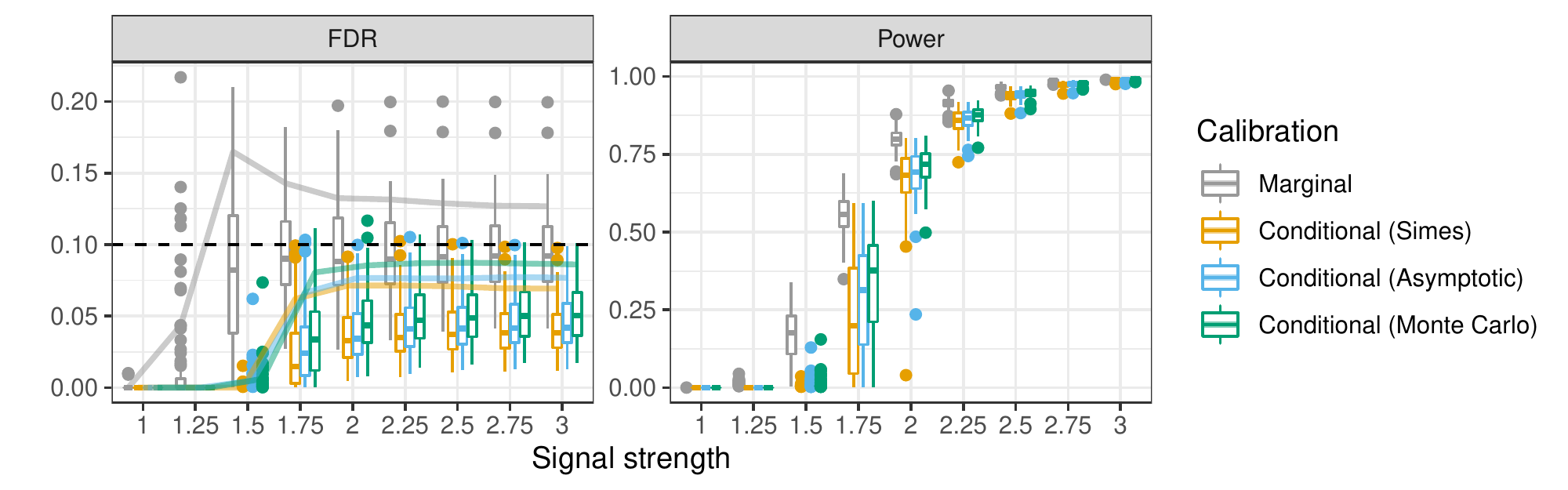}}
\caption{FDR and power in a simulated outlier detection problem, using the BH procedure with Storey's correction. The conditional calibration method is applied with $\delta = 0.25$ instead of $\delta=0.1$. Other details are as in Figure~\ref{fig:sim-bh}.
}
  \label{fig:sim-bh-delta0.25}
\end{figure}

\begin{figure}[!htb]
  \centering
  { \includegraphics[width=0.9\textwidth]{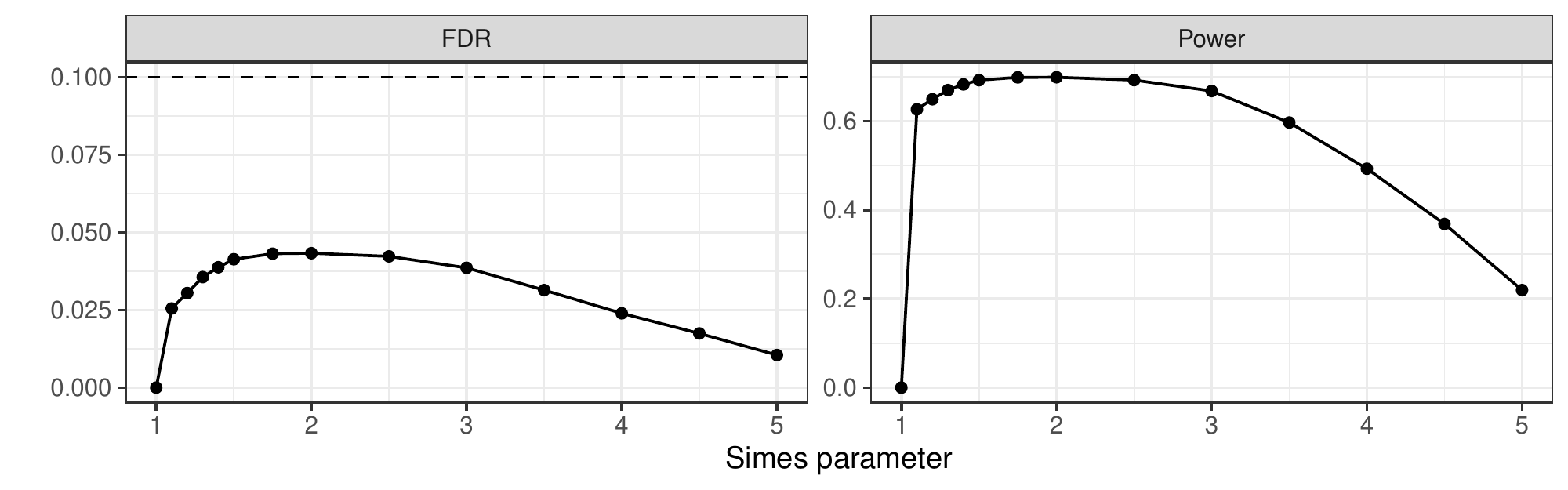}}
\caption{Performance of simultaneously calibrated conformal p-values as a function of the Simes parameter $n/k$. The signal strength is equal to 2. Other details are as in Figure~\ref{fig:sim-bh}.
}
  \label{fig:sim-simes-tune}
\end{figure}

\begin{figure}[!htb]
  \centering
  { \includegraphics[width=\textwidth]{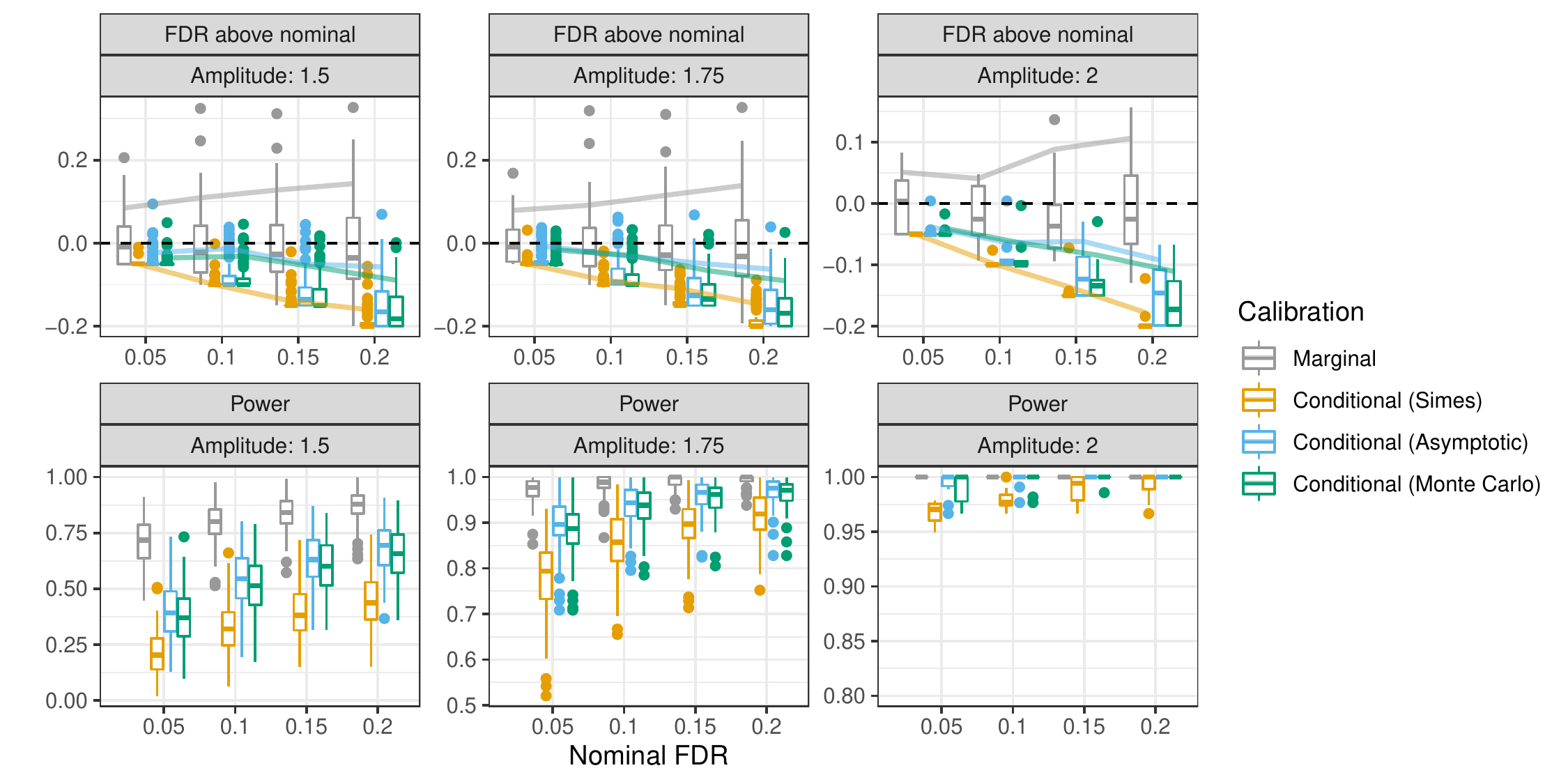}}
\caption{Performance of different methods for calibrating conformal p-values in a simulated outlier batch detection problem, for different values of the signal strength. Other details are as in Figure~\ref{fig:global-storey-large}.
}
  \label{fig:global-storey-large}
\end{figure}

\begin{figure}[!htb]
  \centering
  { \includegraphics[width=\textwidth]{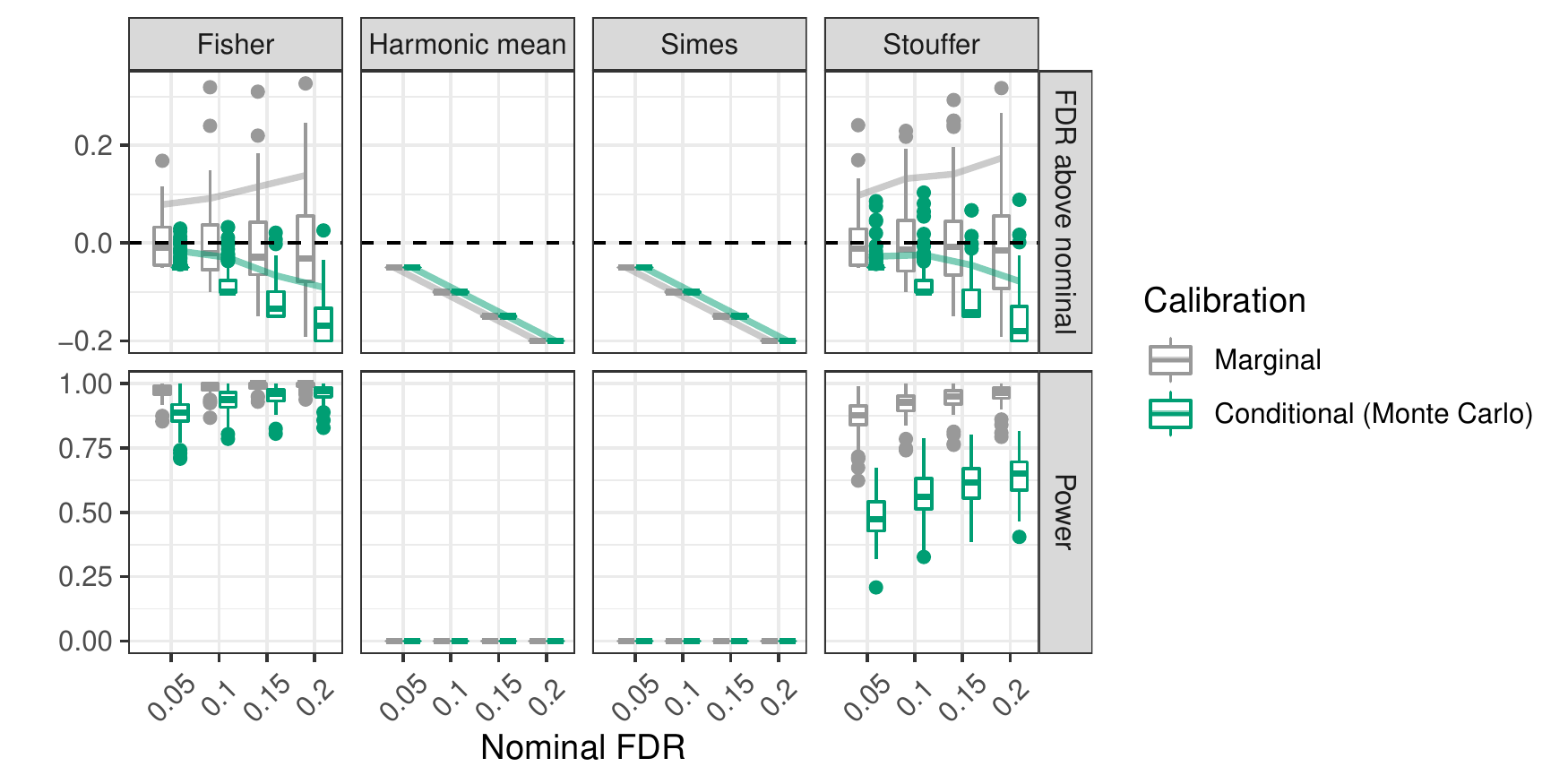}}
\caption{Performance of different methods for combining p-values from the same batch for the purpose of global testing. Other details are as in Figure~\ref{fig:global-storey}.
}
  \label{fig:global-combinations}
\end{figure}

\clearpage
\subsection{Outlier detection on real data}\label{app:real_data_exp}

{\small
\begingroup\fontsize{10}{12}\selectfont

\begin{longtable}[t]{lrllllllll}
\caption{Outlier detection performance on real data, using different data sets, machine learning models, and nominal FDR levels. The BH procedure is applied with Storey’s correction to control the FDR. Other details are as in Table~\ref{tab:data-bh-storey}. \label{tab:data-bh-long-storey}}\\
\toprule
\multicolumn{1}{c}{ } & \multicolumn{5}{c}{FDR} & \multicolumn{4}{c}{Power} \\
\cmidrule(l{3pt}r{3pt}){2-6} \cmidrule(l{3pt}r{3pt}){7-10}
\multicolumn{2}{c}{ } & \multicolumn{2}{c}{Mean} & \multicolumn{2}{c}{90th percentile} & \multicolumn{2}{c}{Mean} & \multicolumn{2}{c}{90-th quantile} \\
\cmidrule(l{3pt}r{3pt}){3-4} \cmidrule(l{3pt}r{3pt}){5-6} \cmidrule(l{3pt}r{3pt}){7-8} \cmidrule(l{3pt}r{3pt}){9-10}
Model & Nominal & Marg. & Cond. & Marg. & Cond. & Marg. & Cond. & Marg. & Cond.\\
\midrule
\endfirsthead
\caption[]{Outlier detection performance on real data, using different data sets, machine learning models, and nominal FDR levels. The BH procedure is applied with Storey’s correction to control the FDR. Other details are as in Table~\ref{tab:data-bh-storey}.  \textit{(continued)}}\\
\toprule
\multicolumn{1}{c}{ } & \multicolumn{5}{c}{FDR} & \multicolumn{4}{c}{Power} \\
\cmidrule(l{3pt}r{3pt}){2-6} \cmidrule(l{3pt}r{3pt}){7-10}
\multicolumn{2}{c}{ } & \multicolumn{2}{c}{Mean} & \multicolumn{2}{c}{90th percentile} & \multicolumn{2}{c}{Mean} & \multicolumn{2}{c}{90-th quantile} \\
\cmidrule(l{3pt}r{3pt}){3-4} \cmidrule(l{3pt}r{3pt}){5-6} \cmidrule(l{3pt}r{3pt}){7-8} \cmidrule(l{3pt}r{3pt}){9-10}
Model & Nominal & Marg. & Cond. & Marg. & Cond. & Marg. & Cond. & Marg. & Cond.\\
\midrule
\endhead
\midrule
\multicolumn{10}{r@{}}{\textit{(Continued on Next Page...)}}\
\endfoot
\bottomrule
\endlastfoot
\addlinespace[0.3em]
\multicolumn{10}{l}{\textbf{ALOI}}\\
\hspace{1em} & 0.05 & 0 & 0 & 0 & 0 & 0 & 0 & 0 & \vphantom{7} 0\\
\nopagebreak
\hspace{1em} & 0.10 & 0.001 & 0 & 0 & 0 & 0 & 0 & 0 & 0\\
\nopagebreak
\hspace{1em}\multirow[t]{-3}{*}{\raggedright\arraybackslash IForest} & 0.20 & 0.025 & 0.001 & 0.048 & 0 & 0 & 0 & 0 & 0\\
\cmidrule{1-10}\pagebreak[0]
\hspace{1em} & 0.05 & 0 & 0 & 0 & 0 & 0 & 0 & 0 & \vphantom{6} 0\\
\nopagebreak
\hspace{1em} & 0.10 & 0.005 & 0 & 0 & 0 & 0 & 0 & 0 & 0\\
\nopagebreak
\hspace{1em}\multirow[t]{-3}{*}{\raggedright\arraybackslash Neighbors} & 0.20 & 0.059 & 0.008 & \textcolor{red}{0.245} & 0.01 & 0.002 & 0 & 0.007 & 0\\
\cmidrule{1-10}\pagebreak[0]
\hspace{1em} & 0.05 & 0 & 0 & 0 & 0 & 0 & 0 & 0 & \vphantom{5} 0\\
\nopagebreak
\hspace{1em} & 0.10 & 0.006 & 0 & 0.006 & 0 & 0 & 0 & 0.001 & 0\\
\nopagebreak
\hspace{1em}\multirow[t]{-3}{*}{\raggedright\arraybackslash SVM} & 0.20 & 0.066 & 0.009 & \textcolor{orange}{0.212} & 0.017 & 0.003 & 0 & 0.01 & 0.002\\
\cmidrule{1-10}\pagebreak[0]
\addlinespace[0.3em]
\multicolumn{10}{l}{\textbf{Cover}}\\
\hspace{1em} & 0.05 & 0.004 & 0.001 & 0 & 0 & 0.001 & 0 & 0 & 0\\
\nopagebreak
\hspace{1em} & 0.10 & 0.017 & 0.007 & 0.037 & 0.002 & 0.003 & 0.001 & 0.005 & 0\\
\nopagebreak
\hspace{1em}\multirow[t]{-3}{*}{\raggedright\arraybackslash IForest} & 0.20 & 0.099 & 0.044 & \textcolor{red}{0.297} & 0.148 & 0.012 & 0.006 & 0.038 & 0.02\\
\cmidrule{1-10}\pagebreak[0]
\hspace{1em} & 0.05 & 0.047 & 0.038 & \textcolor{orange}{0.067} & \textcolor{orange}{0.054} & 0.948 & 0.935 & 0.964 & 0.957\\
\nopagebreak
\hspace{1em} & 0.10 & 0.096 & 0.081 & \textcolor{orange}{0.124} & \textcolor{orange}{0.106} & 0.972 & 0.968 & 0.98 & 0.976\\
\nopagebreak
\hspace{1em}\multirow[t]{-3}{*}{\raggedright\arraybackslash Neighbors} & 0.20 & 0.195 & 0.174 & \textcolor{orange}{0.231} & \textcolor{orange}{0.207} & 0.987 & 0.985 & 0.99 & 0.989\\
\cmidrule{1-10}\pagebreak[0]
\hspace{1em} & 0.05 & 0 & 0 & 0 & 0 & 0 & 0 & 0 & \vphantom{4} 0\\
\nopagebreak
\hspace{1em} & 0.10 & 0 & 0 & 0 & 0 & 0 & 0 & 0 & \vphantom{2} 0\\
\nopagebreak
\hspace{1em}\multirow[t]{-3}{*}{\raggedright\arraybackslash SVM} & 0.20 & 0 & 0 & 0 & 0 & 0 & 0 & 0 & \vphantom{2} 0\\
\cmidrule{1-10}\pagebreak[0]
\addlinespace[0.3em]
\multicolumn{10}{l}{\textbf{Credit card}}\\
\hspace{1em} & 0.05 & 0.032 & 0.016 & \textcolor{orange}{0.07} & \textcolor{orange}{0.051} & 0.149 & 0.08 & 0.358 & 0.264\\
\nopagebreak
\hspace{1em} & 0.10 & 0.09 & 0.063 & \textcolor{orange}{0.13} & \textcolor{orange}{0.105} & 0.383 & 0.277 & 0.57 & 0.518\\
\nopagebreak
\hspace{1em}\multirow[t]{-3}{*}{\raggedright\arraybackslash IForest} & 0.20 & 0.191 & 0.162 & \textcolor{orange}{0.228} & \textcolor{orange}{0.202} & 0.679 & 0.611 & 0.782 & 0.746\\
\cmidrule{1-10}\pagebreak[0]
\hspace{1em} & 0.05 & 0 & 0 & 0 & 0 & 0 & 0 & 0 & \vphantom{3} 0\\
\nopagebreak
\hspace{1em} & 0.10 & 0.005 & 0.001 & 0 & 0 & 0.001 & 0 & 0 & 0\\
\nopagebreak
\hspace{1em}\multirow[t]{-3}{*}{\raggedright\arraybackslash Neighbors} & 0.20 & 0.074 & 0.027 & \textcolor{red}{0.283} & 0.086 & 0.006 & 0.003 & 0.024 & 0.008\\
\cmidrule{1-10}\pagebreak[0]
\hspace{1em} & 0.05 & 0 & 0 & 0 & 0 & 0 & 0 & 0 & \vphantom{2} 0\\
\nopagebreak
\hspace{1em} & 0.10 & 0 & 0 & 0 & 0 & 0 & 0 & 0 & \vphantom{1} 0\\
\nopagebreak
\hspace{1em}\multirow[t]{-3}{*}{\raggedright\arraybackslash SVM} & 0.20 & 0 & 0 & 0 & 0 & 0 & 0 & 0 & \vphantom{1} 0\\
\cmidrule{1-10}\pagebreak[0]
\addlinespace[0.3em]
\multicolumn{10}{l}{\textbf{KDDCup99}}\\
\hspace{1em} & 0.05 & 0.039 & 0.02 & \textcolor{red}{0.076} & 0.047 & 0.359 & 0.216 & 0.51 & 0.465\\
\nopagebreak
\hspace{1em} & 0.10 & 0.096 & 0.059 & \textcolor{orange}{0.129} & 0.098 & 0.598 & 0.452 & 0.715 & 0.644\\
\nopagebreak
\hspace{1em}\multirow[t]{-3}{*}{\raggedright\arraybackslash IForest} & 0.20 & 0.194 & 0.131 & \textcolor{orange}{0.23} & 0.168 & 0.754 & 0.684 & 0.825 & 0.753\\
\cmidrule{1-10}\pagebreak[0]
\hspace{1em} & 0.05 & 0.006 & 0 & 0 & 0 & 0.001 & 0 & 0 & 0\\
\nopagebreak
\hspace{1em} & 0.10 & 0.03 & 0.009 & \textcolor{red}{0.147} & 0 & 0.009 & 0.002 & 0.049 & 0\\
\nopagebreak
\hspace{1em}\multirow[t]{-3}{*}{\raggedright\arraybackslash Neighbors} & 0.20 & 0.125 & 0.042 & \textcolor{red}{0.274} & 0.17 & 0.033 & 0.011 & 0.068 & 0.056\\
\cmidrule{1-10}\pagebreak[0]
\hspace{1em} & 0.05 & 0 & 0 & 0 & 0 & 0 & 0 & 0 & \vphantom{1} 0\\
\nopagebreak
\hspace{1em} & 0.10 & 0 & 0 & 0 & 0 & 0 & 0 & 0 & 0\\
\nopagebreak
\hspace{1em}\multirow[t]{-3}{*}{\raggedright\arraybackslash SVM} & 0.20 & 0 & 0 & 0 & 0 & 0 & 0 & 0 & 0\\
\cmidrule{1-10}\pagebreak[0]
\addlinespace[0.3em]
\multicolumn{10}{l}{\textbf{Mammography}}\\
\hspace{1em} & 0.05 & 0.011 & 0 & 0.037 & 0 & 0.01 & 0 & 0.023 & 0\\
\nopagebreak
\hspace{1em} & 0.10 & 0.076 & 0.002 & \textcolor{red}{0.171} & 0 & 0.059 & 0.004 & 0.146 & 0\\
\nopagebreak
\hspace{1em}\multirow[t]{-3}{*}{\raggedright\arraybackslash IForest} & 0.20 & 0.187 & 0.056 & \textcolor{red}{0.286} & 0.17 & 0.176 & 0.059 & 0.337 & 0.22\\
\cmidrule{1-10}\pagebreak[0]
\hspace{1em} & 0.05 & 0 & 0 & 0 & 0 & 0 & 0 & 0 & 0\\
\nopagebreak
\hspace{1em} & 0.10 & 0.016 & 0 & 0.027 & 0 & 0.011 & 0 & 0.021 & 0\\
\nopagebreak
\hspace{1em}\multirow[t]{-3}{*}{\raggedright\arraybackslash Neighbors} & 0.20 & 0.155 & 0.023 & \textcolor{red}{0.263} & 0.084 & 0.175 & 0.024 & 0.285 & 0.078\\
\cmidrule{1-10}\pagebreak[0]
\hspace{1em} & 0.05 & 0.007 & 0 & 0.018 & 0 & 0.003 & 0 & 0.008 & 0\\
\nopagebreak
\hspace{1em} & 0.10 & 0.066 & 0 & \textcolor{red}{0.168} & 0 & 0.04 & 0 & 0.09 & 0\\
\nopagebreak
\hspace{1em}\multirow[t]{-3}{*}{\raggedright\arraybackslash SVM} & 0.20 & 0.188 & 0.046 & \textcolor{red}{0.274} & 0.145 & 0.171 & 0.032 & 0.288 & 0.095\\
\cmidrule{1-10}\pagebreak[0]
\addlinespace[0.3em]
\multicolumn{10}{l}{\textbf{Digits}}\\
\hspace{1em} & 0.05 & 0.012 & 0 & 0.03 & 0 & 0.013 & 0 & 0.036 & 0\\
\nopagebreak
\hspace{1em} & 0.10 & 0.058 & 0.004 & \textcolor{red}{0.164} & 0 & 0.076 & 0.006 & 0.295 & 0\\
\nopagebreak
\hspace{1em}\multirow[t]{-3}{*}{\raggedright\arraybackslash IForest} & 0.20 & \textcolor{orange}{0.202} & 0.052 & \textcolor{red}{0.266} & 0.173 & 0.417 & 0.096 & 0.629 & 0.355\\
\cmidrule{1-10}\pagebreak[0]
\hspace{1em} & 0.05 & 0.006 & 0 & 0.028 & 0 & 0.033 & 0.002 & 0.049 & 0\\
\nopagebreak
\hspace{1em} & 0.10 & 0.059 & 0.006 & \textcolor{red}{0.135} & 0.01 & 0.273 & 0.035 & 0.752 & 0.03\\
\nopagebreak
\hspace{1em}\multirow[t]{-3}{*}{\raggedright\arraybackslash Neighbors} & 0.20 & 0.191 & 0.092 & \textcolor{orange}{0.238} & 0.166 & 0.841 & 0.455 & 0.99 & 0.879\\
\cmidrule{1-10}\pagebreak[0]
\hspace{1em} & 0.05 & 0.004 & 0 & 0.003 & 0 & 0.002 & 0 & 0.001 & 0\\
\nopagebreak
\hspace{1em} & 0.10 & 0.044 & 0 & \textcolor{red}{0.149} & 0 & 0.018 & 0 & 0.045 & 0\\
\nopagebreak
\hspace{1em}\multirow[t]{-3}{*}{\raggedright\arraybackslash SVM} & 0.20 & 0.175 & 0.018 & \textcolor{red}{0.264} & 0.066 & 0.227 & 0.017 & 0.468 & 0.048\\
\cmidrule{1-10}\pagebreak[0]
\addlinespace[0.3em]
\multicolumn{10}{l}{\textbf{Shuttle}}\\
\hspace{1em} & 0.05 & 0.047 & 0.031 & \textcolor{orange}{0.068} & 0.049 & 0.936 & 0.889 & 0.975 & 0.972\\
\nopagebreak
\hspace{1em} & 0.10 & 0.096 & 0.072 & \textcolor{orange}{0.122} & 0.097 & 0.974 & 0.965 & 0.981 & 0.979\\
\nopagebreak
\hspace{1em}\multirow[t]{-3}{*}{\raggedright\arraybackslash IForest} & 0.20 & 0.196 & 0.163 & \textcolor{orange}{0.228} & 0.198 & 0.981 & 0.98 & 0.984 & 0.983\\
\cmidrule{1-10}\pagebreak[0]
\hspace{1em} & 0.05 & 0.048 & 0.031 & \textcolor{orange}{0.066} & 0.047 & 0.99 & 0.951 & 0.998 & 0.991\\
\nopagebreak
\hspace{1em} & 0.10 & 0.098 & 0.07 & \textcolor{orange}{0.125} & 0.093 & 0.999 & 0.996 & 1 & 1\\
\nopagebreak
\hspace{1em}\multirow[t]{-3}{*}{\raggedright\arraybackslash Neighbors} & 0.20 & 0.199 & 0.151 & \textcolor{orange}{0.231} & 0.193 & 1 & 1 & 1 & 1\\
\cmidrule{1-10}\pagebreak[0]
\hspace{1em} & 0.05 & 0.043 & 0.021 & \textcolor{orange}{0.066} & 0.045 & 0.814 & 0.486 & 0.998 & 0.993\\
\nopagebreak
\hspace{1em} & 0.10 & 0.096 & 0.069 & \textcolor{orange}{0.127} & 0.093 & 0.999 & 0.997 & 1 & 0.999\\
\nopagebreak
\hspace{1em}\multirow[t]{-3}{*}{\raggedright\arraybackslash SVM} & 0.20 & 0.198 & 0.148 & \textcolor{orange}{0.233} & 0.182 & 1 & 1 & 1 & 1\\*
\end{longtable}
\endgroup{}
}

\clearpage
{\small
\begingroup\fontsize{10}{12}\selectfont

\begin{longtable}[t]{lrllllllll}
\caption{Outlier detection performance on real data, using different data sets, machine learning models, and nominal FDR levels. The BH procedure is applied without Storey’s correction to control the FDR. Other details are as in Table~\ref{tab:data-bh-long-storey}. \label{tab:data-bh-long}}\\
\toprule
\multicolumn{1}{c}{ } & \multicolumn{5}{c}{FDR} & \multicolumn{4}{c}{Power} \\
\cmidrule(l{3pt}r{3pt}){2-6} \cmidrule(l{3pt}r{3pt}){7-10}
\multicolumn{2}{c}{ } & \multicolumn{2}{c}{Mean} & \multicolumn{2}{c}{90th percentile} & \multicolumn{2}{c}{Mean} & \multicolumn{2}{c}{90-th quantile} \\
\cmidrule(l{3pt}r{3pt}){3-4} \cmidrule(l{3pt}r{3pt}){5-6} \cmidrule(l{3pt}r{3pt}){7-8} \cmidrule(l{3pt}r{3pt}){9-10}
Model & Nominal & Marg. & Cond. & Marg. & Cond. & Marg. & Cond. & Marg. & Cond.\\
\midrule
\endfirsthead
\caption[]{Outlier detection performance on real data, using different data sets, machine learning models, and nominal FDR levels. The BH procedure is applied without Storey’s correction to control the FDR. Other details are as in Table~\ref{tab:data-bh-long-storey}.  \textit{(continued)}}\\
\toprule
\multicolumn{1}{c}{ } & \multicolumn{5}{c}{FDR} & \multicolumn{4}{c}{Power} \\
\cmidrule(l{3pt}r{3pt}){2-6} \cmidrule(l{3pt}r{3pt}){7-10}
\multicolumn{2}{c}{ } & \multicolumn{2}{c}{Mean} & \multicolumn{2}{c}{90th percentile} & \multicolumn{2}{c}{Mean} & \multicolumn{2}{c}{90-th quantile} \\
\cmidrule(l{3pt}r{3pt}){3-4} \cmidrule(l{3pt}r{3pt}){5-6} \cmidrule(l{3pt}r{3pt}){7-8} \cmidrule(l{3pt}r{3pt}){9-10}
Model & Nominal & Marg. & Cond. & Marg. & Cond. & Marg. & Cond. & Marg. & Cond.\\
\midrule
\endhead
\midrule
\multicolumn{10}{r@{}}{\textit{(Continued on Next Page...)}}\
\endfoot
\bottomrule
\endlastfoot
\addlinespace[0.3em]
\multicolumn{10}{l}{\textbf{ALOI}}\\
\hspace{1em} & 0.05 & 0 & 0 & 0 & 0 & 0 & 0 & 0 & \vphantom{7} 0\\
\nopagebreak
\hspace{1em} & 0.10 & 0.001 & 0 & 0 & 0 & 0 & 0 & 0 & 0\\
\nopagebreak
\hspace{1em}\multirow[t]{-3}{*}{\raggedright\arraybackslash IForest} & 0.20 & 0.032 & 0.002 & 0.08 & 0 & 0 & 0 & 0 & 0\\
\cmidrule{1-10}\pagebreak[0]
\hspace{1em} & 0.05 & 0 & 0 & 0 & 0 & 0 & 0 & 0 & \vphantom{6} 0\\
\nopagebreak
\hspace{1em} & 0.10 & 0.004 & 0 & 0 & 0 & 0 & 0 & 0 & 0\\
\nopagebreak
\hspace{1em}\multirow[t]{-3}{*}{\raggedright\arraybackslash Neighbors} & 0.20 & 0.057 & 0.008 & \textcolor{orange}{0.234} & 0.01 & 0.002 & 0 & 0.006 & 0\\
\cmidrule{1-10}\pagebreak[0]
\hspace{1em} & 0.05 & 0 & 0 & 0 & 0 & 0 & 0 & 0 & \vphantom{5} 0\\
\nopagebreak
\hspace{1em} & 0.10 & 0.007 & 0 & 0.009 & 0 & 0 & 0 & 0.001 & 0\\
\nopagebreak
\hspace{1em}\multirow[t]{-3}{*}{\raggedright\arraybackslash SVM} & 0.20 & 0.076 & 0.012 & \textcolor{orange}{0.231} & 0.019 & 0.003 & 0.001 & 0.011 & 0.002\\
\cmidrule{1-10}\pagebreak[0]
\addlinespace[0.3em]
\multicolumn{10}{l}{\textbf{Cover}}\\
\hspace{1em} & 0.05 & 0.004 & 0 & 0 & 0 & 0.001 & 0 & 0 & 0\\
\nopagebreak
\hspace{1em} & 0.10 & 0.014 & 0.006 & 0.026 & 0 & 0.002 & 0.001 & 0.003 & 0\\
\nopagebreak
\hspace{1em}\multirow[t]{-3}{*}{\raggedright\arraybackslash IForest} & 0.20 & 0.089 & 0.03 & \textcolor{red}{0.285} & 0.083 & 0.01 & 0.005 & 0.035 & 0.013\\
\cmidrule{1-10}\pagebreak[0]
\hspace{1em} & 0.05 & 0.043 & 0.034 & \textcolor{orange}{0.058} & 0.047 & 0.942 & 0.929 & 0.961 & 0.954\\
\nopagebreak
\hspace{1em} & 0.10 & 0.086 & 0.073 & \textcolor{orange}{0.109} & 0.097 & 0.97 & 0.965 & 0.978 & 0.974\\
\nopagebreak
\hspace{1em}\multirow[t]{-3}{*}{\raggedright\arraybackslash Neighbors} & 0.20 & 0.176 & 0.158 & \textcolor{orange}{0.211} & 0.191 & 0.985 & 0.983 & 0.989 & 0.987\\
\cmidrule{1-10}\pagebreak[0]
\hspace{1em} & 0.05 & 0 & 0 & 0 & 0 & 0 & 0 & 0 & \vphantom{4} 0\\
\nopagebreak
\hspace{1em} & 0.10 & 0 & 0 & 0 & 0 & 0 & 0 & 0 & \vphantom{2} 0\\
\nopagebreak
\hspace{1em}\multirow[t]{-3}{*}{\raggedright\arraybackslash SVM} & 0.20 & 0 & 0 & 0 & 0 & 0 & 0 & 0 & \vphantom{2} 0\\
\cmidrule{1-10}\pagebreak[0]
\addlinespace[0.3em]
\multicolumn{10}{l}{\textbf{Credit card}}\\
\hspace{1em} & 0.05 & 0.027 & 0.013 & \textcolor{orange}{0.065} & 0.047 & 0.125 & 0.066 & 0.332 & 0.248\\
\nopagebreak
\hspace{1em} & 0.10 & 0.082 & 0.057 & \textcolor{orange}{0.122} & 0.099 & 0.351 & 0.249 & 0.542 & 0.482\\
\nopagebreak
\hspace{1em}\multirow[t]{-3}{*}{\raggedright\arraybackslash IForest} & 0.20 & 0.173 & 0.146 & \textcolor{orange}{0.212} & 0.186 & 0.642 & 0.57 & 0.763 & 0.712\\
\cmidrule{1-10}\pagebreak[0]
\hspace{1em} & 0.05 & 0 & 0 & 0 & 0 & 0 & 0 & 0 & \vphantom{3} 0\\
\nopagebreak
\hspace{1em} & 0.10 & 0.005 & 0.001 & 0 & 0 & 0 & 0 & 0 & 0\\
\nopagebreak
\hspace{1em}\multirow[t]{-3}{*}{\raggedright\arraybackslash Neighbors} & 0.20 & 0.072 & 0.021 & \textcolor{red}{0.292} & 0.062 & 0.006 & 0.002 & 0.024 & 0.006\\
\cmidrule{1-10}\pagebreak[0]
\hspace{1em} & 0.05 & 0 & 0 & 0 & 0 & 0 & 0 & 0 & \vphantom{2} 0\\
\nopagebreak
\hspace{1em} & 0.10 & 0 & 0 & 0 & 0 & 0 & 0 & 0 & \vphantom{1} 0\\
\nopagebreak
\hspace{1em}\multirow[t]{-3}{*}{\raggedright\arraybackslash SVM} & 0.20 & 0 & 0 & 0 & 0 & 0 & 0 & 0 & \vphantom{1} 0\\
\cmidrule{1-10}\pagebreak[0]
\addlinespace[0.3em]
\multicolumn{10}{l}{\textbf{KDDCup99}}\\
\hspace{1em} & 0.05 & 0.033 & 0.017 & \textcolor{red}{0.073} & 0.045 & 0.329 & 0.188 & 0.496 & 0.465\\
\nopagebreak
\hspace{1em} & 0.10 & 0.086 & 0.055 & \textcolor{orange}{0.118} & 0.093 & 0.562 & 0.437 & 0.695 & 0.627\\
\nopagebreak
\hspace{1em}\multirow[t]{-3}{*}{\raggedright\arraybackslash IForest} & 0.20 & 0.174 & 0.124 & \textcolor{orange}{0.209} & 0.159 & 0.736 & 0.673 & 0.796 & 0.745\\
\cmidrule{1-10}\pagebreak[0]
\hspace{1em} & 0.05 & 0.006 & 0 & 0 & 0 & 0.001 & 0 & 0 & 0\\
\nopagebreak
\hspace{1em} & 0.10 & 0.028 & 0.008 & \textcolor{red}{0.136} & 0 & 0.008 & 0.002 & 0.05 & 0\\
\nopagebreak
\hspace{1em}\multirow[t]{-3}{*}{\raggedright\arraybackslash Neighbors} & 0.20 & 0.122 & 0.041 & \textcolor{red}{0.276} & 0.17 & 0.032 & 0.011 & 0.069 & 0.055\\
\cmidrule{1-10}\pagebreak[0]
\hspace{1em} & 0.05 & 0 & 0 & 0 & 0 & 0 & 0 & 0 & \vphantom{1} 0\\
\nopagebreak
\hspace{1em} & 0.10 & 0 & 0 & 0 & 0 & 0 & 0 & 0 & 0\\
\nopagebreak
\hspace{1em}\multirow[t]{-3}{*}{\raggedright\arraybackslash SVM} & 0.20 & 0 & 0 & 0 & 0 & 0 & 0 & 0 & 0\\
\cmidrule{1-10}\pagebreak[0]
\addlinespace[0.3em]
\multicolumn{10}{l}{\textbf{Mammography}}\\
\hspace{1em} & 0.05 & 0.008 & 0 & 0.011 & 0 & 0.007 & 0 & 0.006 & 0\\
\nopagebreak
\hspace{1em} & 0.10 & 0.061 & 0.002 & \textcolor{red}{0.161} & 0 & 0.045 & 0.003 & 0.125 & 0\\
\nopagebreak
\hspace{1em}\multirow[t]{-3}{*}{\raggedright\arraybackslash IForest} & 0.20 & 0.167 & 0.054 & \textcolor{red}{0.269} & 0.169 & 0.156 & 0.058 & 0.305 & 0.217\\
\cmidrule{1-10}\pagebreak[0]
\hspace{1em} & 0.05 & 0 & 0 & 0 & 0 & 0 & 0 & 0 & 0\\
\nopagebreak
\hspace{1em} & 0.10 & 0.011 & 0 & 0.01 & 0 & 0.007 & 0 & 0.005 & 0\\
\nopagebreak
\hspace{1em}\multirow[t]{-3}{*}{\raggedright\arraybackslash Neighbors} & 0.20 & 0.126 & 0.021 & \textcolor{red}{0.241} & 0.078 & 0.14 & 0.022 & 0.272 & 0.088\\
\cmidrule{1-10}\pagebreak[0]
\hspace{1em} & 0.05 & 0.003 & 0 & 0.004 & 0 & 0.001 & 0 & 0.003 & 0\\
\nopagebreak
\hspace{1em} & 0.10 & 0.053 & 0 & \textcolor{red}{0.155} & 0 & 0.033 & 0 & 0.077 & 0\\
\nopagebreak
\hspace{1em}\multirow[t]{-3}{*}{\raggedright\arraybackslash SVM} & 0.20 & 0.171 & 0.046 & \textcolor{red}{0.264} & 0.147 & 0.145 & 0.032 & 0.257 & 0.094\\
\cmidrule{1-10}\pagebreak[0]
\addlinespace[0.3em]
\multicolumn{10}{l}{\textbf{Digits}}\\
\hspace{1em} & 0.05 & 0.01 & 0 & 0.013 & 0 & 0.01 & 0 & 0.012 & 0\\
\nopagebreak
\hspace{1em} & 0.10 & 0.051 & 0.004 & \textcolor{red}{0.155} & 0 & 0.062 & 0.006 & 0.246 & 0\\
\nopagebreak
\hspace{1em}\multirow[t]{-3}{*}{\raggedright\arraybackslash IForest} & 0.20 & 0.182 & 0.054 & \textcolor{red}{0.248} & 0.174 & 0.357 & 0.099 & 0.59 & 0.326\\
\cmidrule{1-10}\pagebreak[0]
\hspace{1em} & 0.05 & 0.005 & 0 & 0.008 & 0 & 0.027 & 0.002 & 0.025 & 0\\
\nopagebreak
\hspace{1em} & 0.10 & 0.044 & 0.005 & \textcolor{orange}{0.123} & 0.008 & 0.202 & 0.031 & 0.667 & 0.027\\
\nopagebreak
\hspace{1em}\multirow[t]{-3}{*}{\raggedright\arraybackslash Neighbors} & 0.20 & 0.172 & 0.084 & \textcolor{orange}{0.217} & 0.16 & 0.791 & 0.421 & 0.982 & 0.871\\
\cmidrule{1-10}\pagebreak[0]
\hspace{1em} & 0.05 & 0.003 & 0 & 0 & 0 & 0.001 & 0 & 0 & 0\\
\nopagebreak
\hspace{1em} & 0.10 & 0.04 & 0 & \textcolor{red}{0.151} & 0 & 0.014 & 0 & 0.041 & 0\\
\nopagebreak
\hspace{1em}\multirow[t]{-3}{*}{\raggedright\arraybackslash SVM} & 0.20 & 0.158 & 0.021 & \textcolor{red}{0.255} & 0.097 & 0.154 & 0.019 & 0.39 & 0.043\\
\cmidrule{1-10}\pagebreak[0]
\addlinespace[0.3em]
\multicolumn{10}{l}{\textbf{Shuttle}}\\
\hspace{1em} & 0.05 & 0.042 & 0.029 & \textcolor{orange}{0.06} & 0.044 & 0.926 & 0.88 & 0.975 & 0.972\\
\nopagebreak
\hspace{1em} & 0.10 & 0.086 & 0.066 & \textcolor{orange}{0.11} & 0.09 & 0.972 & 0.961 & 0.981 & 0.978\\
\nopagebreak
\hspace{1em}\multirow[t]{-3}{*}{\raggedright\arraybackslash IForest} & 0.20 & 0.178 & 0.149 & \textcolor{orange}{0.212} & 0.187 & 0.98 & 0.979 & 0.983 & 0.983\\
\cmidrule{1-10}\pagebreak[0]
\hspace{1em} & 0.05 & 0.042 & 0.029 & \textcolor{orange}{0.062} & 0.042 & 0.987 & 0.925 & 0.997 & 0.99\\
\nopagebreak
\hspace{1em} & 0.10 & 0.087 & 0.065 & \textcolor{orange}{0.111} & 0.085 & 0.998 & 0.995 & 1 & 1\\
\nopagebreak
\hspace{1em}\multirow[t]{-3}{*}{\raggedright\arraybackslash Neighbors} & 0.20 & 0.178 & 0.139 & \textcolor{orange}{0.208} & 0.174 & 1 & 0.999 & 1 & 1\\
\cmidrule{1-10}\pagebreak[0]
\hspace{1em} & 0.05 & 0.038 & 0.019 & \textcolor{orange}{0.06} & 0.045 & 0.792 & 0.462 & 0.997 & 0.993\\
\nopagebreak
\hspace{1em} & 0.10 & 0.087 & 0.064 & \textcolor{orange}{0.112} & 0.086 & 0.999 & 0.996 & 1 & 0.999\\
\nopagebreak
\hspace{1em}\multirow[t]{-3}{*}{\raggedright\arraybackslash SVM} & 0.20 & 0.178 & 0.138 & \textcolor{orange}{0.208} & 0.168 & 1 & 1 & 1 & 1\\*
\end{longtable}
\endgroup{}
}

\clearpage
{\small

\begin{longtable}[t]{lrllllllll}
\caption{Outlier batch detection performance on real data, using Storey's correction to control the FDR. Other details are as in Table~\ref{tab:data-global-long}. \label{tab:data-global-long-storey}}\\
\toprule
\multicolumn{1}{c}{ } & \multicolumn{5}{c}{FDR} & \multicolumn{4}{c}{Power} \\
\cmidrule(l{3pt}r{3pt}){2-6} \cmidrule(l{3pt}r{3pt}){7-10}
\multicolumn{2}{c}{ } & \multicolumn{2}{c}{Mean} & \multicolumn{2}{c}{90th percentile} & \multicolumn{2}{c}{Mean} & \multicolumn{2}{c}{90-th quantile} \\
\cmidrule(l{3pt}r{3pt}){3-4} \cmidrule(l{3pt}r{3pt}){5-6} \cmidrule(l{3pt}r{3pt}){7-8} \cmidrule(l{3pt}r{3pt}){9-10}
Model & Nominal & Marg. & Cond. & Marg. & Cond. & Marg. & Cond. & Marg. & Cond.\\
\midrule
\endfirsthead
\caption[]{Outlier batch detection performance on real data, using Storey's correction to control the FDR. Other details are as in Table~\ref{tab:data-global-long}.  \textit{(continued)}}\\
\toprule
\multicolumn{1}{c}{ } & \multicolumn{5}{c}{FDR} & \multicolumn{4}{c}{Power} \\
\cmidrule(l{3pt}r{3pt}){2-6} \cmidrule(l{3pt}r{3pt}){7-10}
\multicolumn{2}{c}{ } & \multicolumn{2}{c}{Mean} & \multicolumn{2}{c}{90th percentile} & \multicolumn{2}{c}{Mean} & \multicolumn{2}{c}{90-th quantile} \\
\cmidrule(l{3pt}r{3pt}){3-4} \cmidrule(l{3pt}r{3pt}){5-6} \cmidrule(l{3pt}r{3pt}){7-8} \cmidrule(l{3pt}r{3pt}){9-10}
Model & Nominal & Marg. & Cond. & Marg. & Cond. & Marg. & Cond. & Marg. & Cond.\\
\midrule
\endhead
\midrule
\multicolumn{10}{r@{}}{\textit{(Continued on Next Page...)}}\
\endfoot
\bottomrule
\endlastfoot
\addlinespace[0.3em]
\multicolumn{10}{l}{\textbf{ALOI}}\\
\hspace{1em} & 0.05 & 0.029 & 0.008 & \textcolor{red}{0.1} & 0.003 & 0 & 0 & 0.002 & 0\\
\nopagebreak
\hspace{1em} & 0.10 & 0.07 & 0.016 & \textcolor{red}{0.2} & 0.081 & 0.001 & 0 & 0.004 & 0.002\\
\nopagebreak
\hspace{1em}\multirow[t]{-3}{*}{\raggedright\arraybackslash IForest} & 0.20 & 0.157 & 0.048 & \textcolor{red}{0.332} & 0.16 & 0.004 & 0.001 & 0.009 & 0.003\\
\cmidrule{1-10}\pagebreak[0]
\hspace{1em} & 0.05 & 0.038 & 0.009 & \textcolor{red}{0.108} & 0.041 & 0.023 & 0.009 & 0.04 & 0.017\\
\nopagebreak
\hspace{1em} & 0.10 & 0.079 & 0.025 & \textcolor{red}{0.176} & 0.083 & 0.049 & 0.02 & 0.078 & 0.037\\
\nopagebreak
\hspace{1em}\multirow[t]{-3}{*}{\raggedright\arraybackslash LOF} & 0.20 & 0.167 & 0.069 & \textcolor{red}{0.288} & 0.156 & 0.109 & 0.045 & 0.154 & 0.073\\
\cmidrule{1-10}\pagebreak[0]
\hspace{1em} & 0.05 & 0.035 & 0.006 & \textcolor{red}{0.1} & 0.006 & 0.003 & 0.001 & 0.007 & 0.003\\
\nopagebreak
\hspace{1em} & 0.10 & 0.069 & 0.025 & \textcolor{red}{0.173} & 0.09 & 0.006 & 0.003 & 0.012 & 0.007\\
\nopagebreak
\hspace{1em}\multirow[t]{-3}{*}{\raggedright\arraybackslash SVM} & 0.20 & 0.157 & 0.058 & \textcolor{red}{0.33} & 0.17 & 0.014 & 0.005 & 0.023 & 0.01\\
\cmidrule{1-10}\pagebreak[0]
\addlinespace[0.3em]
\multicolumn{10}{l}{\textbf{Cover}}\\
\hspace{1em} & 0.05 & 0.037 & 0.021 & \textcolor{red}{0.099} & \textcolor{red}{0.079} & 0.1 & 0.068 & 0.185 & 0.122\\
\nopagebreak
\hspace{1em} & 0.10 & 0.08 & 0.05 & \textcolor{red}{0.158} & \textcolor{orange}{0.12} & 0.184 & 0.132 & 0.333 & 0.243\\
\nopagebreak
\hspace{1em}\multirow[t]{-3}{*}{\raggedright\arraybackslash IForest} & 0.20 & 0.176 & 0.118 & \textcolor{red}{0.277} & \textcolor{orange}{0.206} & 0.332 & 0.253 & 0.535 & 0.451\\
\cmidrule{1-10}\pagebreak[0]
\hspace{1em} & 0.05 & 0.046 & 0.029 & \textcolor{red}{0.074} & \textcolor{orange}{0.056} & 1 & 1 & 1 & 1\\
\nopagebreak
\hspace{1em} & 0.10 & 0.096 & 0.068 & \textcolor{red}{0.138} & \textcolor{orange}{0.109} & 1 & 1 & 1 & 1\\
\nopagebreak
\hspace{1em}\multirow[t]{-3}{*}{\raggedright\arraybackslash LOF} & 0.20 & 0.197 & 0.147 & \textcolor{red}{0.274} & \textcolor{orange}{0.211} & 1 & 1 & 1 & 1\\
\cmidrule{1-10}\pagebreak[0]
\hspace{1em} & 0.05 & 0 & 0 & 0 & 0 & 0 & 0 & 0 & \vphantom{2} 0\\
\nopagebreak
\hspace{1em} & 0.10 & 0 & 0 & 0 & 0 & 0 & 0 & 0 & \vphantom{2} 0\\
\nopagebreak
\hspace{1em}\multirow[t]{-3}{*}{\raggedright\arraybackslash SVM} & 0.20 & 0 & 0 & 0 & 0 & 0 & 0 & 0 & \vphantom{2} 0\\
\cmidrule{1-10}\pagebreak[0]
\addlinespace[0.3em]
\multicolumn{10}{l}{\textbf{Credit card}}\\
\hspace{1em} & 0.05 & 0.04 & 0.026 & \textcolor{orange}{0.064} & 0.049 & 0.965 & 0.951 & 0.983 & 0.972\\
\nopagebreak
\hspace{1em} & 0.10 & 0.086 & 0.059 & \textcolor{orange}{0.126} & 0.094 & 0.981 & 0.973 & 0.992 & 0.987\\
\nopagebreak
\hspace{1em}\multirow[t]{-3}{*}{\raggedright\arraybackslash IForest} & 0.20 & 0.179 & 0.131 & \textcolor{red}{0.251} & 0.184 & 0.992 & 0.988 & 0.998 & 0.996\\
\cmidrule{1-10}\pagebreak[0]
\hspace{1em} & 0.05 & 0.039 & 0.022 & \textcolor{red}{0.1} & \textcolor{red}{0.085} & 0.034 & 0.022 & 0.055 & 0.037\\
\nopagebreak
\hspace{1em} & 0.10 & 0.081 & 0.048 & \textcolor{red}{0.179} & \textcolor{orange}{0.11} & 0.062 & 0.043 & 0.097 & 0.066\\
\nopagebreak
\hspace{1em}\multirow[t]{-3}{*}{\raggedright\arraybackslash LOF} & 0.20 & 0.17 & 0.112 & \textcolor{red}{0.309} & \textcolor{orange}{0.208} & 0.122 & 0.084 & 0.178 & 0.135\\
\cmidrule{1-10}\pagebreak[0]
\hspace{1em} & 0.05 & 0 & 0 & 0 & 0 & 0 & 0 & 0 & \vphantom{1} 0\\
\nopagebreak
\hspace{1em} & 0.10 & 0 & 0 & 0 & 0 & 0 & 0 & 0 & \vphantom{1} 0\\
\nopagebreak
\hspace{1em}\multirow[t]{-3}{*}{\raggedright\arraybackslash SVM} & 0.20 & 0 & 0 & 0 & 0 & 0 & 0 & 0 & \vphantom{1} 0\\
\cmidrule{1-10}\pagebreak[0]
\addlinespace[0.3em]
\multicolumn{10}{l}{\textbf{KDDCup99}}\\
\hspace{1em} & 0.05 & 0.044 & 0.019 & \textcolor{red}{0.077} & 0.043 & 0.998 & 0.993 & 1 & 0.999\\
\nopagebreak
\hspace{1em} & 0.10 & 0.091 & 0.044 & \textcolor{red}{0.145} & 0.08 & 1 & 0.998 & 1 & 1\\
\nopagebreak
\hspace{1em}\multirow[t]{-3}{*}{\raggedright\arraybackslash IForest} & 0.20 & 0.191 & 0.099 & \textcolor{red}{0.267} & 0.166 & 1 & 1 & 1 & 1\\
\cmidrule{1-10}\pagebreak[0]
\hspace{1em} & 0.05 & 0.045 & 0.021 & \textcolor{red}{0.094} & \textcolor{red}{0.076} & 0.064 & 0.03 & 0.103 & 0.052\\
\nopagebreak
\hspace{1em} & 0.10 & 0.086 & 0.038 & \textcolor{red}{0.17} & 0.089 & 0.108 & 0.053 & 0.164 & 0.087\\
\nopagebreak
\hspace{1em}\multirow[t]{-3}{*}{\raggedright\arraybackslash LOF} & 0.20 & 0.182 & 0.074 & \textcolor{red}{0.287} & 0.165 & 0.19 & 0.096 & 0.266 & 0.146\\
\cmidrule{1-10}\pagebreak[0]
\hspace{1em} & 0.05 & 0 & 0 & 0 & 0 & 0 & 0 & 0 & 0\\
\nopagebreak
\hspace{1em} & 0.10 & 0 & 0 & 0 & 0 & 0 & 0 & 0 & 0\\
\nopagebreak
\hspace{1em}\multirow[t]{-3}{*}{\raggedright\arraybackslash SVM} & 0.20 & 0 & 0 & 0 & 0 & 0 & 0 & 0 & 0\\
\cmidrule{1-10}\pagebreak[0]
\addlinespace[0.3em]
\multicolumn{10}{l}{\textbf{Mammography}}\\
\hspace{1em} & 0.05 & 0.034 & 0.005 & \textcolor{orange}{0.066} & 0.019 & 0.476 & 0.228 & 0.628 & 0.394\\
\nopagebreak
\hspace{1em} & 0.10 & 0.069 & 0.014 & \textcolor{orange}{0.116} & 0.03 & 0.599 & 0.334 & 0.742 & 0.521\\
\nopagebreak
\hspace{1em}\multirow[t]{-3}{*}{\raggedright\arraybackslash IForest} & 0.20 & 0.138 & 0.035 & \textcolor{orange}{0.215} & 0.067 & 0.728 & 0.467 & 0.844 & 0.658\\
\cmidrule{1-10}\pagebreak[0]
\hspace{1em} & 0.05 & 0.032 & 0.01 & \textcolor{orange}{0.067} & 0.032 & 0.429 & 0.202 & 0.575 & 0.35\\
\nopagebreak
\hspace{1em} & 0.10 & 0.066 & 0.018 & \textcolor{orange}{0.114} & 0.047 & 0.561 & 0.314 & 0.691 & 0.485\\
\nopagebreak
\hspace{1em}\multirow[t]{-3}{*}{\raggedright\arraybackslash LOF} & 0.20 & 0.14 & 0.04 & \textcolor{orange}{0.201} & 0.087 & 0.697 & 0.457 & 0.814 & 0.619\\
\cmidrule{1-10}\pagebreak[0]
\hspace{1em} & 0.05 & 0.008 & 0 & 0.022 & 0 & 0.34 & 0.144 & 0.437 & 0.225\\
\nopagebreak
\hspace{1em} & 0.10 & 0.02 & 0.002 & 0.048 & 0.005 & 0.451 & 0.228 & 0.546 & 0.332\\
\nopagebreak
\hspace{1em}\multirow[t]{-3}{*}{\raggedright\arraybackslash SVM} & 0.20 & 0.043 & 0.009 & 0.083 & 0.023 & 0.57 & 0.341 & 0.655 & 0.45\\
\cmidrule{1-10}\pagebreak[0]
\addlinespace[0.3em]
\multicolumn{10}{l}{\textbf{Digits}}\\
\hspace{1em} & 0.05 & 0.04 & 0.006 & \textcolor{red}{0.074} & 0.017 & 0.924 & 0.673 & 0.986 & 0.842\\
\nopagebreak
\hspace{1em} & 0.10 & 0.084 & 0.016 & \textcolor{red}{0.141} & 0.033 & 0.968 & 0.814 & 0.995 & 0.926\\
\nopagebreak
\hspace{1em}\multirow[t]{-3}{*}{\raggedright\arraybackslash IForest} & 0.20 & 0.176 & 0.038 & \textcolor{red}{0.268} & 0.075 & 0.991 & 0.911 & 1 & 0.981\\
\cmidrule{1-10}\pagebreak[0]
\hspace{1em} & 0.05 & 0.045 & 0.007 & \textcolor{red}{0.082} & 0.019 & 0.999 & 0.977 & 1 & 1\\
\nopagebreak
\hspace{1em} & 0.10 & 0.093 & 0.019 & \textcolor{red}{0.151} & 0.038 & 1 & 0.994 & 1 & 1\\
\nopagebreak
\hspace{1em}\multirow[t]{-3}{*}{\raggedright\arraybackslash LOF} & 0.20 & 0.191 & 0.046 & \textcolor{red}{0.277} & 0.084 & 1 & 0.999 & 1 & 1\\
\cmidrule{1-10}\pagebreak[0]
\hspace{1em} & 0.05 & 0.049 & 0.007 & \textcolor{red}{0.091} & 0.021 & 0.824 & 0.511 & 0.896 & 0.665\\
\nopagebreak
\hspace{1em} & 0.10 & \textcolor{orange}{0.104} & 0.02 & \textcolor{red}{0.166} & 0.045 & 0.905 & 0.674 & 0.956 & 0.8\\
\nopagebreak
\hspace{1em}\multirow[t]{-3}{*}{\raggedright\arraybackslash SVM} & 0.20 & \textcolor{orange}{0.206} & 0.049 & \textcolor{red}{0.302} & 0.098 & 0.957 & 0.81 & 0.985 & 0.9\\
\cmidrule{1-10}\pagebreak[0]
\addlinespace[0.3em]
\multicolumn{10}{l}{\textbf{Shuttle}}\\
\hspace{1em} & 0.05 & 0.046 & 0.021 & \textcolor{red}{0.077} & 0.04 & 1 & 1 & 1 & 1\\
\nopagebreak
\hspace{1em} & 0.10 & 0.094 & 0.047 & \textcolor{red}{0.142} & 0.087 & 1 & 1 & 1 & 1\\
\nopagebreak
\hspace{1em}\multirow[t]{-3}{*}{\raggedright\arraybackslash IForest} & 0.20 & 0.194 & 0.102 & \textcolor{red}{0.281} & 0.161 & 1 & 1 & 1 & 1\\
\cmidrule{1-10}\pagebreak[0]
\hspace{1em} & 0.05 & 0.045 & 0.019 & \textcolor{red}{0.082} & 0.04 & 1 & 1 & 1 & 1\\
\nopagebreak
\hspace{1em} & 0.10 & 0.095 & 0.041 & \textcolor{red}{0.158} & 0.081 & 1 & 1 & 1 & 1\\
\nopagebreak
\hspace{1em}\multirow[t]{-3}{*}{\raggedright\arraybackslash LOF} & 0.20 & 0.193 & 0.094 & \textcolor{red}{0.287} & 0.166 & 1 & 1 & 1 & 1\\
\cmidrule{1-10}\pagebreak[0]
\hspace{1em} & 0.05 & 0.038 & 0.013 & \textcolor{orange}{0.066} & 0.026 & 1 & 1 & 1 & 1\\
\nopagebreak
\hspace{1em} & 0.10 & 0.085 & 0.034 & \textcolor{orange}{0.121} & 0.059 & 1 & 1 & 1 & 1\\
\nopagebreak
\hspace{1em}\multirow[t]{-3}{*}{\raggedright\arraybackslash SVM} & 0.20 & 0.18 & 0.083 & \textcolor{red}{0.244} & 0.137 & 1 & 1 & 1 & 1\\*
\end{longtable}
}

\clearpage
{\small

\begin{longtable}[t]{lrllllllll}
\caption{Outlier batch detection performance on real data, using different data sets, machine learning models, and nominal FDR levels. Other details are as in Table~\ref{tab:data-global-long-storey}. \label{tab:data-global-long}}\\
\toprule
\multicolumn{1}{c}{ } & \multicolumn{5}{c}{FDR} & \multicolumn{4}{c}{Power} \\
\cmidrule(l{3pt}r{3pt}){2-6} \cmidrule(l{3pt}r{3pt}){7-10}
\multicolumn{2}{c}{ } & \multicolumn{2}{c}{Mean} & \multicolumn{2}{c}{90th percentile} & \multicolumn{2}{c}{Mean} & \multicolumn{2}{c}{90-th quantile} \\
\cmidrule(l{3pt}r{3pt}){3-4} \cmidrule(l{3pt}r{3pt}){5-6} \cmidrule(l{3pt}r{3pt}){7-8} \cmidrule(l{3pt}r{3pt}){9-10}
Model & Nominal & Marg. & Cond. & Marg. & Cond. & Marg. & Cond. & Marg. & Cond.\\
\midrule
\endfirsthead
\caption[]{Outlier batch detection performance on real data, using different data sets, machine learning models, and nominal FDR levels. Other details are as in Table~\ref{tab:data-global-long-storey}.  \textit{(continued)}}\\
\toprule
\multicolumn{1}{c}{ } & \multicolumn{5}{c}{FDR} & \multicolumn{4}{c}{Power} \\
\cmidrule(l{3pt}r{3pt}){2-6} \cmidrule(l{3pt}r{3pt}){7-10}
\multicolumn{2}{c}{ } & \multicolumn{2}{c}{Mean} & \multicolumn{2}{c}{90th percentile} & \multicolumn{2}{c}{Mean} & \multicolumn{2}{c}{90-th quantile} \\
\cmidrule(l{3pt}r{3pt}){3-4} \cmidrule(l{3pt}r{3pt}){5-6} \cmidrule(l{3pt}r{3pt}){7-8} \cmidrule(l{3pt}r{3pt}){9-10}
Model & Nominal & Marg. & Cond. & Marg. & Cond. & Marg. & Cond. & Marg. & Cond.\\
\midrule
\endhead
\midrule
\multicolumn{10}{r@{}}{\textit{(Continued on Next Page...)}}\
\endfoot
\bottomrule
\endlastfoot
\addlinespace[0.3em]
\multicolumn{10}{l}{\textbf{ALOI}}\\
\hspace{1em} & 0.05 & 0.027 & 0.009 & \textcolor{red}{0.1} & 0.034 & 0 & 0 & 0.002 & 0\\
\nopagebreak
\hspace{1em} & 0.10 & 0.067 & 0.021 & \textcolor{red}{0.191} & 0.096 & 0.001 & 0 & 0.004 & 0.002\\
\nopagebreak
\hspace{1em}\multirow[t]{-3}{*}{\raggedright\arraybackslash IForest} & 0.20 & 0.157 & 0.059 & \textcolor{red}{0.341} & 0.186 & 0.004 & 0.001 & 0.008 & 0.004\\
\cmidrule{1-10}\pagebreak[0]
\hspace{1em} & 0.05 & 0.032 & 0.01 & \textcolor{red}{0.093} & \textcolor{orange}{0.054} & 0.021 & 0.009 & 0.038 & 0.018\\
\nopagebreak
\hspace{1em} & 0.10 & 0.07 & 0.025 & \textcolor{red}{0.154} & 0.084 & 0.044 & 0.02 & 0.066 & 0.033\\
\nopagebreak
\hspace{1em}\multirow[t]{-3}{*}{\raggedright\arraybackslash LOF} & 0.20 & 0.152 & 0.07 & \textcolor{red}{0.279} & 0.154 & 0.097 & 0.045 & 0.138 & 0.069\\
\cmidrule{1-10}\pagebreak[0]
\hspace{1em} & 0.05 & 0.034 & 0.006 & \textcolor{red}{0.1} & 0.001 & 0.003 & 0.001 & 0.007 & 0.004\\
\nopagebreak
\hspace{1em} & 0.10 & 0.07 & 0.027 & \textcolor{red}{0.186} & 0.096 & 0.006 & 0.003 & 0.012 & 0.007\\
\nopagebreak
\hspace{1em}\multirow[t]{-3}{*}{\raggedright\arraybackslash SVM} & 0.20 & 0.154 & 0.064 & \textcolor{red}{0.294} & 0.19 & 0.013 & 0.006 & 0.022 & 0.012\\
\cmidrule{1-10}\pagebreak[0]
\addlinespace[0.3em]
\multicolumn{10}{l}{\textbf{Cover}}\\
\hspace{1em} & 0.05 & 0.031 & 0.02 & \textcolor{red}{0.086} & \textcolor{orange}{0.071} & 0.091 & 0.065 & 0.172 & 0.113\\
\nopagebreak
\hspace{1em} & 0.10 & 0.072 & 0.046 & \textcolor{red}{0.143} & \textcolor{orange}{0.111} & 0.168 & 0.126 & 0.309 & 0.234\\
\nopagebreak
\hspace{1em}\multirow[t]{-3}{*}{\raggedright\arraybackslash IForest} & 0.20 & 0.155 & 0.109 & \textcolor{red}{0.243} & 0.189 & 0.304 & 0.24 & 0.506 & 0.427\\
\cmidrule{1-10}\pagebreak[0]
\hspace{1em} & 0.05 & 0.04 & 0.027 & \textcolor{orange}{0.067} & 0.05 & 1 & 1 & 1 & 1\\
\nopagebreak
\hspace{1em} & 0.10 & 0.086 & 0.063 & \textcolor{orange}{0.121} & 0.095 & 1 & 1 & 1 & 1\\
\nopagebreak
\hspace{1em}\multirow[t]{-3}{*}{\raggedright\arraybackslash LOF} & 0.20 & 0.176 & 0.138 & \textcolor{orange}{0.234} & 0.19 & 1 & 1 & 1 & 1\\
\cmidrule{1-10}\pagebreak[0]
\hspace{1em} & 0.05 & 0 & 0 & 0 & 0 & 0 & 0 & 0 & \vphantom{2} 0\\
\nopagebreak
\hspace{1em} & 0.10 & 0 & 0 & 0 & 0 & 0 & 0 & 0 & \vphantom{2} 0\\
\nopagebreak
\hspace{1em}\multirow[t]{-3}{*}{\raggedright\arraybackslash SVM} & 0.20 & 0 & 0 & 0 & 0 & 0 & 0 & 0 & \vphantom{2} 0\\
\cmidrule{1-10}\pagebreak[0]
\addlinespace[0.3em]
\multicolumn{10}{l}{\textbf{Credit card}}\\
\hspace{1em} & 0.05 & 0.035 & 0.025 & \textcolor{orange}{0.058} & 0.042 & 0.963 & 0.951 & 0.983 & 0.972\\
\nopagebreak
\hspace{1em} & 0.10 & 0.077 & 0.056 & \textcolor{orange}{0.116} & 0.087 & 0.98 & 0.973 & 0.992 & 0.986\\
\nopagebreak
\hspace{1em}\multirow[t]{-3}{*}{\raggedright\arraybackslash IForest} & 0.20 & 0.16 & 0.124 & \textcolor{orange}{0.223} & 0.17 & 0.992 & 0.988 & 0.998 & 0.995\\
\cmidrule{1-10}\pagebreak[0]
\hspace{1em} & 0.05 & 0.035 & 0.021 & \textcolor{red}{0.097} & \textcolor{red}{0.085} & 0.031 & 0.022 & 0.047 & 0.037\\
\nopagebreak
\hspace{1em} & 0.10 & 0.072 & 0.047 & \textcolor{red}{0.168} & \textcolor{orange}{0.107} & 0.057 & 0.04 & 0.087 & 0.062\\
\nopagebreak
\hspace{1em}\multirow[t]{-3}{*}{\raggedright\arraybackslash LOF} & 0.20 & 0.153 & 0.103 & \textcolor{red}{0.276} & 0.2 & 0.11 & 0.08 & 0.159 & 0.12\\
\cmidrule{1-10}\pagebreak[0]
\hspace{1em} & 0.05 & 0 & 0 & 0 & 0 & 0 & 0 & 0 & \vphantom{1} 0\\
\nopagebreak
\hspace{1em} & 0.10 & 0 & 0 & 0 & 0 & 0 & 0 & 0 & \vphantom{1} 0\\
\nopagebreak
\hspace{1em}\multirow[t]{-3}{*}{\raggedright\arraybackslash SVM} & 0.20 & 0 & 0 & 0 & 0 & 0 & 0 & 0 & \vphantom{1} 0\\
\cmidrule{1-10}\pagebreak[0]
\addlinespace[0.3em]
\multicolumn{10}{l}{\textbf{KDDCup99}}\\
\hspace{1em} & 0.05 & 0.039 & 0.019 & \textcolor{orange}{0.067} & 0.043 & 0.998 & 0.994 & 1 & 0.999\\
\nopagebreak
\hspace{1em} & 0.10 & 0.08 & 0.043 & \textcolor{orange}{0.126} & 0.079 & 0.999 & 0.998 & 1 & 1\\
\nopagebreak
\hspace{1em}\multirow[t]{-3}{*}{\raggedright\arraybackslash IForest} & 0.20 & 0.171 & 0.099 & \textcolor{orange}{0.238} & 0.158 & 1 & 1 & 1 & 1\\
\cmidrule{1-10}\pagebreak[0]
\hspace{1em} & 0.05 & 0.043 & 0.021 & \textcolor{red}{0.095} & \textcolor{red}{0.08} & 0.06 & 0.032 & 0.093 & 0.052\\
\nopagebreak
\hspace{1em} & 0.10 & 0.08 & 0.039 & \textcolor{red}{0.168} & 0.09 & 0.101 & 0.056 & 0.143 & 0.087\\
\nopagebreak
\hspace{1em}\multirow[t]{-3}{*}{\raggedright\arraybackslash LOF} & 0.20 & 0.165 & 0.075 & \textcolor{red}{0.267} & 0.158 & 0.176 & 0.1 & 0.246 & 0.148\\
\cmidrule{1-10}\pagebreak[0]
\hspace{1em} & 0.05 & 0 & 0 & 0 & 0 & 0 & 0 & 0 & 0\\
\nopagebreak
\hspace{1em} & 0.10 & 0 & 0 & 0 & 0 & 0 & 0 & 0 & 0\\
\nopagebreak
\hspace{1em}\multirow[t]{-3}{*}{\raggedright\arraybackslash SVM} & 0.20 & 0 & 0 & 0 & 0 & 0 & 0 & 0 & 0\\
\cmidrule{1-10}\pagebreak[0]
\addlinespace[0.3em]
\multicolumn{10}{l}{\textbf{Mammography}}\\
\hspace{1em} & 0.05 & 0.035 & 0.007 & \textcolor{orange}{0.066} & 0.022 & 0.482 & 0.259 & 0.63 & 0.433\\
\nopagebreak
\hspace{1em} & 0.10 & 0.069 & 0.018 & \textcolor{orange}{0.111} & 0.044 & 0.606 & 0.376 & 0.736 & 0.56\\
\nopagebreak
\hspace{1em}\multirow[t]{-3}{*}{\raggedright\arraybackslash IForest} & 0.20 & 0.14 & 0.045 & \textcolor{orange}{0.209} & 0.079 & 0.735 & 0.517 & 0.846 & 0.694\\
\cmidrule{1-10}\pagebreak[0]
\hspace{1em} & 0.05 & 0.032 & 0.011 & \textcolor{orange}{0.062} & 0.035 & 0.435 & 0.232 & 0.577 & 0.373\\
\nopagebreak
\hspace{1em} & 0.10 & 0.067 & 0.022 & \textcolor{orange}{0.108} & 0.05 & 0.571 & 0.354 & 0.694 & 0.52\\
\nopagebreak
\hspace{1em}\multirow[t]{-3}{*}{\raggedright\arraybackslash LOF} & 0.20 & 0.142 & 0.05 & 0.194 & 0.093 & 0.705 & 0.505 & 0.803 & 0.656\\
\cmidrule{1-10}\pagebreak[0]
\hspace{1em} & 0.05 & 0.012 & 0.001 & 0.028 & 0.002 & 0.382 & 0.189 & 0.472 & 0.272\\
\nopagebreak
\hspace{1em} & 0.10 & 0.027 & 0.004 & 0.056 & 0.013 & 0.497 & 0.29 & 0.585 & 0.39\\
\nopagebreak
\hspace{1em}\multirow[t]{-3}{*}{\raggedright\arraybackslash SVM} & 0.20 & 0.057 & 0.015 & 0.1 & 0.038 & 0.62 & 0.418 & 0.69 & 0.515\\
\cmidrule{1-10}\pagebreak[0]
\addlinespace[0.3em]
\multicolumn{10}{l}{\textbf{Digits}}\\
\hspace{1em} & 0.05 & 0.035 & 0.007 & \textcolor{orange}{0.057} & 0.018 & 0.92 & 0.719 & 0.985 & 0.876\\
\nopagebreak
\hspace{1em} & 0.10 & 0.075 & 0.019 & \textcolor{orange}{0.118} & 0.037 & 0.966 & 0.851 & 0.994 & 0.952\\
\nopagebreak
\hspace{1em}\multirow[t]{-3}{*}{\raggedright\arraybackslash IForest} & 0.20 & 0.156 & 0.047 & \textcolor{orange}{0.223} & 0.081 & 0.99 & 0.935 & 1 & 0.987\\
\cmidrule{1-10}\pagebreak[0]
\hspace{1em} & 0.05 & 0.04 & 0.008 & \textcolor{orange}{0.071} & 0.02 & 0.999 & 0.984 & 1 & 1\\
\nopagebreak
\hspace{1em} & 0.10 & 0.083 & 0.022 & \textcolor{red}{0.137} & 0.041 & 1 & 0.996 & 1 & 1\\
\nopagebreak
\hspace{1em}\multirow[t]{-3}{*}{\raggedright\arraybackslash LOF} & 0.20 & 0.169 & 0.055 & \textcolor{red}{0.242} & 0.092 & 1 & 1 & 1 & 1\\
\cmidrule{1-10}\pagebreak[0]
\hspace{1em} & 0.05 & 0.043 & 0.009 & \textcolor{red}{0.078} & 0.025 & 0.811 & 0.55 & 0.883 & 0.689\\
\nopagebreak
\hspace{1em} & 0.10 & 0.089 & 0.023 & \textcolor{red}{0.138} & 0.048 & 0.898 & 0.712 & 0.94 & 0.813\\
\nopagebreak
\hspace{1em}\multirow[t]{-3}{*}{\raggedright\arraybackslash SVM} & 0.20 & 0.179 & 0.056 & \textcolor{red}{0.251} & 0.104 & 0.953 & 0.841 & 0.979 & 0.912\\
\cmidrule{1-10}\pagebreak[0]
\addlinespace[0.3em]
\multicolumn{10}{l}{\textbf{Shuttle}}\\
\hspace{1em} & 0.05 & 0.041 & 0.02 & \textcolor{orange}{0.068} & 0.042 & 1 & 1 & 1 & 1\\
\nopagebreak
\hspace{1em} & 0.10 & 0.084 & 0.046 & \textcolor{red}{0.131} & 0.084 & 1 & 1 & 1 & 1\\
\nopagebreak
\hspace{1em}\multirow[t]{-3}{*}{\raggedright\arraybackslash IForest} & 0.20 & 0.172 & 0.103 & \textcolor{orange}{0.236} & 0.155 & 1 & 1 & 1 & 1\\
\cmidrule{1-10}\pagebreak[0]
\hspace{1em} & 0.05 & 0.041 & 0.019 & \textcolor{red}{0.073} & 0.039 & 1 & 1 & 1 & 1\\
\nopagebreak
\hspace{1em} & 0.10 & 0.084 & 0.042 & \textcolor{red}{0.132} & 0.078 & 1 & 1 & 1 & 1\\
\nopagebreak
\hspace{1em}\multirow[t]{-3}{*}{\raggedright\arraybackslash LOF} & 0.20 & 0.17 & 0.095 & \textcolor{orange}{0.237} & 0.161 & 1 & 1 & 1 & 1\\
\cmidrule{1-10}\pagebreak[0]
\hspace{1em} & 0.05 & 0.034 & 0.014 & \textcolor{orange}{0.056} & 0.027 & 1 & 1 & 1 & 1\\
\nopagebreak
\hspace{1em} & 0.10 & 0.074 & 0.035 & \textcolor{orange}{0.11} & 0.059 & 1 & 1 & 1 & 1\\
\nopagebreak
\hspace{1em}\multirow[t]{-3}{*}{\raggedright\arraybackslash SVM} & 0.20 & 0.16 & 0.085 & \textcolor{orange}{0.22} & 0.134 & 1 & 1 & 1 & 1\\*
\end{longtable}
}

\end{document}